\documentclass[runningheads,a4paper]{llncs}

\usepackage{graphicx}
\usepackage{url}

\usepackage{footmisc}
\usepackage{microtype}
\usepackage{booktabs} 
\usepackage{xspace}
\usepackage{color}
\usepackage{enumitem}
\usepackage{times}
\usepackage{appendix}
\usepackage{amsmath,amssymb,latexsym} 
\usepackage{algcompatible}
\usepackage{algorithm}
\usepackage{algpseudocode}
\algrenewcommand\alglinenumber[1]{\tiny #1:}
\usepackage{listings}
\usepackage{verbatim}
\usepackage{multicol}
\usepackage{float}
\usepackage{pifont}
\usepackage{lineno}
\usepackage{array}
\usepackage{soul}
\usepackage{mathtools}
\usepackage{hyphenat}
\usepackage{setspace}
\usepackage{tcolorbox}
\usepackage{cite}
\usepackage{wrapfig}
\usepackage{lipsum}
\usepackage{layout}
\usepackage[english]{babel}
\usepackage[utf8]{inputenc}
\usepackage{fancyhdr}
\usepackage{tabularx}
\usepackage{caption}
\usepackage{ragged2e}
\usepackage{hyperref}
\usepackage{cleveref}
\usepackage{fancyvrb}
\usepackage{epstopdf}
\usepackage{epsfig}
\usepackage{etoolbox}
\usepackage{breqn}
\usepackage{relsize}
\usepackage{pgfplots}
\usepackage{subfiles}
\usepackage{authblk}
\usepackage{pdfpages}
\pgfplotsset{compat=1.7}
\usepackage{calligra}
\usepackage{calrsfs}
\usepackage{subfig}
\usepackage{todonotes}

\newcommand{\punt}[1]{}
\newcommand{\cmnt}[1]{}

\newcounter{history}

\newcommand{\secref}[1]{Section~\ref{sec:#1}}
\newcommand{\figref}[1]{Figure~\ref{fig:#1}}

\newcommand{\thmref}[1]{Theorem~\ref{thm:#1}}
\newcommand{\lemref}[1]{Lemma~\ref{lem:#1}}

\newcommand{\defref}[1]{Definition~\ref{def:#1}}

\newcommand{\propref}[1]{Property~\ref{prop:#1}}

\newcommand{\linref}[1]{Line~\ref{lin:#1}}
\newcommand{\algoref}[1]{{Algo~\ref{algo:#1}}}

\newcommand{\Lineref}[1]{Line~\ref{lin:#1}}

\newcommand{\ignore}[1]{}

\newcommand{\tobj} {t-object\xspace}
\newcommand{\txns}[1] {txns(#1)}
\newcommand {\comm}[1] {committed(#1)}
\newcommand {\aborted}[1] {aborted(#1)}

\newcommand {\live}[1] {live(#1)}
\newcommand {\term}[1] {term(#1)}

\newcommand{\tseq} {t-sequential\xspace}

\newcommand{\lastw} {lastWrite}

\newcommand{\lupdt}[2] {#2.lastUpdt(#1)}

\newcommand{\mr} {MR}
\newcommand{\tr} {TR}

\newcommand{\validity} {validity\xspace}
\newcommand{\legal} {legal\xspace}
\newcommand{\legality} {legality\xspace}

\newcommand{\op} {operation\xspace}
\newcommand{\mth} {method\xspace}

\newcommand{\cc} {correctness-criterion\xspace}
\newcommand{\ccs} {correctness-criteria\xspace}

\newcommand{\inv} {$inv$}
\newcommand{\rsp} {$rsp$}

\newcommand{\evts}[1] {evts(#1)}
\newcommand{\met}[1] {methods(#1)}

\newcommand{\init} {\emph{init}\xspace}
\newcommand{\tbeg} {\emph{t\_begin}\xspace}
\newcommand{\tread} {\emph{t\_read}\xspace}
\newcommand{\twrite} {\emph{t\_write}\xspace}
\newcommand{\tins} {\emph{t\_insert}\xspace}
\newcommand{\tdel} {\emph{t\_delete}\xspace}
\newcommand{\tlook} {\emph{t\_lookup}\xspace}
\newcommand{\tryc} {\emph{tryC}\xspace}
\newcommand{\trya} {\emph{tryA}\xspace}

\newcommand{\dell} {\emph{list\_del}}

\newcommand{\up} {\emph{up}}

\newcommand{\opq} {opaque\xspace}
\newcommand{\opty} {opacity\xspace}

\newcommand{\tab} {hash-table\xspace}

\newcommand{\otm} {\textit{OSTM}\xspace}
\newcommand{\sotm} {\textit{SV-OSTM}\xspace}

\newcommand{\rwtm} {\textit{RWSTM}\xspace}

\newcommand{\lsl} {lazyrb\text{-}list\xspace}
\newcommand{\lazy} {lazy-list\xspace}

\newcommand{\rvm} {\emph{rvm}\xspace}

\newcommand{\fevt}[1] {#1.firstEvt}
\newcommand{\levt}[1] {#1.lastEvt}

\newcommand{\fkmth}[3] {#3.firstKeyMth(#1, #2)}
\newcommand{\pkmth}[3] {#3.prevKeyMth(#1, #2)}

\newcommand\tabspace[1][1cm]{\hspace*{#1}}

\newcommand{\txsetst}[1] {L\_txlog.setStatus($L\_txstatus \downarrow$, $ OK \downarrow$)}

\newcommand{\lsldel} {lslDel($preds[] \downarrow$, $currs[] \downarrow$)}

\newcommand{\npintv} {\emph{intraTransValidation()}}
\newcommand{\nptc} {\emph{STM tryC()}}
\newcommand{\npins} {\emph{STM insert()}}
\newcommand{\npdel} {\emph{STM delete()}}
\newcommand{\npluk} {\emph{STM lookup()}}
\newcommand{\nplsls} {\emph{list\_lookup()}}

\newcommand{\nplslins} {\emph{list\_Ins()}}

\newcommand{\npcld}{\emph{commonLu\&Del()}}

\newcommand{\rn} {\textcolor{red}{RL}\xspace}
\newcommand{\bn} {\textcolor{blue}{BL}\xspace}

\newcommand{\rc} {\textcolor{red}{currs[0]}}
\newcommand{\bc} {\textcolor{blue}{currs[1]}}
\newcommand{\bp} {\textcolor{blue}{preds[0]}}
\newcommand{\rp} {\textcolor{red}{preds[1]}}

\newcommand{\preds} {G\_preds[]}
\newcommand{\currs} {G\_currs[]}

\newcommand{\mvotm} {\textit{MVOSTM}\xspace}

\newcommand{\hmvotm} {\textit{HT-MVOSTM}\xspace}
\newcommand{\lmvotm} {\textit{list-MVOSTM}\xspace}
\newcommand{\hotm} {\textit{HT-OSTM}\xspace}
\newcommand{\kotm} {\textit{KOSTM}\xspace}
\newcommand{\hkotm} {\textit{HT-KOSTM}\xspace}
\newcommand{\lkotm} {\textit{list-KOSTM}\xspace}
\newcommand{\mvotmgc} {\textit{MVOSTM-GC}\xspace}
\newcommand{\hmvotmgc} {\textit{HT-MVOSTM-GC}\xspace}
\newcommand{\lmvotmgc} {\textit{list-MVOSTM-GC}\xspace}

\newcommand{\llog} {txLog\xspace}

\newcommand{\txlfind} {$L\_txlog.find(L\_t\_id \downarrow, L\_obj\_id \downarrow, L\_key \downarrow, L\_rec \uparrow)$}
\newcommand{\txll} {$L\_txlog$}
\newcommand{\txllf} {$L\_txlog()$}
\newcommand{\llsopn}[1] {$L\_rec.setOpn(L\_obj\_id \downarrow$, $L\_key \downarrow$, #1)}
\newcommand{\llsval}[1] {$L\_rec.setVal(L\_obj\_id \downarrow$, $ L\_key \downarrow$, #1)}
\newcommand{\llsopst}[1] {$L\_rec.setOpStatus(L\_obj\_id \downarrow$, $ L\_key \downarrow$, #1)}

\newcommand{\llgopn}[1] {$L\_rec.getOpn(L\_obj\_id \downarrow$, $ L\_key \downarrow$)}

\newcommand{\llgval}[1] {$L\_rec.getVal(L\_obj\_id \downarrow$, $ L\_key \downarrow$)}

\newcommand{\llspc} {$L\_rec.setPred\&Curr(L\_obj\_id \downarrow$, $ L\_key \downarrow$, $G\_pred \downarrow$, $G\_curr \downarrow$)}

\newcommand{\llgkeyobj} {$L\_rec.getKey\&Objid(L\_rec_{i} \downarrow$)}

\newcommand{\txgllist} {$L\_txlog.getList(L\_t\_id \downarrow$)}

\newcommand{\cld}{$commonLu\&Del$($L\_t\_id \downarrow, L\_obj\_id \downarrow, L\_key \downarrow, L\_val \uparrow, L\_op\_status \uparrow$)}

\newcommand{\cnt} {G\_cnt}

\newcommand{\gi} {$get\&inc(G\_cnt \downarrow)$}
\newcommand{\glslhead}{getListHead($L\_obj\_id \downarrow, L\_key \downarrow$)}

\newcommand{\opg}[2] {OPG(#1, #2)}
\newcommand{\ord}[1] {ord(#1)}
\newcommand{\ordfn} {ord}

\newcommand{\mv} {mv}
\newcommand{\rvf} {rvf}
\newcommand{\rt} {rt}
\newcommand{\rvmt} {rv\_method\xspace}
\newcommand{\upmt} {upd\_method\xspace}

\newcommand{\checkv} {check\_versions}

\newcommand{\instup} {ins\_tuple}
\newcommand{\gc} {gc}

\newcommand{\livel} {liveList}
\newcommand{\locko} {lockOrder}

\newcommand{\valid} {valid}
\newcommand{\seq}[2] {linearize(#2, #1)}
\newcommand{\stl}[3] {#3.stl(#1, #2)}
\newcommand{\lts}[3] {#3.lts(#1, #2)}
\newcommand{\aco} {accessOrder}

\newcommand{\lsls}[1] {list\_lookup($L\_obj\_id \downarrow,  L\_key \downarrow, G\_pred \uparrow, G\_curr \uparrow$)}
\newcommand{\find} {find\_lts}

\newcommand{\lslins}[1] {list\_Ins($G\_pred \downarrow$, $G\_curr \downarrow$, $node \uparrow$)}
\newcommand{\rlsol} {releaseOrderedLocks($L\_list \downarrow$)}

\newcommand {\cmnts}[1] {\State{\textcolor{gray}{/* #1 */} }}

\algrenewcommand{\algorithmiccomment}[1]{/* #1 */}

\graphicspath{{./figs/}}

\begin{document}

\mainmatter  

\title{\bf An Innovative Approach to Achieve Compositionality Efficiently using Multi-Version Object Based Transactional Systems \thanks{A poster version of this work received \textbf{best poster award} in NETYS-2018. An initial version of this work was accepted as \textbf{work in progress} in AADDA workshop, $ICDCN-2018$.}} 


%
%

\titlerunning{Multi-Version Object Based Transactional Systems}

\author{Chirag Juyal\inst{1}\and Sandeep Kulkarni\inst{2}\and Sweta Kumari\inst{1}\and Sathya Peri\inst{1}\and Archit Somani\inst{1}\footnote{Author sequence follows the lexical order of last names. All the authors can be contacted at the addresses given above. Archit Somani's phone number: +91 - 7095044601.}\vspace{-3mm}}
\authorrunning{C.Juyal \and S.Kulkarni \and S.Kumari \and S.Peri \and A.Somani}

\institute{Department of Computer Science \& Engineering, IIT Hyderabad, Kandi, Telangana, India \\
\texttt{(cs17mtech11014, cs15resch01004, sathya\_p, cs15resch01001)@iith.ac.in} \and Department of Computer Science, Michigan State University, MI, USA \\
	\texttt{sandeep@cse.msu.edu}}	

%
%

\maketitle

\vspace{-3mm}
\begin{abstract}
\noindent
The rise of multi-core systems has necessitated the need for concurrent programming. However, developing correct, efficient concurrent programs  is notoriously difficult. Software Transactional Memory Systems (STMs) are a convenient programming interface for a programmer to access shared memory without worrying about concurrency issues. Another advantage of STMs is that they facilitate compositionality of concurrent programs with great ease. Different concurrent operations that need to be composed to form a single atomic unit is achieved by encapsulating them in a single transaction.

Most of the STMs proposed in the literature are based on read/write primitive operations on memory buffers. We denote them as \emph{Read-Write STMs} or \emph{RWSTMs}. On the other hand, there have been some STMs that have been proposed (transactional boosting and its variants) that work on higher level operations such as hash-table insert, delete, lookup, etc. We call them Object STMs or OSTMs. 

It was observed in databases that storing multiple versions in RWSTMs provides greater concurrency. In this paper, we combine both these ideas for harnessing greater concurrency in STMs - multiple versions with objects semantics. We propose the notion of \emph{Multi-version Object STMs} or \emph{\mvotm{s}}. Specifically, we introduce and implement \mvotm for the hash-table object, denoted as \emph{\hmvotm} and list object, \emph{\lmvotm}. These objects export insert, delete and lookup \mth{s} within the transactional framework. We also show that both these \mvotm{s} satisfy \opty and ensure that transaction with lookup only \mth{s} do not abort if unbounded versions are used. 



Experimental results show that \lmvotm outperform almost two to twenty fold speedup than existing state-of-the-art list based STMs (Trans-list, Boosting-list, NOrec-list, list-MVTO, and list-OSTM).  Similarly, \hmvotm shows a significant performance gain of almost two to  nineteen times over the existing state-of-the-art hash-table based STMs  (ESTM, RWSTMs, HT-MVTO, and HT-OSTM). 


\vspace{-3mm}
\end{abstract}
\vspace{-5mm}
\section{Introduction}
\label{sec:intro}


The rise of multi-core systems has necessitated the need for concurrent programming. However, developing correct concurrent programs without compromising on efficiency is a big challenge. Software Transactional Memory Systems (STMs) are a convenient programming interface for a programmer to access shared memory without worrying about concurrency issues. Another advantage of STMs is that they facilitate compositionality of concurrent programs with great ease. Different concurrent operations that need to be composed to form a single atomic unit is achieved by encapsulating them in a single transaction. Next, we discuss different types of STMs considered in the literature and identify the need to develop multi-version object STMs proposed in this paper. 



\vspace{1mm}
\noindent
\textbf{Read-Write STMs:} Most of the STMs proposed in the literature (such as NOrec \cite{Dalessandro+:NoRec:PPoPP:2010}, ESTM \cite{Felber+:ElasticTrans:2017:jpdc}) are based on read/write operations on \emph{transaction objects} or \emph{\tobj{s}}. We denote them as \emph{Read Write STMs} or \emph{RWSTMs}. These STMs typically export following methods: (1) \tbeg{}: begins a transaction, (2) \tread{} (or $r$): reads from a \tobj, (3) \twrite{} (or $w$): writes to a \tobj, (4) \tryc{}: validates and tries to commit the transaction by writing values to the shared memory. If validation is successful, then it returns commit. Otherwise, it returns abort. 

\begin{figure}
	\centering
	\captionsetup{justification=centering}
	\centerline{\scalebox{0.4}{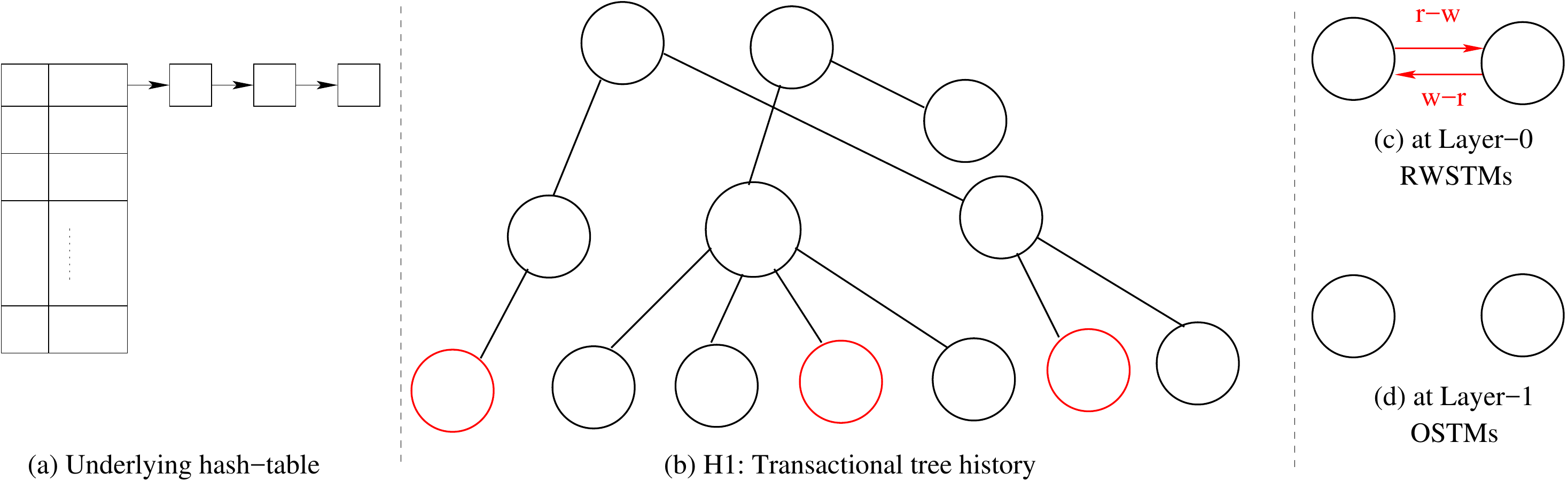}}
	\caption{Advantages of OSTMs over RWSTMs}
	\label{fig:tree-exec}
\end{figure}

 

\vspace{1mm}
\noindent
\textbf{Object STMs:} Some STMs have been proposed that work on higher level operations such as \tab. We call them \emph{Object STMs} or \emph{OSTMs}. It has been shown that \otm{s} provide greater concurrency. The concept of Boosting by Herlihy et al.\cite{HerlihyKosk:Boosting:PPoPP:2008}, the optimistic variant by Hassan et al. \cite{Hassan+:OptBoost:PPoPP:2014} and more recently \hotm system by Peri et al. \cite{Peri+:OSTM:Netys:2018} are some examples that demonstrate the performance benefits achieved by \otm{s}. 

\vspace{1mm}
\noindent 
\textbf{Benefit of \otm{s} over \rwtm{s}: } We now illustrate the advantage of \otm{s} by considering a \tab based STM system. We assume that the \op{s} of the \tab are insert (or $ins$), lookup (or $lu$) and delete (or $del$). Each \tab consists of $B$ buckets with the elements in each bucket arranged in the form of a linked-list. \figref{tree-exec}(a) represents a \tab{} with the first bucket containing keys $\langle k_2,~ k_5,~ k_7 \rangle$. \figref{tree-exec} (b) shows the execution by two transaction $T_1$ and $T_2$ represented in the form of a tree. $T_1$ performs lookup \op{s} on keys $k_2$ and $k_7$ while $T_2$ performs a delete on $k_5$. The delete on key $k_5$ generates read on the keys $k_2,k_5$ and writes the keys $k_2,k_5$ assuming that delete is performed similar to delete \op in \lazy \cite{Heller+:LazyList:PPL:2007}. The lookup on $k_2$ generates read on $k_2$ while the lookup on $k_7$ generates read on $k_2, k_7$. Note that in this execution $k_5$ has already been deleted by the time lookup on $k_7$ is performed. 


  
In this execution, we denote the read-write \op{s} (leaves) as layer-0 and $lu, del$ methods as layer-1. Consider the history (execution) at layer-0 (while ignoring higher-level \op{s}), denoted as $H0$. It can be verified this history is not \opq \cite{GuerKap:Opacity:PPoPP:2008}. This is because between the two reads of $k_2$ by $T_1$, $T_2$ writes to $k_2$. It can be seen that if history $H0$ is input to a \rwtm{s} one of the transactions between $T_1$ or $T_2$ would be aborted to ensure opacity \cite{GuerKap:Opacity:PPoPP:2008}. The \figref{tree-exec} (c) shows the presence of a cycle in the conflict graph of $H0$. 
  
Now, consider the history $H1$ at layer-1 consists of $lu$, and $del$ \mth{s}, while ignoring the read/write \op{s} since they do not overlap (referred to as pruning in \cite[Chap 6]{WeiVoss:TIS:2002:Morg}). These methods work on distinct keys ($k_2$, $k_5$, and $k_7$). They do not overlap and are not conflicting. So, they can be re-ordered in either way. Thus, $H1$ is \opq{} \cite{GuerKap:Opacity:PPoPP:2008} with equivalent serial history $T_1 T_2$ (or $T_2 T_1$) and the corresponding conflict graph shown in \figref{tree-exec} (d). Hence, a \tab based \otm{} system does not have to abort either of $T_1$ or $T_2$. This shows that \otm{s} can reduce the number of aborts and provide greater concurrency. 


\ignore{
\color{red}
\vspace{1mm}
\noindent
\textbf{Multi-Version Object STMs:} Having shown the advantage achieved by \otm{s},\todo{issue 1} We now explore the notion of \emph{Multi-Version Object STMs} or \emph{\mvotm{s}}. It was observed in databases and \rwtm{s} that by storing multiple versions for each \tobj, greater concurrency can be obtained \cite{Kumar+:MVTO:ICDCN:2014}. Maintaining multiple versions can ensure that more read operations succeed because the reading \op{} will have an appropriate version to read. This motivated us to develop \mvotm{s}. 

issue1 = We did not show this. The writeup suggests that we did this. 

\color{blue}
}

\vspace{1mm}
\noindent
\textbf{Multi-Version Object STMs:} Having seen the advantage achieved by \otm{s} (which was exploited in some works such as \cite{HerlihyKosk:Boosting:PPoPP:2008}, \cite{Hassan+:OptBoost:PPoPP:2014}, \cite{Peri+:OSTM:Netys:2018}), in this paper we propose and evaluate \emph{Multi-version Object STMs} or \emph{\mvotm{s}}. Our work is motivated by the observation that in databases and \rwtm{s} by storing multiple versions for each \tobj, greater concurrency can be obtained \cite{Kumar+:MVTO:ICDCN:2014}. Specifically, maintaining multiple versions can ensure that more read operations succeed because the reading \op{} will have an appropriate version to read. Our goal is to evaluate the benefit of \mvotm{s} over both multi-version \rwtm{s} as well as single version \otm{s}. 

\begin{figure}
	\centering
	\captionsetup{justification=centering}
	\scalebox{.4}{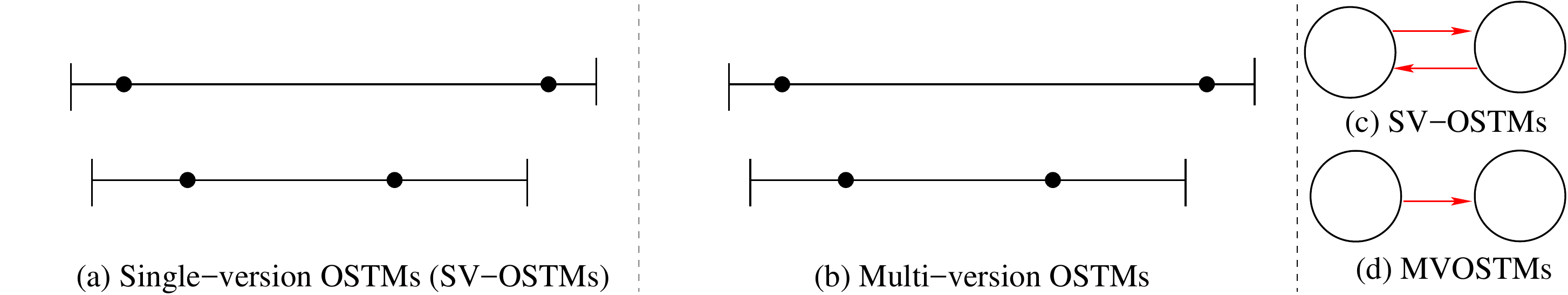}
	\caption{Advantages of multi-version over single version \otm{}}
	\label{fig:pop}
\end{figure}
\noindent 
\textbf{Potential benefit of \mvotm{s} over \otm{s} and multi-version \rwtm{s}:} We now illustrate the advantage of \mvotm{s} as compared to single-version \otm{s} (\sotm{s}) using \tab object having the same \op{s} as discussed above: $ins, lu, del$. \figref{pop} (a) represents a history H with two concurrent transactions $T_1$ and $T_2$ operating on a \tab{} $ht$. $T_1$ first tries to perform a $lu$ on key $k_2$. But due to the absence of key $k_2$ in $ht$, it obtains a value of $null$. Then $T_2$ invokes $ins$ method on the same key $k_2$ and inserts the value $v_2$ in $ht$. Then $T_2$ deletes the key $k_1$ from $ht$ and returns $v_0$ implying that some other transaction had previously inserted $v_0$ into $k_1$. The second method of $T_1$ is $lu$ on the key $k_1$. With this execution, any \sotm system has to return abort for $T_1$'s $lu$ \op  to ensure correctness, i.e., \opty. Otherwise, if $T_1$ would have obtained a return value $v_0$ for $k_1$, then the history would not be \opq anymore. This is reflected by a cycle in the corresponding conflict graph between $T_1$ and $T_2$, as shown in \figref{pop} (c). Thus to ensure \opty, \sotm system has to return abort for $T_1$'s lookup on $k_1$.

In an \mvotm based on \tab, denoted as \emph{\hmvotm}, whenever a transaction inserts or deletes a key $k$, a new version is created. Consider the above example with a \hmvotm, as shown in \figref{pop} (b). Even after $T_2$ deletes $k_1$, the previous value of $v_0$ is still retained. Thus, when $T_1$ invokes $lu$ on $k_1$ after the delete on $k_1$ by $T_2$, \hmvotm return $v_0$ (as previous value). With this, the resulting history is \opq{} with equivalent serial history being $T_1 T_2$. The corresponding conflict graph is shown in \figref{pop} (d) does not have a cycle. 

Thus, \mvotm reduces the number of aborts and achieve greater concurrency than \sotm{s} while ensuring the compositionality. We believe that the benefit of \mvotm over multi-version \rwtm is similar to \sotm over single-version \rwtm as explained above.

\mvotm is a generic concept which can be applied to any data structure. In this paper, we have considered the list and \tab{} based \mvotm{s}, \lmvotm and \hmvotm respectively. Experimental results of list-MVOSTM outperform almost two to twenty fold speedup than existing state-of-the-art STMs used to implement a list: Trans-list \cite{ZhangDech:LockFreeTW:SPAA:2016}, Boosting-list \cite{HerlihyKosk:Boosting:PPoPP:2008}, NOrec-list \cite{Dalessandro+:NoRec:PPoPP:2010} and \sotm \cite{Peri+:OSTM:Netys:2018} under high contention. Similarly, \hmvotm shows significant performance gain almost two to nineteen times better than existing state-of-the-art STMs used to implement a \tab{}: ESTM \cite{Felber+:ElasticTrans:2017:jpdc}, NOrec \cite{Dalessandro+:NoRec:PPoPP:2010} and \sotm \cite{Peri+:OSTM:Netys:2018}. To the best of our knowledge, this is the first work to explore the idea of using multiple versions in \otm{s} to achieve greater concurrency. 

\hmvotm and \lmvotm use an unbounded number of versions for each key. To address this issue, we develop two variants for both \tab and list data structures (or DS): (1) A garbage collection method in \mvotm to delete the unwanted versions of a key, denoted as \mvotmgc. Garbage collection gave a performance gain of 15\%  over \mvotm without garbage collection in the best case. Thus, the overhead of garbage collection is less than the performance improvement due to improved memory usage. (2) Placing a limit of $K$ on the number versions in \mvotm, resulting in \kotm. This gave a performance gain of 22\% over \mvotm without garbage collection in the best case. 


\noindent 
\textbf{Contributions of the paper:}
\begin{itemize}
\item We propose a new notion of multi-version objects based STM system, \mvotm. Specifically develop it for list and \tab{} objects, \lmvotm and \hmvotm respectively. 

\item We show \lmvotm and \hmvotm satisfy \emph{opacity} \cite{GuerKap:Opacity:PPoPP:2008}, standard \cc for STMs.


\item Our experiments show that both \lmvotm and \hmvotm provides greater concurrency and reduces the number of aborts as compared to \sotm{s}, single-version \rwtm{s} and, multi-version \rwtm{s}. We achieve this by maintaining multiple versions corresponding to each key.

\item For efficient space utilization in \mvotm with unbounded versions we develop \emph{Garbage Collection} for \mvotm (i.e. \mvotmgc) and bounded version \mvotm (i.e. \kotm).

\end{itemize}


\vspace{-4mm}
\section{Building System Model}
\label{sec:model}

The basic model we consider is adapted from Peri et al. \cite{Peri+:OSTM:Netys:2018}. We assume that our system consists of a finite set of $P$ processors, accessed by a finite number of $n$ threads that run in a completely asynchronous manner and communicate using shared objects. The threads communicate with each other by invoking higher-level \mth{s} on the shared objects and getting corresponding responses. Consequently, we make no assumption about the relative speeds of the threads. We also assume that none of these processors and threads fail or crash abruptly.
 
\vspace{1mm}
\noindent
\textbf{Events and Methods:} We assume that the threads execute atomic \emph{events} and the events by different threads are (1) read/write on shared/local memory objects, (2) \mth{} invocations (or \emph{\inv}) event and responses (or \emph{\rsp}) event on higher level shared-memory objects.

Within a transaction, a process can invoke layer-1 \mth{s} (or \op{s}) on a \emph{\tab} \tobj. A \tab{}($ht$) consists of multiple key-value pairs of the form $\langle k, v \rangle$. The keys and values are respectively from sets $\mathcal{K}$ and $\mathcal{V}$. The \mth{s} that a thread can invoke are: (1) $\tbeg_i{}$: begins a transaction and returns a unique id to the invoking thread. (2) $\tins_i(ht, k, v)$: transaction $T_i$ inserts a value $v$ onto key $k$ in $ht$. (3) $\tdel_i(ht, k, v)$: transaction $T_i$ deletes the key $k$ from the \tab{} $ht$ and returns the current value $v$ for $T_i$. If key $k$ does not exist, it returns $null$. (4) $\tlook_i(ht, k, v)$: returns the current value $v$ for key $k$ in $ht$  for $T_i$. Similar to \tdel, if the key $k$ does not exist then \tlook returns $null$. (5) $\tryc_i$: which tries to commit all the \op{s} of $T_i$  and (6) $\trya_i$: aborts $T_i$. We assume that each \mth{} consists of an \inv{} and \rsp{} event.


We denote \tins{} and \tdel{} as \emph{update} \mth{s} (or $\upmt{}$) since both of these change the underlying data structure. We denote \tdel{} and \tlook{} as \emph{return-value methods (or $\rvmt{}$)} as these operations return values from $ht$. A \mth{} may return $ok$ if successful or $\mathcal{A}$(abort) if it sees an inconsistent state of $ht$. 


\vspace{1mm}
\noindent
\textbf{Transactions:} Following the notations used in database multi-level transactions\cite{WeiVoss:TIS:2002:Morg}, we model a transaction as a two-level tree. The \emph{layer-0} consist of read/write events and \emph{layer-1} of the tree consists of \mth{s} invoked by a transaction.

Having informally explained a transaction, we formally define a transaction $T$ as the tuple $\langle \evts{T}, <_T\rangle$. Here $\evts{T}$ are all the read/write events at \emph{layer-0} of the transaction. $<_T$ is a total order among all the events of the transaction.

We denote the first and last events of a transaction $T_i$ as $\fevt{T_i}$ and $\levt{T_i}$. Given any other read/write event $rw$ in $T_i$, we assume that $\fevt{T_i} <_{T_i} rw <_{T_i} \levt{T_i}$. All the \mth{s} of $T_i$ are denoted as $\met{T_i}$. 

\vspace{1mm}
\noindent
\textbf{Histories:} A \emph{history} is a sequence of events belonging to different transactions. The collection of events is denoted as $\evts{H}$. Similar to a transaction, we denote a history $H$ as tuple $\langle \evts{H},<_H \rangle$ where all the events are totally ordered by $<_H$. The set of \mth{s} that are in $H$ is denoted by $\met{H}$. A \mth{} $m$ is \emph{incomplete} if $\inv(m)$ is in $\evts{H}$ but not its corresponding response event. Otherwise, $m$ is \emph{complete} in $H$. 

Coming to transactions in $H$, the set of transactions in $H$ are denoted as $\txns{H}$. The set of committed (resp., aborted) transactions in $H$ is denoted by $\comm{H}$ (resp., $\aborted{H}$). The set of \emph{live} transactions in $H$ are those which are neither committed nor aborted.  On the other hand, the set of \emph{terminated} transactions are those which have either committed or aborted. 


We denote two histories $H_1, H_2$ as \emph{equivalent} if their events are the same, i.e., $\evts{H_1} = \evts{H_2}$. A history $H$ is qualified to be \emph{well-formed} if: (1) all the \mth{s} of a transaction $T_i$ in $H$ are totally ordered, i.e. a transaction invokes a \mth{} only after it receives a response of the previous \mth{} invoked by it (2) $T_i$ does not invoke any other \mth{} after it received an $\mathcal{A}$ response or after $\tryc(ok)$ \mth. We only consider \emph{well-formed} histories for \otm.

A \mth{} $m_{ij}$ ($j^{th}$ method of a transaction $T_i$) in a history $H$ is said to be \emph{isolated} or \emph{atomic} if for any other event $e_{pqr}$ ($r^{th}$ event of method $m_{pq}$) belonging to some other \mth{} $m_{pq}$ of transaction $T_p$ either $e_{pqr}$ occurs before $\inv(m_{ij})$ or after $\rsp(m_{ij})$. 

\vspace{1mm}
\noindent
\textbf{Sequential Histories:} A history $H$ is said to be \emph{sequential} (term used in \cite{KuznetsovPeri:Non-interference:TCS:2017, KuznetsovRavi:ConcurrencyTM:OPODIS:2011}) if all the methods in it are complete and isolated. From now onwards, most of our discussion would relate to sequential histories. 

Since in sequential histories all the \mth{s} are isolated, we treat each \mth as a whole without referring to its $inv$ and $rsp$ events. For a sequential history $H$, we construct the \emph{completion} of $H$, denoted $\overline{H}$, by inserting $\trya_k(\mathcal{A})$ immediately after the last \mth{} of every transaction $T_k \in \live{H}$. Since all the \mth{s} in a sequential history are complete, this definition only has to take care of completed transactions. 


\vspace{1mm}
\noindent
\textbf{Real-time Order and Serial Histories:} Given a history $H$, $<_H$ orders all the events in $H$. For two complete \mth{s} $m_{ij}, m_{pq}$ in $\met{H}$, we denote $m_{ij} \prec_H^{\mr} m_{pq}$ if $\rsp(m_{ij}) <_H \inv(m_{pq})$. Here \mr{} stands for method real-time order. It must be noted that all the \mth{s} of the same transaction are ordered. Similarly, for two transactions $T_{i}, T_{p}$ in $\term{H}$, we denote $(T_{i} \prec_H^{\tr} T_{p})$ if $(\levt{T_{i}} <_H \fevt{T_{p}})$. Here \tr{} stands for transactional real-time order. 

\cmnt{
Thus, $\prec$ partially orders all the \mth{s} and transactions in $H$. It can be seen that if $H$ is sequential, then $\prec_H^{\mr}$ totally orders all the \mth{s} in $H$. Formally, $\langle (H \text{ is seqential}) \implies (\forall m_{ij}, m_{pq} \in \met{H}: (m_{ij} \prec_H^{\mr} m_{pq}) \lor (m_{pq} \prec_H^{\mr} m_{ij}))\rangle$. 
}

We define a history $H$ as \emph{serial} \cite{Papad:1979:JACM} or \emph{t-sequential} \cite{KuznetsovRavi:ConcurrencyTM:OPODIS:2011} if all the transactions in $H$ have terminated and can be totally ordered w.r.t $\prec_{\tr}$, i.e. all the transactions execute one after the other without any interleaving. Intuitively, a history $H$ is serial if all its transactions can be isolated. Formally, $\langle (H \text{ is serial}) \implies (\forall T_{i} \in \txns{H}: (T_i \in \term{H}) \land (\forall T_{i}, T_{p} \in \txns{H}: (T_{i} \prec_H^{\tr} T_{p}) \lor (T_{p} \prec_H^{\tr} T_{i}))\rangle$. Since all the methods within a transaction are ordered, a serial history is also sequential.

\ignore{
\vspace{1mm}
\noindent
\textbf{Real-time Order \& Serial Histories:}  Two \mth{s} $m_{ij}$ and $m_{pq}$ of history $H$ are in real-time order, if $\rsp(m_{ij}) <_H \inv(m_{pq})$. Similarly, two transactions $T_i$ and $T_j$ are in real-time order, if  $(\levt{T_{i}} <_H \fevt{T_{j}})$, where $\levt{T_{i}}$ and $\fevt{T_{j}}$ represents the last method of $T_i$ and first method of $T_j$ respectively. A history $H$ is said to be serial if all the transactions are atomic and totally ordered.

\vspace{1mm}
\noindent
\textbf{\textit{ Valid and Legal Histories:}} A history $H$ is said to valid if all the \rvmt{s} are lookup from previously committed 
}

\noindent
To simplify our analysis, we assume that there exists an initial transaction $T_0$ that invokes $\tdel$ \mth on all the keys of the \tab{} used by any transaction. 

\vspace{1mm}
\noindent
\textbf{Valid Histories:} A \rvmt{} (\tdel{} and \tlook{}) $m_{ij}$ on key $k$ is valid if it returns the value updated by any of the previous committed transaction that updated key $k$. A history $H$ is said to valid if all the \rvmt{s} of H are valid. 

\vspace{1mm}
\noindent
\textbf{Legal Histories:} A \rvmt $m_{ij}$ on key $k$ is legal if it returns the value updated the latest committed transaction that updated key $k$. A history $H$ is said to be legal, if all the \rvmt{s} of H are legal. 

We define \emph{\legality{}} of \rvmt{s} on sequential histories which we use to define correctness criterion as opacity \cite{GuerKap:Opacity:PPoPP:2008}. Consider a sequential history $H$ having a \rvmt{} $\rvm_{ij}(ht, k, v)$ (with $v \neq null$) as $j^{th}$ method belonging to transaction $T_i$. We define this \rvm \mth{} to be \emph{\legal} if: 
\vspace{-1mm}
\begin{enumerate}
	\item[LR1] \label{step:leg-same} If the $\rvm_{ij}$ is not first \mth of $T_i$ to operate on $\langle ht, k \rangle$ and $m_{ix}$ is the previous \mth of $T_i$ on $\langle ht, k \rangle$. Formally, $\rvm_{ij} \neq \fkmth{\langle ht, k \rangle}{T_i}{H}$ $\land (m_{ix}\\(ht, k, v') = \pkmth{\langle ht, k \rangle}{T_i}{H})$ (where $v'$ could be null). Then,
	\begin{enumerate}
		\setlength\itemsep{0em}
		\item If $m_{ix}(ht, k, v')$ is a \tins{} \mth then $v = v'$. 
		\item If $m_{ix}(ht, k, v')$ is a \tlook{} \mth then $v = v'$. 
		\item If $m_{ix}(ht, k, v')$ is a \tdel{} \mth then $v = null$.
	\end{enumerate}
	
	In this case, we denote $m_{ix}$ as the last update \mth{} of $\rvm_{ij}$, i.e.,  $m_{ix}(ht, k, v') = \\\lupdt{\rvm_{ij}(ht, k, v)}{H}$. 
	
	\item[LR2] \label{step:leg-ins} If $\rvm_{ij}$ is the first \mth{} of $T_i$ to operate on $\langle ht, k \rangle$ and $v$ is not null. Formally, $\rvm_{ij}(ht, k, v) = \fkmth{\langle ht, k \rangle}{T_i}{H} \land (v \neq null)$. Then,
	\begin{enumerate}
		\setlength\itemsep{0em}
		\item There is a \tins{} \mth{} $\tins_{pq}(ht, k, v)$ in $\met{H}$ such that $T_p$ committed before $\rvm_{ij}$. Formally, $\langle \exists \tins_{pq}(ht, k, v) \in \met{H} : \tryc_p \prec_{H}^{\mr} \rvm_{ij} \rangle$. 
		\item There is no other update \mth{} $up_{xy}$ of a transaction $T_x$ operating on $\langle ht, k \rangle$ in $\met{H}$ such that $T_x$ committed after $T_p$ but before $\rvm_{ij}$. Formally, $\langle \nexists up_{xy}(ht, k, v'') \in \met{H} : \tryc_p \prec_{H}^{\mr} \tryc_x \prec_{H}^{\mr} \rvm_{ij} \rangle$. 		
	\end{enumerate}
	
	In this case, we denote $\tryc_{p}$ as the last update \mth{} of $\rvm_{ij}$, i.e.,  $\tryc_{p}(ht, k, v)$= $\lupdt{\rvm_{ij}(ht, k, v)}{H}$.
	
	\item[LR3] \label{step:leg-del} If $\rvm_{ij}$ is the first \mth of $T_i$ to operate on $\langle ht, k \rangle$ and $v$ is null. Formally, $\rvm_{ij}(ht, k, v) = \fkmth{\langle ht, k \rangle}{T_i}{H} \land (v = null)$. Then,
	\begin{enumerate}
		\setlength\itemsep{0em}
		\item There is \tdel{} \mth{} $\tdel_{pq}(ht, k, v')$ in $\met{H}$ such that $T_p$ (which could be $T_0$ as well) committed before $\rvm_{ij}$. Formally, $\langle \exists \tdel_{pq}\\(ht, k,$ $ v') \in \met{H} : \tryc_p \prec_{H}^{\mr} \rvm_{ij} \rangle$. Here $v'$ could be null. 
		\item There is no other update \mth{} $up_{xy}$ of a transaction $T_x$ operating on $\langle ht, k \rangle$ in $\met{H}$ such that $T_x$ committed after $T_p$ but before $\rvm_{ij}$. Formally, $\langle \nexists up_{xy}(ht, k, v'') \in \met{H} : \tryc_p \prec_{H}^{\mr} \tryc_x \prec_{H}^{\mr} \rvm_{ij} \rangle$. 		
	\end{enumerate}
	In this case, we denote $\tryc_{p}$ as the last update \mth{} of $\rvm_{ij}$, i.e., $\tryc_{p}(ht, k, v)$ $= \lupdt{\rvm_{ij}(ht, k, v)}{H}$. 
\end{enumerate}
We assume that when a transaction $T_i$ operates on key $k$ of a \tab{} $ht$, the result of this \mth is stored in \emph{local logs} of $T_i$, $\llog_i$ for later \mth{s} to reuse. Thus, only the first \rvmt{} operating on $\langle ht, k \rangle$ of $T_i$ accesses the shared-memory. The other \rvmt{s} of $T_i$ operating on $\langle ht, k \rangle$ do not access the shared-memory and they see the effect of the previous \mth{} from the \emph{local logs}, $\llog_i$. This idea is utilized in LR1. With reference to LR2 and LR3, it is possible that $T_x$ could have aborted before $\rvm_{ij}$. For LR3, since we are assuming that transaction $T_0$ has invoked a \tdel{} \mth{} on all the keys used of the \tab{} objects, there exists at least one \tdel{} \mth{} for every \rvmt on $k$ of $ht$. We formally prove legality in \secref{cmvostm} and then we finally show that history generated by \hmvotm{} is \opq{} \cite{GuerKap:Opacity:PPoPP:2008}.


Coming to \tins{} \mth{s}, since a \tins{} \mth{} always returns $ok$ as they overwrite the node if already present therefore they always take effect on the $ht$. Thus, we denote all \tins{} \mth{s} as \legal{} and only give legality definition for \rvmt{}. We denote a sequential history $H$ as \emph{\legal} or \emph{linearized} if all its \rvm \mth{s} are \legal.


\vspace{1mm}
\noindent
\textbf{Opacity:} It is a \emph{\ccs} for STMs \cite{GuerKap:Opacity:PPoPP:2008}. A sequential history $H$ is said to be \opq{} if there exists a serial history $S$ such that: (1) $S$ is equivalent to $\overline{H}$, i.e., $\evts{\overline{H}} = \evts{S}$ (2) $S$ is \legal{} and (3) $S$ respects the transactional real-time order of $H$, i.e., $\prec_H^{\tr} \subseteq \prec_S^{\tr}$. 



\section{Graph Characterization of opacity}
\label{sec:gcofo}

To prove that a STM system satisfies opacity, it is useful to consider graph characterization of histories. In this section, we describe the graph characterization of Guerraoui and Kapalka \cite{tm-book} modified for sequential histories.

Consider a history $H$ which consists of multiple version for each \tobj. The graph characterization uses the notion of \textit{version order}. Given $H$ and a \tobj{} $k$, we define a version order for $k$ as any (non-reflexive) total order on all the versions of $k$ ever created by committed transactions in $H$. It must be noted that the version order may or may not be the same as the actual order in which the version of $k$ are generated in $H$. A version order of $H$, denoted as $\ll_H$ is the union of the version orders of all the \tobj{s} in $H$. 

Consider the history $H3$ as shown in \figref{mvostm3} $: lu_1(ht, k_{x, 0}, null), lu_2(ht, k_{x, 0}, null),\\ lu_1(ht, k_{y, 0}, null), lu_3(ht, k_{z, 0}, null), ins_1(ht, k_{x, 1}, v_{11}), ins_3(ht,
k_{y, 3}, v_{31}), ins_2(ht, \\k_{y, 2}, v_{21}), ins_1(ht, k_{z, 1}, v_{12}), c_1, c_2, lu_4(ht, k_{x, 1}, v_{11}), lu_4(ht, k_{y, 2}, v_{21}), ins_3(ht, k_{z, 3},\\ v_{32}), c_3, lu_4(ht, k_{z, 1}, v_{12}), lu_5(ht, k_{x, 1}, v_{11}), lu_6(ht, k_{y, 2}, v_{21}), c_4, c_5, c_6$. Using the notation that a committed transaction $T_i$ writing to $k_x$ creates a version $k_{x, i}$, a possible version order for $H3$ $\ll_{H3}$ is: $\langle k_{x, 0} \ll k_{x, 1} \rangle, \langle k_{y, 0} \ll k_{y, 2} \ll k_{y, 3} \rangle, \langle k_{z, 0} \ll k_{z, 1} \ll k_{z, 3} \rangle $.
\cmnt{
Consider the history $H4: l_1(ht, k_{x, 0}, NULL) l_2(ht, k_{x, 0}, NULL) l_1(ht, k_{y, 0}, NULL) l_3(ht, k_{z, 0},\\ NULL) i_1(ht, k_{x, 1}, v_{11}) i_3(ht, k_{y, 3}, v_{31}) i_2(ht, k_{y, 2}, v_{21}) i_1(ht, k_{z, 1}, v_{12}) c_1 c_2 l_4(ht, k_{x, 1}, v_{11}) l_4(ht,\\ k_{y, 2}, v_{21}) i_3(ht, k_{z, 3}, v_{32}) c_3 l_4(ht, k_{z, 1}, v_{12}) l_5(ht, k_{x, 1}, v_{11}), l_6(ht, k_{y, 2}, v_{21}) c_4, c_5, c_6$. In our representation, we abbreviate \tins{} as $i$, \tdel{} as $d$ and \tlook{} as $l$. Using the notation that a committed transaction $T_i$ writing to $k_x$ creates a version $k_{x, i}$, a possible version order for $H4$ $\ll_{H4}$ is: $\langle k_{x, 0} \ll k_{x, 1} \rangle, \langle k_{y, 0} \ll k_{y, 2} \ll k_{y, 3} \rangle, \langle k_{z, 0} \ll k_{z, 1} \ll k_{z, 3} \rangle $. 
}
\begin{figure}
	\centering
	\captionsetup{justification=centering}
	\centerline{\scalebox{0.45}{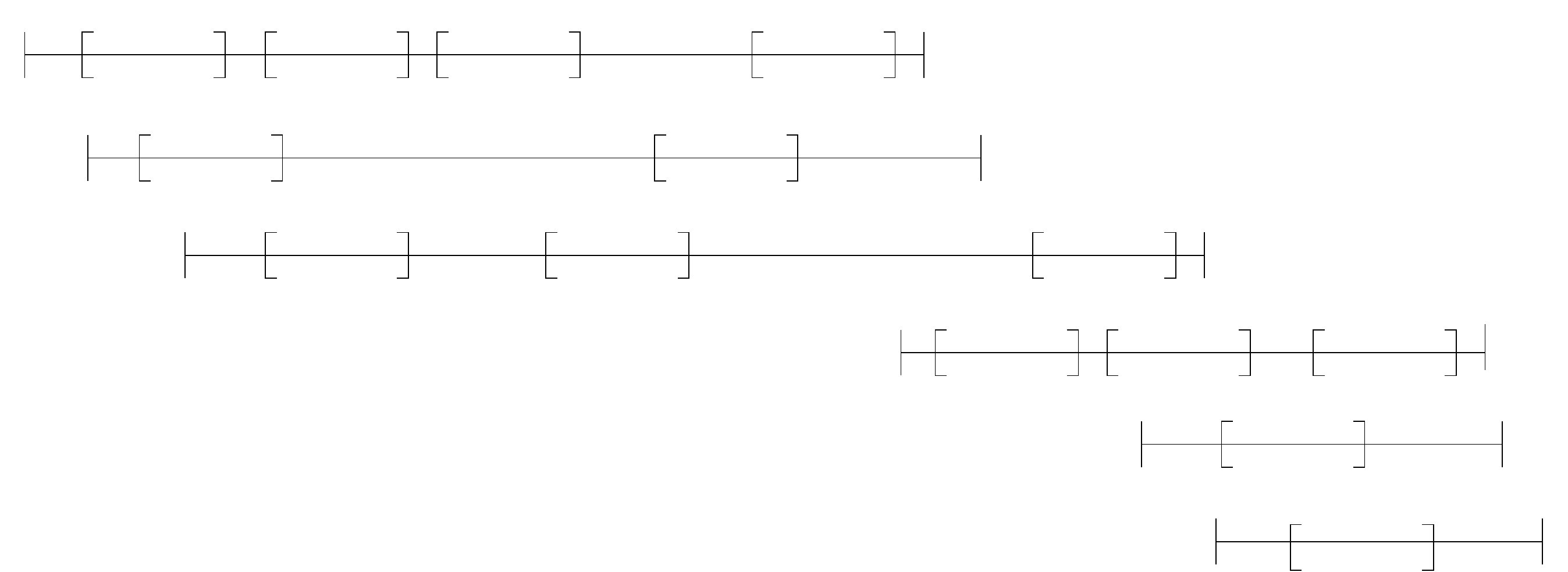}}
	\caption{History $H3$ in time line view}
	\label{fig:mvostm3}
\end{figure}
\cmnt{
\begin{figure}[tbph]
	\centerline{\scalebox{0.45}{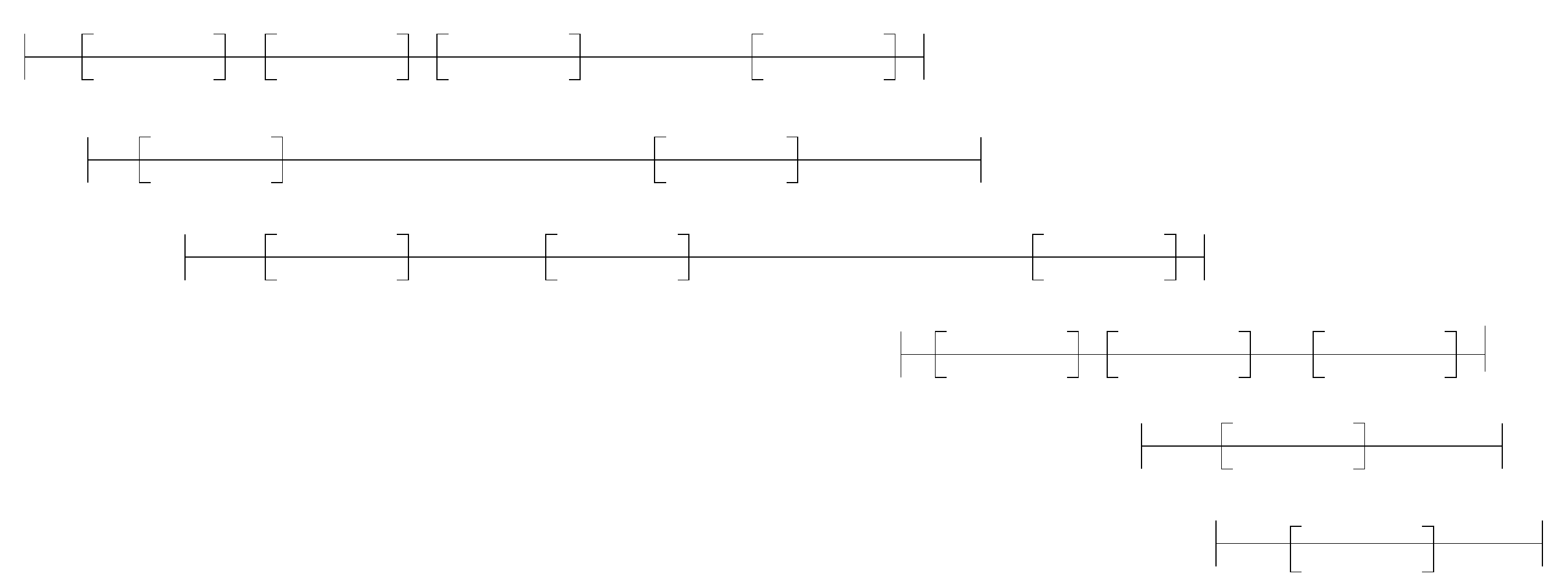}}
	\caption{History $H4$ in time line view}
	\label{fig:mvostm3}
\end{figure}
}
We define the graph characterization based on a given version order. Consider a history $H$ and a version order $\ll$. We then define a graph (called opacity graph) on $H$ using $\ll$, denoted as $\opg{H}{\ll} = (V, E)$. The vertex set $V$ consists of a vertex for each transaction $T_i$ in $\overline{H}$. The edges of the graph are of three kinds and are defined as follows:
\begin{enumerate}
\setlength\itemsep{0em}
\item \textit{\rt}(real-time) edges: If commit of $T_i$ happens before beginning of  $T_j$ in $H$, then there exist a real-time edge from $v_i$ to $v_j$. We denote set of such edges as $\rt(H)$.
\item \textit{\rvf}(return value-from) edges: If $T_j$ invokes \rvmt on key $k_1$ from $T_i$ which has already been committed in $H$, then there exist a return value-from edge from $v_i$ to $v_j$. If $T_i$ is having \upmt{} as insert on the same key $k_1$ then $ins_i(k_{1, i}, v_{i1}) <_H c_i <_H \rvm_j(k_{1, i}, v_{i1})$. If $T_i$ is having \upmt{} as delete on the same key $k_1$ then $del_i(k_{1, i}, null) <_H c_i <_H \rvm_j(k_{1, i}, null)$. We denote set of such edges as $\rvf(H)$.
\item \textit{\mv}(multi-version) edges: This is based on version order. Consider a triplet with successful methods as  $\up_i(k_{1, i},u)$, $\rvm_j(k_{1, i},u)$, $\up_k(k_{1, k},v)$ , where $u \neq v$. As we can observe it from $\rvm_j(k_{1,i},u)$, $c_i <_H\rvm_j(k_{1,i},u)$. if $k_{1,i} \ll k_{1,k}$ then there exist a multi-version edge from $v_j$ to $v_k$. Otherwise ($k_{1,k} \ll k_{1,i}$), there exist a multi-version edge from $v_k$ to $v_i$. We denote set of such edges as $\mv(H, \ll)$.
\end{enumerate}
\cmnt{
\begin{enumerate}

\item \textit{\rt}(real-time) edges: If commit of $T_i$ happens before beginning of  $T_j$ in $H$, then there exist a real-time edge from $v_i$ to $v_j$. We denote set of such edges as $\rt(H)$.

\item \textit{\rvf}(return value-from) edges: If $T_j$ invokes \rvmt on key $k_1$ from $T_i$ which has already been committed in $H$, then there exist a return value-from edge from $v_i$ to $v_j$. If $T_i$ is having \upmt{} as insert on the same key $k_1$ then $i_i(k_{1, i}, v_{i1}) <_H c_i <_H \rvm_j(k_{1, i}, v_{i1})$. If $T_i$ is having \upmt{} as delete on the same key $k_1$ then $d_i(k_{1, i}, nil_{i1}) <_H c_i <_H \rvm_j(k_{1, i}, nil_{i1})$. We denote set of such edges as $\rvf(H)$.

\item \textit{\mv}(multi-version) edges: This is based on version order. Consider a triplet with successful methods as  $\up_i(k_{1,i},u)$ $\rvm_j(k_1,u)$ $\up_k(k_{1,k},v)$ , where $u \neq v$. As we can observe it from $\rvm_j(k_1,u)$, $c_i <_H\rvm_j(k_1,u)$. if $k_{1,i} \ll k_{1,k}$ then there exist a multi-version edge from $v_j$ to $v_k$. Otherwise ($k_{1,k} \ll k_{1,i}$), there exist a multi-version edge from $v_k$ to $v_i$. We denote set of such edges as $\mv(H, \ll)$.
\vspace{-.3cm}
\end{enumerate}
}
\noindent We now show that if a version order $\ll$ exists for a history $H$ such that it is acyclic, then $H$ is \opq. 

Using this construction, the $\opg{H3}{\ll_{H3}}$ for history $H3$ and $\ll_{H3}$ is given above is shown in \figref{mvostm1}. The edges are annotated. The only \mv{} edge from $T4$ to $T3$ is because of \tobj{s} $k_y, k_z$. $T4$ lookups value $v_{12}$ for $k_z$ from $T1$ whereas $T3$ also inserts $v_{32}$ to $k_z$ and commits before $lu_4(ht, k_{z,1}, v_{12})$. 

\begin{figure}[H]
\centerline{\scalebox{0.7}{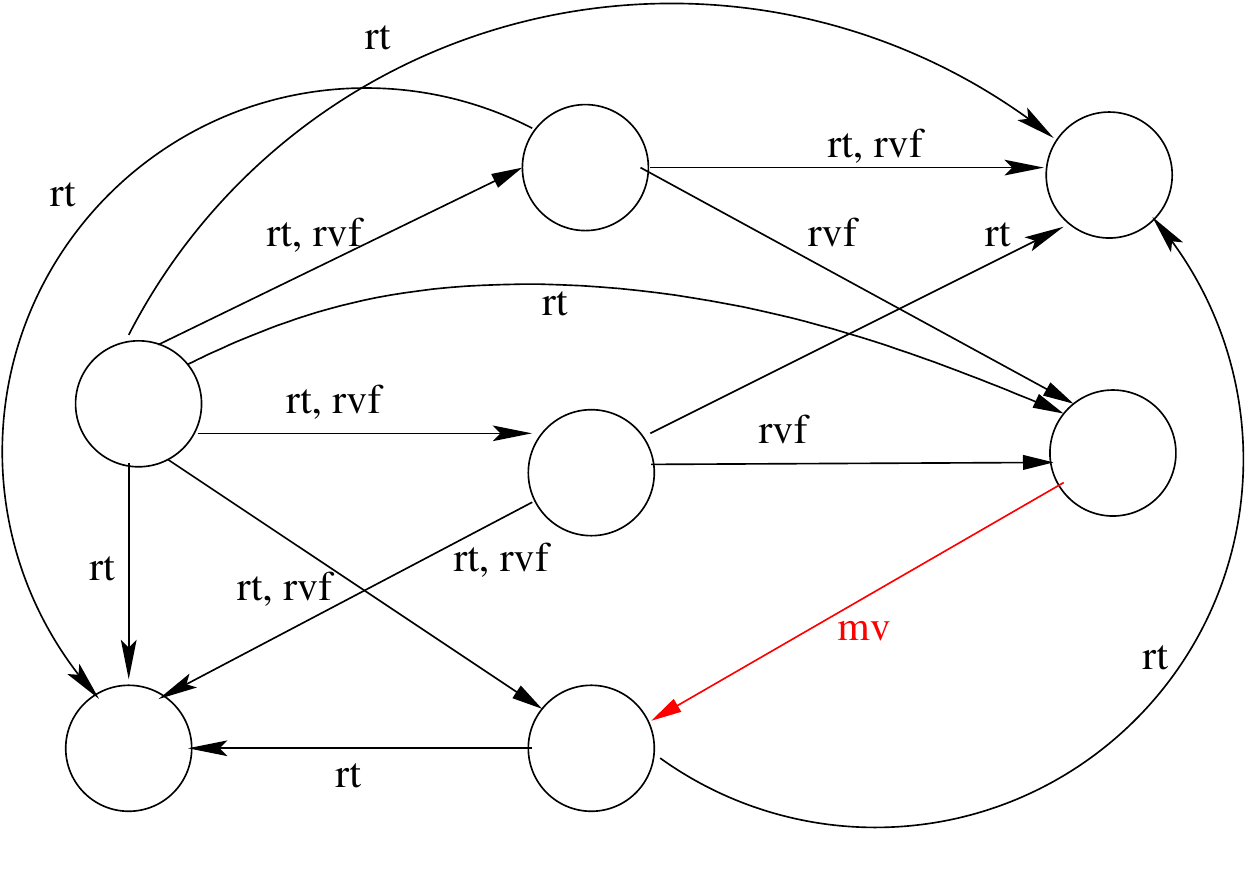}}
\caption{$\opg{H3}{\ll_{H3}}$}
\label{fig:mvostm1}
\end{figure}

Given a history $H$ and a version order $\ll$, consider the graph $\opg{\overline{H}}{\ll}$. While considering the $\rt{}$ edges in this graph, we only consider the real-time relation of $H$ and not $\overline{H}$. It can be seen that $\prec_H^{RT} \subseteq \prec_{\overline{H}}^{RT}$ but with this assumption, $\rt(H) = \rt(\overline{H})$. Hence,  we get the following property, 

\begin{property}
\label{prop:hoverh}
The graphs $\opg{H}{\ll}$ and $\opg{\overline{H}}{\ll}$ are the same for any history $H$ and $\ll$. 
\end{property}
\begin{definition}
\label{def:seqver}
For a \tseq{} history $S$, we define a version order $\ll_S$ as follows: For two version $k_{x,i}, k_{x,j}$ created by committed transactions $T_i, T_j$ in $S$, $\langle k_{x,i} \ll_S k_{x,j} \Leftrightarrow T_i <_S T_j \rangle $. 
\end{definition}
Now we show the correctness of our graph characterization using the following lemmas and theorem. 

\begin{lemma}
\label{lem:seracycle}
Consider a \legal{} \tseq{} history $S$. Then the graph $\opg{S, \ll_S}$ is acyclic.
\end{lemma}

\begin{proof}
We numerically order all the transactions in $S$ by their real-time order by using a function \textit{\ordfn}. For two transactions $T_i, T_j$, we define $\ord{T_i} < \ord{T_j} \Leftrightarrow T_i <_S T_j$. Let us analyze the edges of $\opg{S, \ll_S}$ one by one: 
\begin{itemize}
\item \rt{} edges: It can be seen that all the \rt{} edges go from a lower \ordfn{} transaction to a higher \ordfn{} transaction. 

\item \rvf{} edges: If $T_j$ lookups $k_x$ from $T_i$ in $S$ then $T_i$ is a committed transaction with $\ord{T_i} < \ord{T_j}$. Thus, all the \rvf{} edges from a lower \ordfn{} transaction to a higher \ordfn{} transaction.

\item \mv{} edges: Consider a successful \rvmt{} $\rvm_j(k_x, u)$ and a committed transaction $T_k$ writing $v$ to $k_x$ where $u \neq v$. Let $c_i$ be $\rvm_j(k_x, u)$'s \lastw. Thus, $\up_i(k_{x,i}, u) \in \evts{T_i}$. Thus, we have that $\ord{T_i} < \ord{T_j}$. Now there are two cases w.r.t $T_i$: (1) Suppose $\ord{T_k} < \ord{T_i}$. We now have that $T_k \ll T_i$. In this case, the mv edge is from $T_k$ to $T_i$. (2) Suppose $\ord{T_i} < \ord{T_k}$ which implies that $T_i \ll T_k$. Since $S$ is legal, we get that $\ord{T_j} < \ord{T_k}$. This case also implies that there is an edge from $\ord{T_j}$ to $\ord{T_k}$. Hence, in this case as well the \mv{} edges go from a transaction with lower \ordfn{} to a transaction with higher \ordfn{}. 

\end{itemize}

Thus, in all the three cases the edges go from a lower \ordfn{} transaction to higher \ordfn{} transaction. This implies that the graph is acyclic. 
\end{proof}

\begin{lemma}
\label{lem:eqv_hist_mvorder}
Consider two histories $H, H'$ that are equivalent to each other. Consider a version order $\ll_H$ on the \tobj{s} created by $H$. The mv edges $\mv(H, \ll_H)$ induced by $\ll_H$ are the same in $H$ and $H'$.
\end{lemma}

\begin{proof}
Since the histories are equivalent to each other, the version order $\ll_H$ is applicable to both of them. It can be seen that the \mv{} edges depend only on events of the history and version order $\ll$. It does not depend on the ordering of the events in $H$. Hence, the \mv{} edges of $H$ and $H'$ are equivalent to each other. 
\end{proof}

\noindent Using these lemmas, we prove the following theorem.

\begin{theorem}
\label{thm:opg}
A \valid{} history H is opaque iff there exists a version order $\ll_H$ such that $\opg{H}{\ll_H}$ is acyclic.
\end{theorem}

\begin{proof}
\textbf{(if part):} Here we have a version order $\ll_H$ such that $G_H=\opg{H}{\ll}$ is acyclic. Now we have to show that $H$ is opaque. Since the $G_H$ is acyclic, a topological sort can be obtained on all the vertices of $G_H$. Using the topological sort, we can generate a \tseq{} history $S$. It can be seen that $S$ is equivalent to $\overline{H}$. Since $S$ is obtained by a topological sort on $G_H$ which maintains the real-time edges of $H$, it can be seen that $S$ respects the \rt{} order of $H$, i.e $\prec_H^{RT} \subseteq \prec_S^{RT}$. 

Similarly, since $G_H$ maintains return value-from (\rvf{}) order of $H$, it can be seen that if $T_j$ lookups $k_x$ from $T_i$ in $H$ then $T_i$ terminates before $lu_j(k_x)$ and $T_j$ in $S$. Thus, $S$ is \valid. Now it remains to be shown that $S$ is \legal. We prove this using contradiction. Assume that $S$ is not legal. Thus, there is a successful \rvmt{} $\rvm_j(k_x, u)$ such that its \lastw{} in $S$ is $c_k$ and $T_k$ updates value $v (\neq u)$ to $k_x$, i.e $\up_k(k_{x,k}, v) \in \evts{T_k}$. Further, we also have that there is a transaction $T_i$ that insert $u$ to $k_x$, i.e $\up_i(k_{x,i}, u) \in \evts{T_i}$. Since $S$ is \valid, as shown above, we have that $T_i \prec_{S}^{RT} T_k \prec_{S}^{RT} T_j$.

Now in $\ll_H$, if $k_{x,k} \ll_H k_{x,i}$ then there is an edge from $T_k$ to $T_i$ in $G_H$. Otherwise ($k_{x,i} \ll_H k_{x,k}$), there is an edge from $T_j$ to $T_k$. Thus in either case $T_k$ can not be in between $T_i$ and $T_j$ in $S$ contradicting our assumption. This shows that $S$ is legal.



\textbf{(Only if part):} Here we are given that $H$ is opaque and we have to show that there exists a version order $\ll$ such that $G_H=\opg{H}{\ll} (=\opg{\overline{H}}{\ll}$, \propref{hoverh}) is acyclic. Since $H$ is opaque there exists a \legal{} \tseq{} history $S$ equivalent to $\overline{H}$ such that it respects real-time order of $H$. Now, we define a version order for $S$, $\ll_S$ as in \defref{seqver}. Since the $S$ is equivalent to $\overline{H}$, $\ll_S$ is applicable to $\overline{H}$ as well. From \lemref{seracycle}, we get that $G_S=\opg{S}{\ll_S}$ is acyclic. Now consider $G_H = \opg{\overline{H}}{\ll_S}$. The vertices of $G_H$ are the same as $G_S$. Coming to the edges, 

\begin{itemize}
\item \rt{} edges: We have that $S$ respects real-time order of $H$, i.e $\prec_{H}^{RT} \subseteq \prec_{S}^{RT}$. Hence, all the \rt{} edges of $H$ are a subset of $S$. 

\item \rvf{} edges: Since $\overline{H}$ and $S$ are equivalent, the return value-from relation of $\overline{H}$ and $S$ are the same. Hence, the \rvf{} edges are the same in $G_H$ and $G_S$. 

\item \mv{} edges: Since the version-order and the \op{s} of the $H$ and $S$ are the same, from \lemref{eqv_hist_mvorder} it can be seen that $\overline{H}$ and $S$ have the same \mv{} edges as well.
\end{itemize}

Thus, the graph $G_H$ is a subgraph of $G_S$. Since we already know that $G_S$ is acyclic from \lemref{seracycle}, we get that $G_H$ is also acyclic. 
\end{proof}

\section{\hmvotm Design and Data Structure}
\label{sec:mvdesign}

\hmvotm is a \tab based \mvotm that explores the idea of using multiple versions in \otm{s} for \tab object to achieve greater concurrency. The design of \hmvotm is similar to \hotm{} \cite{Peri+:OSTM:Netys:2018} consisting of $B$ buckets. All the keys of the \tab in the range $\mathcal{K}$ are statically allocated to one of these buckets. 

Each bucket consists of linked-list of nodes along with two sentinel nodes \emph{head} and \emph{tail} with values -$\infty$ and +$\infty$ respectively. The structure of each node is as $\langle key, ~ lock, ~ \\marked, ~ vl, ~ nnext \rangle$. The $key$ is a unique value from the set of all keys $\mathcal{K}$. All the nodes are stored in increasing order in each bucket as shown in \figref{1mvostmdesign} (a), similar to any linked-list based concurrent set implementation \cite{Heller+:LazyList:PPL:2007, Harris:NBList:DISC:2001}. In the rest of the document, we use the terms key and node interchangeably. To perform any operation on a key, the corresponding $lock$ is acquired. $marked$ is a boolean field which represents whether the key is deleted or not. The deletion is performed in a lazy manner similar to the concurrent linked-lists structure \cite{Heller+:LazyList:PPL:2007}. If the $marked$ field is true then key corresponding to the node has been logically deleted; otherwise, it is present. The $vl$ field of the node points to the version list (shown in \figref{1mvostmdesign} (b)) which stores multiple versions corresponding to the key. The last field of the node is $nnext$ which stores the address of the next node. It can be seen that the list of keys in a bucket is as an extension of \emph{\lazy} \cite{Heller+:LazyList:PPL:2007}. Given a node $n$ in the linked-list of bucket $B$, we denote its fields as $n.key(k.key), ~ n.lock(k.lock), ~ n.marked(k.marked), ~ n.vl(k.vl), ~ n.nnext(k.nnext)$.



\cmnt{
\begin{figure}
	\centering
	\captionsetup{justification=centering}
	\centerline{\scalebox{0.47}{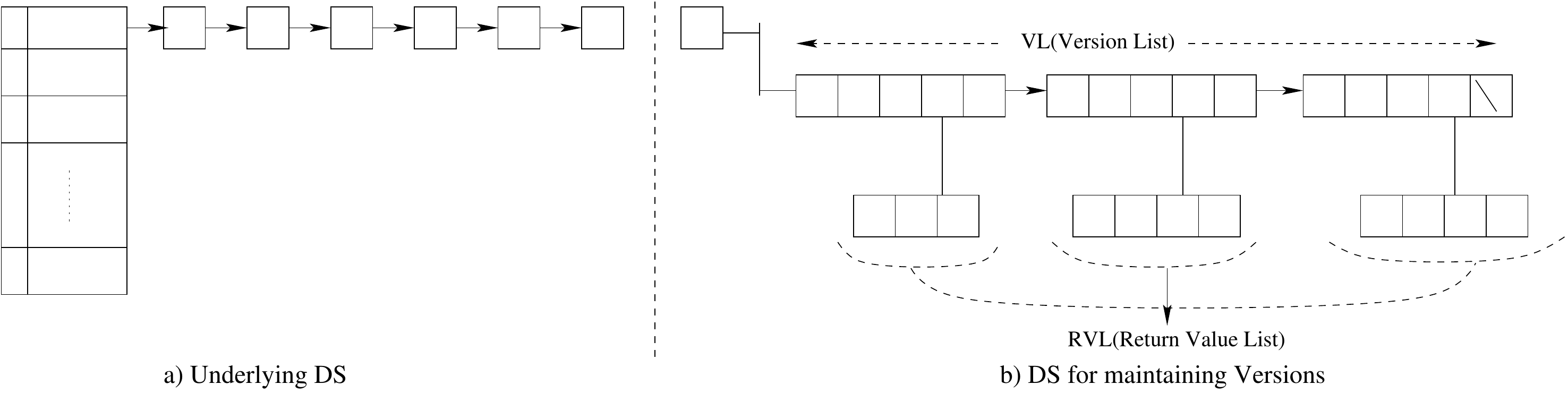}}
	\caption{\htmvotm design}
	\label{fig:1mvostmdesign1}
\end{figure}
}
\begin{figure}
	\centering
	\captionsetup{justification=centering}
	\centerline{\scalebox{0.4}{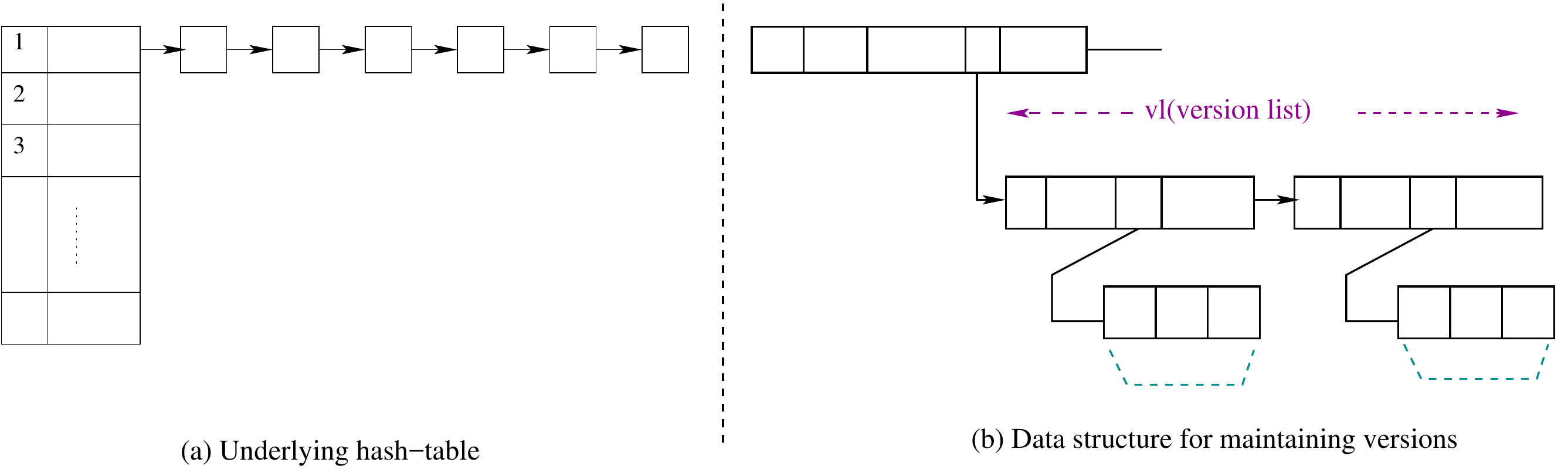}}
	\caption{\emph{HT-MVOSTM} design}
	\label{fig:1mvostmdesign}
\end{figure}

The structure of each version in the $vl$ of a key $k$ is $\langle ts, ~ val, ~ rvl, ~ vnext \rangle$ as shown in \figref{1mvostmdesign} (b). The field $ts$ denotes the unique timestamp of the version. In our algorithm, every transaction is assigned a unique timestamp when it begins which is also its $id$. Thus $ts$ of this version is the timestamp of the transaction that created it. All the versions in the $vl$ of $k$ are sorted by $ts$. Since the timestamps are unique, we denote a version, $ver$ of a node $n$ with key $k$ having $ts$ $j$ as $n.vl[j].ver$ or $k.vl[j].ver$. The corresponding fields in the version as $k.vl[j].ts, ~ k.vl[j].val, ~ k.vl[j].rvl, ~ k.vl[j].vnext$. 

The field $val$ contains the value updated by an update transaction. If this version is created by an insert \mth $\tins_i(ht, k, v)$ by transaction $T_i$, then $val$ will be $v$. On the other hand, if the \mth is $\tdel_i(ht, k)$ with the return value $v$, then $val$ will be $null$. In this case, as per the algorithm, the node of key $k$ will also be marked. \hmvotm algorithm does not immediately physically remove deleted keys from the \tab. The need for this is explained below. Thus a \rvmt (\tdel or \tlook) on key $k$ can return $null$ when it does not find the key or encounters a $null$ value for $k$. 

The $rvl$ field stands for \emph{return value list} which is a list of all the transactions that executed \rvmt{} on this version, i.e., those transactions which returned $val$. The field $vnext$ points to the next available version of that key. 

Number of versions in $vl$ (the length of the list) as per \hmvotm can be bounded or unbounded. It can be bounded by having a limit on the number of versions such as $K$. Whenever a new version $ver$ is created and is about to be added to $vl$, the length of $vl$ is checked. If the length becomes greater than $K$, the version with lowest $ts$ (i.e., the oldest) is replaced with the new version $ver$ and thus maintaining the length back to $K$. If the length is unbounded, then we need a garbage collection scheme to delete unwanted versions for efficiency. 

\ignore{

	\begin{figure}
	\captionsetup{justification=centering}
		\centerline{\scalebox{0.40}{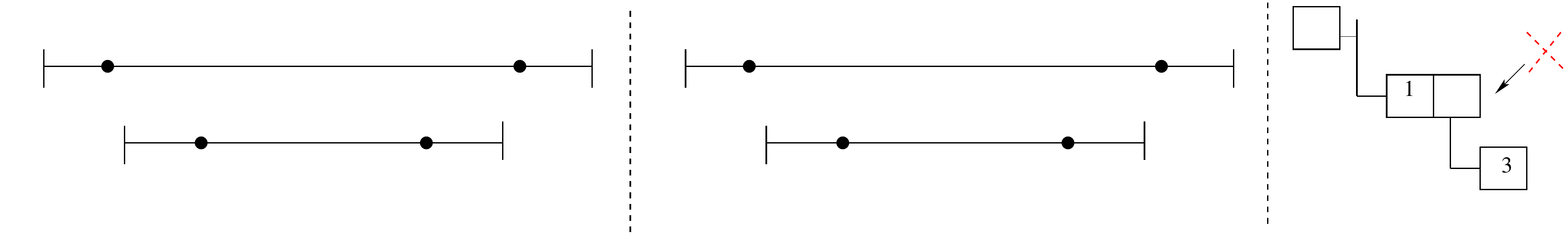}}
		\caption{Need of maintaining deleted node in underlying DS to satisfy opacity}
		\label{fig:mvostm81}
	\end{figure}	

\textbf{Why do we need to store the deleted node?} This will be clear by the \figref{mvostm81}, where we have two concurrent transactions $T_2$ and $T_3$. History in the \figref{mvostm81} (a) is not opaque because we can't come up with any serial order. To make it serial (or opaque) the second method $ins_2(ht, k_{3,2})$ of transaction $T_2$ has to return abort. For that \hmvotm scheduler have to keep the information about the already executed conflicting method $del_3(ht, k_{3,0}, NULL)$ of transaction $T_3$ using $ts$. Therefore, $del_3(ht, k_{3,0}, NULL)$ of transaction $T_3$ add itself into $T_1.rvl$ (assume that transaction $T_1$ has already inserted a version on key $k_3$ refer \figref{mvostm81} (c). So in future if any lower time-stamp transaction less than $T_3$ will come then that lower transaction will ABORT (in this case transaction $T_2$ is aborting in (\figref{mvostm81} (b)) because higher time-stamp already present in the $rvl$ (\figref{mvostm81} (c)) of the same version. After aborting $T_2$ we will get the equivalent serial history $T_2$ followed by $T_3$ 
}

\vspace{1mm}
\noindent
\textbf{Marked Nodes:} \hmvotm stores keys even after they have been deleted (nodes which have $marked$ field as true). This is because some other concurrent transactions could read from a different version of this key and not the $null$ value inserted by the deleting transaction. Consider for instance the transaction $T_1$ performing $\tlook(ht, k)$ as shown in \figref{pop} (b). Due to the presence of previous version $v_0$, \hmvotm could return this earlier version $v_0$ for $\tlook(ht, k)$ \mth. Whereas, it is not possible for \hotm to return the version $v_0$  because $k$ has been removed from the system after the delete by $T_2$. In that case, $T_1$ would have to be aborted. Thus as explained in \secref{intro}, storing multiple versions increases the concurrency. 

To store deleted keys along with live keys (or unmarked node) in a \lazy will increase the traversal time to access unmarked nodes. Consider the \figref{nostm2}, in which there are four keys $\langle k_5, k_8, k_9, k_{12}\rangle$ present in the list. Here $\langle k_5, k_8, k_9 \rangle$ are marked (or deleted) nodes while $k_{12}$ is unmarked. Now, consider an access the key $k_{12}$ as by \hmvotm as a part of one of its methods. Then \hmvotm would have to unnecessarily traverse the marked nodes to reach key $k_{12}$. 

\begin{figure}
	\centering
	\begin{minipage}[b]{0.49\textwidth}
		\centering
		\captionsetup{justification=centering}
		\scalebox{.38}{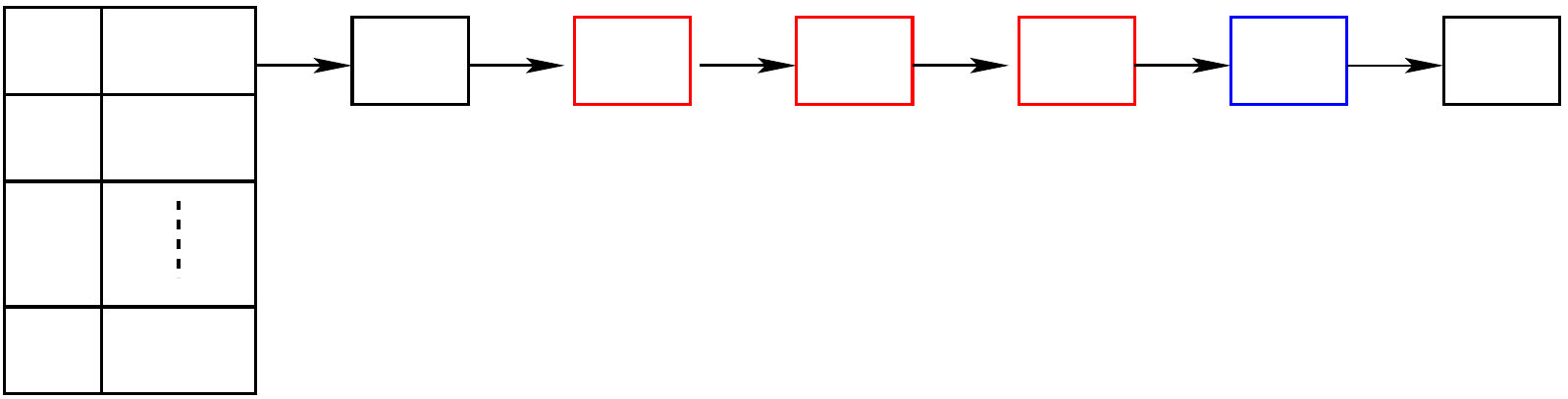}
		\caption{Searching $k_{12}$ over \emph{lazy-list}}
		\label{fig:nostm2}
	\end{minipage}   
	\hfill
	\begin{minipage}[b]{0.49\textwidth}
		\scalebox{.38}{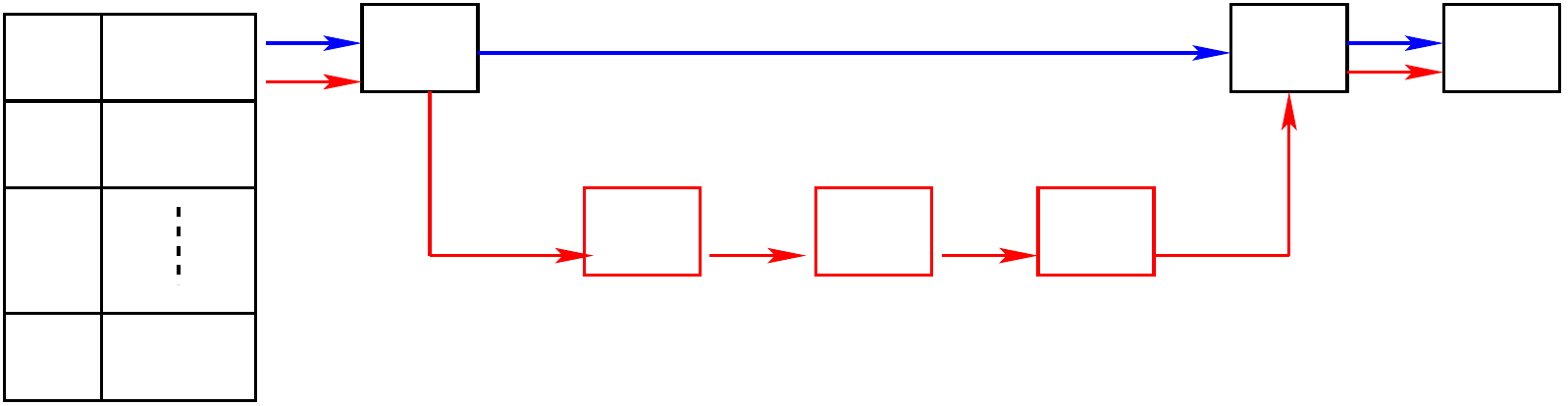}
		\centering
		\captionsetup{justification=centering}
		\caption{Searching $k_{12}$ over $\lsl{}$}
		\label{fig:nostm}
	\end{minipage}
\end{figure}
This motivated us to modify the \lazy structure of nodes in each bucket to form a skip list based on red and blue links. We denote it as \emph{red-blue lazy-list} or \emph{\lsl}. This idea was earlier explored by Peri et al. in developing \otm{s} \cite{Peri+:OSTM:Netys:2018}. $\lsl{}$ consists of nodes with two links, red link (or \rn) and blue link (or \bn). The node which are not marked (or not deleted) are accessible from the head via \bn{}. While all the nodes including the marked ones can be accessed from the head via \rn. With this modification, let us consider the above example of accessing unmarked key $k_{12}$. It can be seen that $k_{12}$ can be accessed much more quickly through \bn as shown in \figref{nostm}. Using the idea of $\lsl{}$, we have modified the structure of each node as \emph{$\langle$ key, lock, marked, vl, \rn, \bn $\rangle$}. Further, for a bucket $B$, we denote its linked-list as $B.\lsl$. 

\ignore{

Now, if node corresponding to the key doesn't exist in the underlying DS then \textbf{How we will maintain the node time-stamp by \rvmt{}?} This case will occur when node corresponding to the key is not present in \bn{} as well as \rn{}. Then \rvmt{} will create the node corresponding to the key in \rn{} with marked field as true and append the $0^{th}$ version in its $vl$ and add itself into $0^{th}.rvl$. Inserting the $0^{th}$ version ensures that the transactions which contain only \rvmt{s} will never abort. 
\begin{figure}
	\centering
	\captionsetup{justification=centering}
	\scalebox{.4}{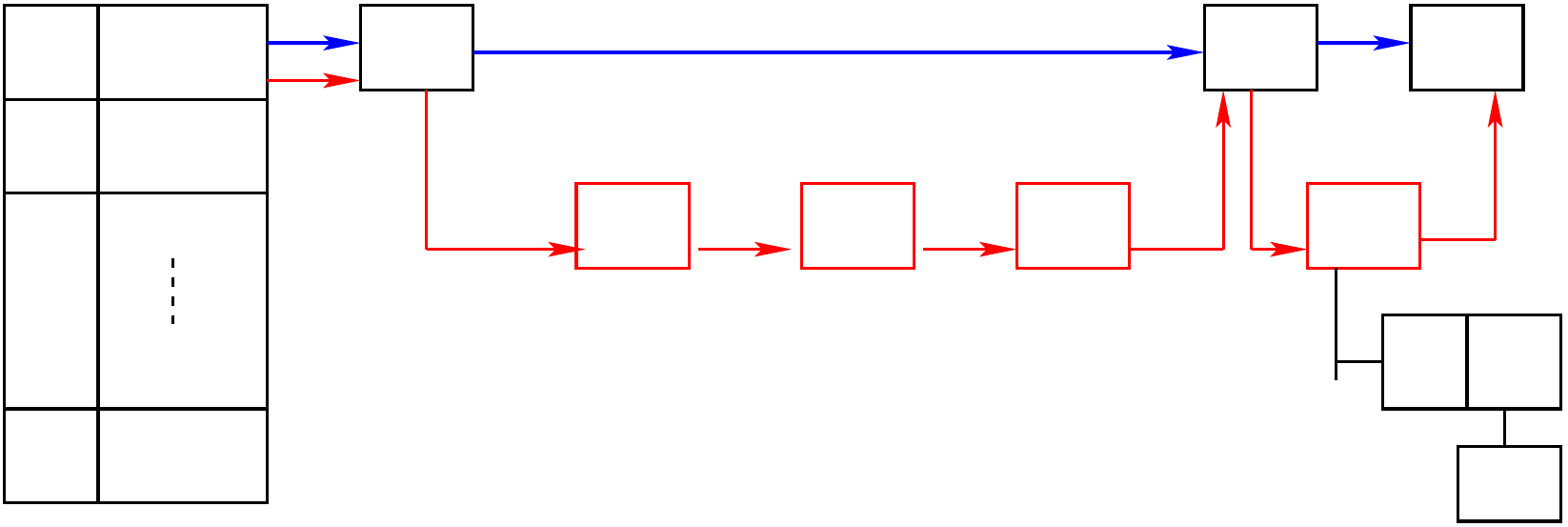}
	\caption{Execution under \lsl{}. $k_{20}$ is added in \lsl{} if not present.}
	\label{fig:nostm1}
\end{figure}

Consider the \figref{nostm}, in which \npluk{} want to search a key $k_{20}$. Here, the key $k_{20}$ is not present in the \bn{} as well as \rn{}. So it will create the node corresponding to the key in \rn{} with marked field as true and append the $0^{th}$ version in its $vl$ and add itself into $0^{th}.rvl$ refer \figref{nostm1}. 

Each transaction maintains local log in the form of $L\_txlog$ where they store unique transaction id, status and $L\_list$. If transaction is currently executing or committed or aborted due to some inconsistency then status will be live, commit or abort, respectively. $L\_list$ is a vector which contains $\langle bucket\_id, key, val, preds, currs, op\_status, opn \rangle$. The $op\_status$ is an operation status which could be $OK$ or $FAIL$.
Each \mth{} identifies the location corresponding to the key in \bn{} and \rn{} from \lsl{} where they have to work and store the $preds$ and $currs$ in the form of an array. $opn$ stands for operation name which could be \npluk, \npins{}, \npdel{} which will work on $key$. 

}

\section{Working of \hmvotm{}} 
\label{sec:pcode}


As explained in \secref{model}, \hmvotm exports \tbeg, \tins, \tdel, \tlook, \tryc \mth{s}. \tdel, \tlook are \rvmt{s} while \tins, \tdel are \upmt{s}. We treat \tdel as both \rvmt as well as \upmt. The \rvmt{s} return the current value of the key. 
The \upmt{s}, update to the keys are first noted down in local log, \emph{\llog}. Then in the \tryc \mth after  validations of these updates are transferred to the shared memory. 
\ignore{

\begin{figure}
	\centering	
	\captionsetup{justification=centering}
	\centerline{\scalebox{0.5}{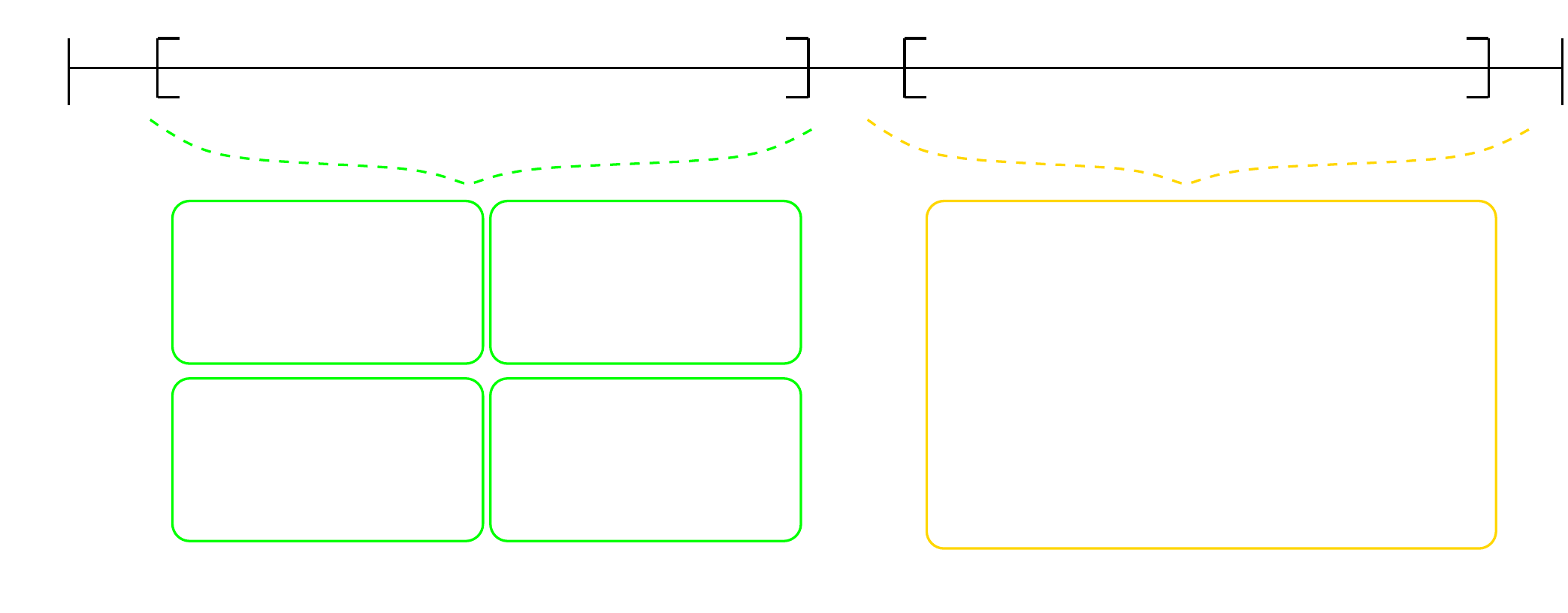}}
	\caption{Transaction life cycle of \mvotm}	
	\label{fig:nostm24}
\end{figure}
}
We now explain the working of \rvmt and \upmt . 


\vspace{1mm}
\noindent 
\textbf{$\tbeg():$} A thread invokes a new transaction $T_i$ using this method. This \mth returns a unique id to the invoking thread by incrementing an atomic counter. This unique id is also the timestamp of the transaction $T_i$. For convenience, we use the notation that $i$ is the timestamp (or id) of the transaction $T_i$. The transaction $T_i$ local log $\llog_i$ is initialized in this \mth. 

\vspace{1mm}
\noindent 
\textbf{\rvmt{s}} - $\tdel_i(ht, k, v)$ and $\tlook_i(ht, k, v):$ Both these \mth{s} return the current value of key $k$. \algoref{rvmt} gives the high-level overview of these \mth{s}. First, the algorithm checks to see if the given key is already in the local log, $\llog$ of $T_i$ (\linref{rvm-chk_log}). If the key is already there then the current \rvmt is not the first method on $k$ and is a subsequent method of $T_i$ on $k$. So, we can return the value of $k$ from the $\llog_i$. 

If the key is not present in the $\llog_i$, then \hmvotm searches into shared memory. Specifically, it searches the bucket to which $k$ belongs to. Every key in the range $\mathcal{K}$ is statically allocated to one of the $B$ buckets. So the algorithms search for $k$ in the corresponding bucket, say $B_k$ to identify the appropriate location, i.e., identify the correct \emph{predecessor} or $pred$ and \emph{current} or  $curr$ keys in the \lsl of $B_k$ without acquiring any locks similar to the search in \lazy \cite{Heller+:LazyList:PPL:2007}. Since each key has two links, \rn and \bn, the algorithm identifies four node references: two $pred$ and two $curr$ according to red and blue links. They are stored in the form of an array with $\bp$ and $\bc$ corresponding to blue links; $\rp$ and $\rc$ corresponding to red links. If both $\rp$ and $\rc$ nodes are unmarked then the $pred, curr$ nodes of both red and blue links will be the same, i.e., $\bp = \rp$ and $\rc = \bc$. Thus depending on the marking of $pred, curr$ nodes, a total of two, three or four different nodes will be identified. Here, the search ensures that $\bp.key \leq \rp.key < k \leq \rc.key \leq \bc.key$. 

Next, the re-entrant locks on all the $pred, curr$ keys are acquired in increasing order to avoid the deadlock. Then all the $pred$ and $curr$ keys are validated by \emph{rv\_Validation()} in \linref{rvm-chk_valid} as follows: (1) If $pred$ and $curr$ nodes of blue links are not marked, i.e, $(\neg \bp.marked) ~ \&\& ~ (\neg \bc.marked)$. (2) If the next links of both blue and red $pred$ nodes point to the correct $curr$ nodes: $(\bp.\bn = \bc) ~ \&\& ~ (\rp.\\\rn = \rc)$. 

If any of these checks fail, then the algorithm retries to find the correct $pred$ and $curr$ keys. It can be seen that the validation check is similar to the validation in concurrent \lazy \cite{Heller+:LazyList:PPL:2007}.

Next, we check if $k$ is in $B_k.\lsl$. If $k$ is not in $B_k$, then we create a new node for $k$ as: $\langle key=k, lock=false, marked=false, vl=v, nnext=\phi \rangle$ and insert it into $B_k.\lsl$ such that it is accessible only via \rn since this node is marked (\linref{rvm-ver_abs}). This node will have a single version $v$ as: $\langle ts=0, val=null, rvl=i, vnext=\phi \rangle$. Here invoking transaction $T_i$ is creating a version with timestamp $0$ to ensure that \rvmt{s} of other transactions will never abort. As we have explained in \figref{pop} (b) of \secref{intro}, even after $T_2$ deletes $k_1$, the previous value of $v_0$ is still retained. Thus, when $T_1$ invokes $lu$ on $k_1$ after the delete on $k_1$ by $T_2$, \hmvotm will return $v_0$ (as previous value). Hence, each \rvmt{s} will find a version to read while maintaining the infinite version corresponding to each key $k$. In $rvl$, $T_i$ adds the timestamp as $i$ in it and $vnext$ is initialized to empty value. Since $val$ is null and the $n$, this version and the node is not technically inserted into $B_k.\lsl$. 

If $k$ is in $B_k.\lsl$ then, $k$ is the same as $\rc$ or $\bc$ or both. Let $n$ be the node of $k$ in $B_k.\lsl$. We then find the version of $n$, $ver_j$ which has the timestamp $j$ such that $j$ has the largest timestamp smaller than $i$ (timestamp of $T_i$). Add $i$ to $ver_j$'s $rvl$ (\linref{rvm-add_i}). Then release the locks, update the local log $\llog_i$ in \linref{rvm-unlock} and return the value stored in $ver_j.val$ in \linref{rvm-ret}). 
\begin{algorithm}
	\scriptsize
	\caption{\emph{\rvmt:} Could be either $\tdel_i(ht, k, v)$ or $\tlook_i(ht, k, v)$ on key $k$ that maps to bucket $B_k$.} \label{algo:rvmt} 	
	\begin{algorithmic}[1]
		\Procedure{$\rvmt_i$}{$ht, k, v$}		
		\If{($k \in \llog_i$)} \label{lin:rvm-chk_log}
			\State Update the local log and return $val$. \label{lin:rvm-pres_log}
		\Else \label{lin:rvm-els_log}
			\State Search in \lsl to identify the $preds[]$ and $currs[]$ for $k$ using \bn and \rn in bucket $B_k$. \label{lin:rvm-search}
			\State Acquire the locks on $preds[]$ and $currs[]$ in increasing order. \label{lin:rvm-locks}
			\If{($!rv\_Validation(preds[], currs[])$)} \label{lin:rvm-chk_valid}
				\State Release the locks and goto \linref{rvm-search}. \label{lin:rvm-retry}
			\EndIf 
			\If{($k  ~ \notin ~ B_k.\lsl$)} \label{lin:rvm-chk_k} 
				\State Create a new node $n$ with key $k$ as: $\langle$ \emph{key = k, lock = false, marked = false, vl = v, nnext = }$\phi \rangle$. \label{lin:rvm-chk_k1} 
				\State \Comment{The $vl$ consists of a single element $v$ with $ts$ as $i$}
				\State Create the version $v$ as: $\langle ts=0, val=null, rvl=i, vnext=\phi \rangle$.  
				\State Insert $n$ into $B_k.\lsl$ such that it is accessible only via \rn{s}. \Comment{$n$ is marked} \label{lin:rvm-ver_abs}  	\State Release the locks; update the $\llog_i$ with $k$.
				\State return $null$.
			\EndIf 
			\State Identify the version $ver_j$ with $ts=j$ such that $j$ is the largest timestamp smaller than $i$.
			\If{($ver_j$ == $null$)}
			\State goto \Lineref{rvm-chk_k1}.
			\EndIf
			\State Add $i$ into the $rvl$ of $ver_j$. \label{lin:rvm-add_i} 
			\State $retVal = ver_j.val$.
			\State Release the locks; update the $\llog_i$ with $k$ and $retVal$. \label{lin:rvm-unlock} 
		\EndIf 
		\State return $retVal$. \label{lin:rvm-ret}
		\EndProcedure
	\end{algorithmic}
\end{algorithm}

\ignore {
The \emph{rv\_Validation()} is called by the \rvmt{} and \upmt{}. It will identify the conflicts among the concurrent methods of different transactions. It will be more clear by the \figref{mvostm9}, where two concurrent conflicting methods of different transactions are working on the same key $k_3$. Initially, at stage $s_1$ in \figref{mvostm9} (c) both the conflicting methods $ins_1(k_3)$ and $lu_2(k_3)$ optimistically (without acquiring locks) identify the same $preds$ and $currs$ for key $k_3$ from underlying $CDS$ in \figref{mvostm9} (a). At stage $s_2$ in \figref{mvostm9} (c), method $ins_1(k_3)$ of transaction $T_1$ acquired the lock on $preds$ and $currs$ and inserted the node into it (\figref{mvostm9} (b)). After successful insertion by $T_1$, $preds$ and $currs$ will change for $lu_2(k_3)$ at stage $s_3$ in \figref{mvostm9} (c). It will caught via \emph{rv\_Validation} method at \Lineref{rvv1} when $(\bp.\bn \neq \bc)$ for $lu_2(k_3)$. After that again it will find the new $preds$ and $currs$ for $lu_2(k_3)$ and eventually it will commit.

\begin{algorithm}
	\scriptsize
	\caption{\emph{rv\_Validation()} }
	\setlength{\multicolsep}{0pt}
	\begin{algorithmic}[1]
		\makeatletter\setcounter{ALG@line}{14}\makeatother
		\Procedure{rv\_validation{()}}{}
		\If{$((\bp.marked) || (\bc.marked) ||(\bp.\bn) \neq \bc || (\rp.\rn) \neq {\rc})$}\label{lin:rvv1}
		\State return $false$. \label{lin:rvv2}
		\Else \label{lin:rvv3}
		\State return $true$. \label{lin:rvv4}
		\EndIf 
		\EndProcedure
	\end{algorithmic}
\end{algorithm}

\cmnt{
\begin{figure}[tbph]
\captionsetup{justification=centering}
	\centerline{\scalebox{0.47}{\input{figs/mvostm9.pdf_t}}}
	\caption{Method validation}
	\label{fig:mvostm91}
\end{figure}
}

\begin{figure}[tbph]
\captionsetup{justification=centering}
	\centerline{\scalebox{0.42}{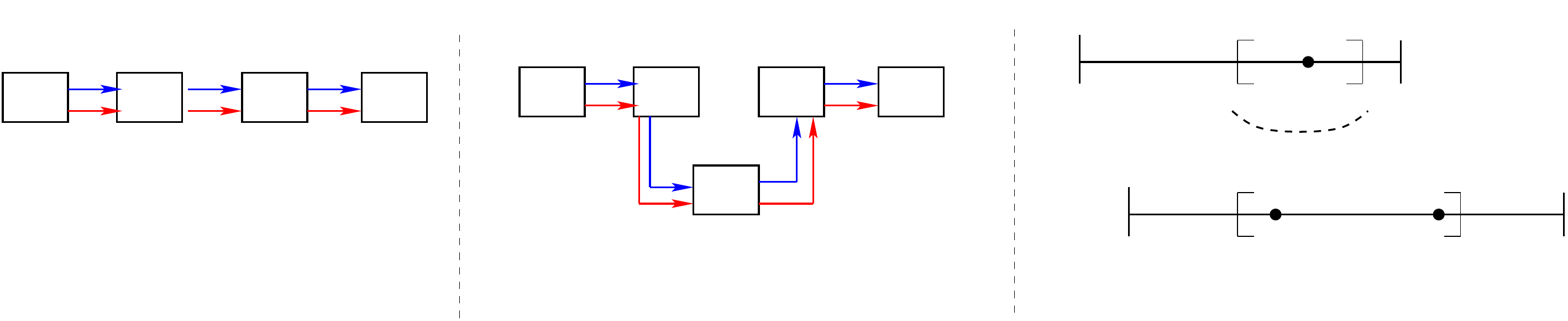}}
	\caption{Interference validation (or rv\_Validation)}
	\label{fig:mvostm9}
\end{figure}
}	

\ignore {

We are considering the chaining \tab{} as a underlying data structure where chaining is done via lazy-list refer \figref{1mvostmdesign} a). Each bucket of the \tab{} is having two sentinel nodes: $head$ and $tail$. $Head$ and $tail$ are initialized as -$\infty$ and +$\infty$ respectively. Keys $\langle k_1, k_2, ...k_n \rangle$ are added in increasing order in the list between the sentinel nodes. 

Each key is maintaining the multiple versions in increasing order of timestamp (\figref{1mvostmdesign} b)). 
For each key $k_1$ of transaction $T_i$, we maintain $k_1.vl$ (version list) which is a list consisting of version tuples in the form $\langle ts, val, mark, rvl, vnext \rangle$.  Description of each field as: $ts$ stands for timestamp which is unique for each transaction, $val$ is the value written by any transaction corresponding to the key, $mark$ is the Boolean variable which can be true or false (if the method corresponding to the key is \npdel{} then the value of $mark$ field will be true, ($T$) and if the method corresponding to the key is \npins{} then the value of $mark$ field will be false, ($F$)), $rvl$ represents $return$-$value$ $list$ which is having all the transactions who has performed \rvmt{} on the same key $k_1$ and $vnext$ is having the information about next available version of the same key $k_1$. The \mvotm system consists of the following main methods: \emph{STM init()}, \emph{STM begin()}, $\npins{}$, $\npluk{}$, $\npdel{}$ and $\nptc{}$.
} 

\vspace{1mm}
\noindent
\textbf{\upmt{s}} - $\tins{}$ and $\tdel$: 
Both the \mth{s} create a version corresponding to the key $k$. The actual effect of \tins{} and \tdel{} in shared memory will take place in \tryc{}. \algoref{mtryc} represents the high-level overview of \tryc{}. 

Initially, to avoid deadlocks, algorithm sorts all the $keys$ in increasing order which are present in the local log, $\llog_i$. In \tryc{}, $\llog_i$ consists of \upmt{s} (\tins{} or \tdel{}) only. For all the \upmt{s} ($opn_i$) it searches the key $k$ in the shared memory corresponding to the bucket $B_k$. It identifies the appropriate location ($pred$ and $curr$) of key $k$ using \bn{} and \rn{} (\Lineref{mtryc3}) in the \lsl of $B_k$ without acquiring any locks similar to \rvmt{} explained above.

Next, it acquires the re-entrant locks on all the $pred$ and $curr$ keys in increasing order. After that, all the $pred$ and $curr$ keys are validated by \emph{tryC\_Validation} in \Lineref{mtryc5} as follows: (1) It does the \emph{rv\_Validation()} as explained above in the \rvmt{}. (2) If key $k$ exists in the $B_k.\lsl$ and let $n$ as a node of $k$. Then algorithm identifies the version of $n$, $ver_j$ which has the timestamp $j$ such that $j$ has the largest timestamp smaller than $i$ (timestamp of $T_i$). If any higher timestamp $k$ of $T_k$ than timestamp $i$ of $T_i$ exist in $ver_j.rvl$ then algorithm returns $Abort$ in \Lineref{mtryc6}.

If all the above steps are true then each \upmt{s} exist in $\llog_i$ will take the effect in the shared memory after doing the \emph{intraTransValidation()} in \Lineref{mtryc11}. If two $\upmt{s}$ of the same transaction have at least one common shared node among its recorded $pred$ and $curr$ keys, then the previous $\upmt{}$ effect may overwrite if the current $\upmt{}$ of $pred$ and $curr$ keys are not updated according to the updates done by the previous $\upmt{}$. Thus to solve this we have \emph{intraTransValidation()} that modifies the $pred$ and $curr$ keys of current operation based on the previous operation in \Lineref{mtryc11}.

\begin{algorithm}
	
	\scriptsize
	\caption{\emph{tryC($T_i$)}: Validate the \upmt{s} of the transaction and then commit}
	\setlength{\multicolsep}{0pt}
	\label{algo:mtryc}
	\begin{algorithmic}[1]
		\makeatletter\setcounter{ALG@line}{27}\makeatother
		\Procedure{$tryC{(T_i)}$}{}
		\State /*Operation name ($opn$) which could be either \tins or \tdel*/
		\State /*Sort the $keys$ of $\llog_i$ in increasing order.*/
		\ForAll{($opn_i$ $\in$ $\llog_i$)} \label{lin:mtryc1}
		\If{(($opn_i$ == \tins{}) $||$ ($opn_i$ == \tdel{}))}\label{lin:mtryc2}
		\State Search in $\lsl$ to identify the $preds[]$ and $currs[]$ for $k$ of $opn_i$ using \bn and \rn in bucket $B_k$. \label{lin:mtryc3}
		\State Acquire the locks on $preds[]$ and $currs[]$ in increasing order. \label{lin:mtryc4}
		\If{($!tryC\_Validation()$)} \label{lin:mtryc5}
		\State return $Abort$.\label{lin:mtryc6}
		\EndIf\label{lin:mtryc7}
		\EndIf\label{lin:mtryc8}
		\EndFor\label{lin:mtryc9}
		\ForAll{($opn_i$ $\in$ $\llog_i$)}\label{lin:mtryc10}
		\State $intraTransValidation()$  modifies the $preds[]$ and $currs[]$ of current operation which would have been updated by the previous operation of the same transaction.\label{lin:mtryc11}
		\If{(($opn_i$ == \tins{}) \&\& ($k  ~ \notin ~ B_k.\lsl$))}\label{lin:mtryc12}
		\State Create new node $n$ with $k$ as: $\langle$ \emph{key = k, lock = false, marked = false, vl = v, nnext = $\phi$} $\rangle$. \label{lin:mtryc13}
		\State Create two versions $v$ as: $\langle$ \emph{ts=i, val=v, rvl=$\phi$, vnext=$\phi$} $\rangle$.
		\State Insert node $n$ into $B_k.\lsl$ such that it is accessible via \rn{} as well as \bn{} \Comment{$lock$ sets $true$}.  \label{lin:mmtryc14}
		\ElsIf{($opn_i$ == \tins{})}\label{lin:mtryc14}
		\State Add the version $v$ as: $\langle$ \emph{ts = i, val = v, rvl = $\phi$, vnext = $\phi$} $\rangle$ into $B_k.\lsl$ such that it is accessible via \rn as well as \bn.\label{lin:mtryc15}
		\EndIf\label{lin:mtryc16}
		\If{($opn_i$ == \tdel{})}\label{lin:mtryc17}
		\State Add the version $i$ as: $\langle$ \emph{ts=i, val=null, rvl=$\phi$, vnext=$\phi$} $\rangle$ into $B_k.\lsl$ such that it is accessible only via \rn.\label{lin:mtryc18}
		\EndIf\label{lin:mtryc19}
		\State Update the $preds[]$ and $currs[]$ of $opn_i$ in $\llog_i$.\label{lin:mtryc20}
		
		\EndFor
		\State Release the locks; return $Commit$.\label{lin:mtryc21}
		\EndProcedure
	\end{algorithmic}
\end{algorithm}
 
Next, we check if \upmt{} is \tins{} and $k$ is in $B_k.\lsl$. If $k$ is not in $B_k$, then create a new node $n$ for $k$ as: $\langle key=k, lock=false, marked=false, vl=v, nnext=\phi \rangle$. This node will have a single version $v$ as: $\langle ts=i, val=v, rvl=\phi, vnext=\phi \rangle$. Here $i$ is the timestamp of the transaction $T_i$ invoking this method; $rvl$ and $vnext$ are initialized to empty values. We set the $val$ as $v$ and insert $n$ into $B_k.\lsl$ such that it is accessible via \rn{} as well as \bn{} and set the lock field to be $true$ (\linref{mmtryc14}). If $k$ is in $B_k.\lsl$ then, $k$ is the same as $\rc$ or $\bc$ or both. Let $n$ be the node of $k$ in $B_k.\lsl$.
Then, we create the version $v$ as: $\langle ts=i, val=v, rvl=\phi, vnext=\phi \rangle$ and insert the version into $B_k.\lsl$ such that it is accessible via \rn{} as well as \bn{} (\linref{mtryc15}).

Subsequently, we check if \upmt{} is \tdel{} and $k$ is in $B_k.\lsl$. Let $n$ be the node of $k$ in $B_k.\lsl$.
Then create the version $v$ as: $\langle ts=i, val=null, rvl=\phi, vnext=\phi \rangle$ and insert the version into $B_k.\lsl$ such that it is accessible only via \rn{} (\linref{mtryc18}). 

Finally, at \Lineref{mtryc20} it updates the $pred$ and $curr$ of $opn_i$ in local log, $\llog_i$. At \Lineref{mtryc21} releases the locks on all the $pred$ and $curr$ in increasing order of keys to avoid deadlocks and return $Commit$.  

We illustrate the helping methods of \rvmt{} and \upmt{} as follows:

\noindent
\textbf{rv\_Validation():} It is called by both the \rvmt{} and \upmt{}. It identifies the conflicts among the concurrent methods of different transactions. Consider an example shown in \figref{mvostm9}, where two concurrent conflicting methods of different transactions are working on the same key $k_3$. Initially, at stage $s_1$ in \figref{mvostm9} (c) both the conflicting method optimistically (without acquiring locks) identify the same $pred$ and $curr$ keys for key $k_3$ from $B_k.\lsl$ in \figref{mvostm9} (a). At stage $s_2$ in \figref{mvostm9} (c), method $ins_1(k_3)$ of transaction $T_1$ acquired the lock on $pred$ and $curr$ keys and inserted the node into $B_k.\lsl$ as shown in \figref{mvostm9} (b). After successful insertion by $T_1$, $pred$ and $curr$ has been changed for $lu_2(k_3)$ at stage $s_3$ in \figref{mvostm9} (c). So, the above modified information is delivered by \emph{rv\_Validation} method at \Lineref{rvv1} when $(\bp.\bn \neq \bc)$ for $lu_2(k_3)$. After that again it will find the new $pred$ and $curr$ for $lu_2(k_3)$ and eventually it will commit. 
\begin{algorithm}
	\scriptsize
	\caption{\emph{rv\_Validation()} }
	\setlength{\multicolsep}{0pt}
	\begin{algorithmic}[1]
		\makeatletter\setcounter{ALG@line}{55}\makeatother
		\Procedure{rv\_validation{()}}{}
		\If{$((\bp.marked) || (\bc.marked) ||(\bp.\bn) \neq \bc || (\rp.\rn) \neq {\rc})$}\label{lin:rvv1}
		\State return $false$. \label{lin:rvv2}
		\Else \label{lin:rvv3}
		\State return $true$. \label{lin:rvv4}
		\EndIf 
		\EndProcedure
	\end{algorithmic}
\end{algorithm}
\begin{figure}
	\captionsetup{justification=centering}
	\centerline{\scalebox{0.4}{\input{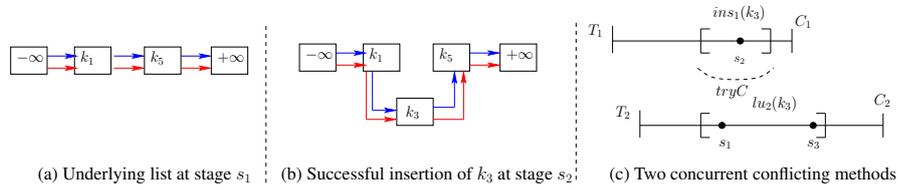}}}
	\caption{rv\_Validation}
	\label{fig:mvostm9}
\end{figure}

\begin{figure}[H] 
	\captionsetup{justification=centering}
	\centerline{\scalebox{0.45}{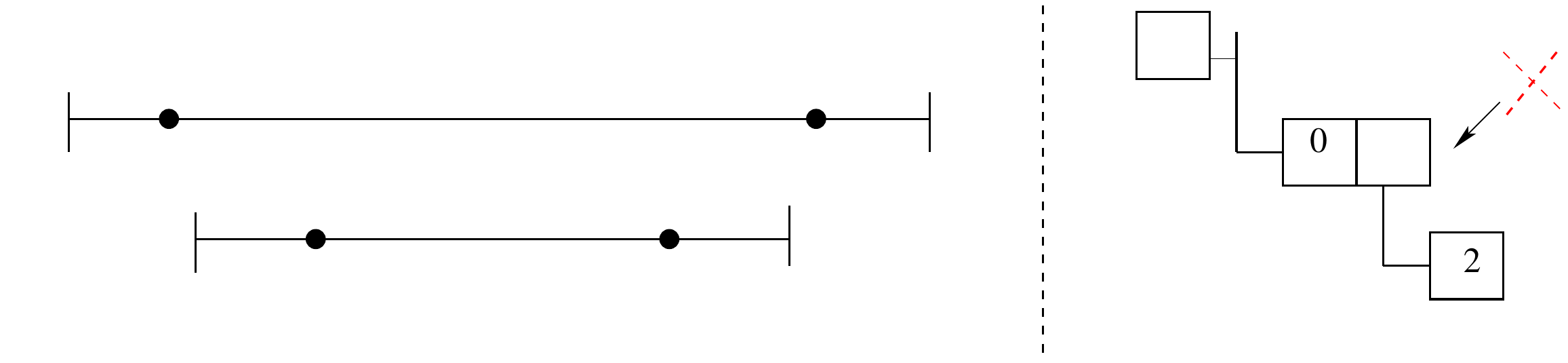}}
	\caption{tryC\_Validation}
	\label{fig:mvostm12}
\end{figure}

\noindent
\textbf{tryC\_Validation:} It is called by \upmt{} in \tryc{}. First it does the \emph{rv\_Validation()} in \Lineref{trycv1}. If its successful and key $k$ exists in the $B_k.\lsl$ and let $n$ as a node of $k$. Then algorithm identifies the version of $n$, $ver_j$ which has the timestamp $j$ such that $j$ has the largest timestamp smaller than $i$ (timestamp of $T_i$). If any higher timestamp $T_k$ than timestamp $T_i$ exist in $ver_j.rvl$ then algorithm returns false (in \Lineref{trycv8}) and eventually return $Abort$ in \Lineref{mtryc6}. Consider an example as shown in \figref{mvostm12} (a), where second method $ins_1$ of transaction $T_1$ returns $Abort$ because higher timestamp of transaction $T_2$ is already present in the $rvl$ of version $T_0$ identified by $T_1$ in \figref{mvostm12} (b).

\begin{algorithm}[H]
	\scriptsize
	\caption{\emph{tryC\_Validation()} }
	\setlength{\multicolsep}{0pt}
	\begin{algorithmic}[1]
		\makeatletter\setcounter{ALG@line}{62}\makeatother
		\Procedure{tryC\_validation{()}}{}
		\If{($!rv\_Validation()$)}\label{lin:trycv1}
		\State Release the locks and retry.\label{lin:trycv2}
		\EndIf\label{lin:trycv3}
		\If{(k $\in$ $B_k.\lsl$)}\label{lin:trycv4}
		\State Identify the version $ver_j$ with $ts=j$ such that $j$ is the largest timestamp smaller than $i$.\label{lin:trycv5}
		\ForAll {$T_k$ in $ver_j.rvl$}\label{lin:trycv6}
		\If{(TS($(T_k)$ $>$ TS($T_i$)))}\label{lin:trycv7}
		\State return $false$.\label{lin:trycv8}
		\EndIf \label{lin:trycv9}
		\EndFor\label{lin:trycv10}
		\EndIf\label{lin:trycv11}
		\State return $true$.\label{lin:trycv12}
		\EndProcedure
	\end{algorithmic}
\end{algorithm}
\begin{algorithm}[H]
	\label{alg:intra} 
	\scriptsize
	\caption{\emph{intraTransValidation()} }
	\setlength{\multicolsep}{0pt}
	\begin{algorithmic}[1]
		\makeatletter\setcounter{ALG@line}{76}\makeatother
		\Procedure{intraTransValidation{()}}{}
		\If{$((\bp.marked) || (\bp.\bn \neq \bc ))$} \label{lin:intra1}
		\If{($opn_k$ == Insert)}\label{lin:intra2}
		\State $\bp_{i}$ $\gets$ $\bp_{k}$.\bn.\label{lin:intra3}
		\Else\label{lin:intra4}
		\State $\bp_{i}$ $\gets$ $\bp_{k}$.\label{lin:intra5}
		\EndIf\label{lin:intra6}
		\EndIf \label{lin:intra9}
		\If{(\rp.\rn $\neq$ \rc)}\label{lin:intra7}
		\State $\rp_{i}$ $\gets$ $\rp_{k}.\rn$.\label{lin:intra8}
		\EndIf\label{lin:intra10}
		\EndProcedure
	\end{algorithmic}
\end{algorithm}

\noindent
\textbf{intraTransValidation:} It is called by \upmt{} in \tryc{}. If two $\upmt{s}$ of the same transaction have at least one common shared node among its recorded $pred$ and $curr$ keys, then the previous $\upmt{}$ effect may overwrite if the current $\upmt{}$ of $pred$ and $curr$ keys are not updated according to the updates done by the previous $\upmt{}$. Thus to solve this we have \emph{intraTransValidation()} that modifies the $pred$ and $curr$ keys of current operation based on the previous operation from \Lineref{intra1} to \Lineref{intra10}. Consider an example as shown in \figref{mvostm11}, where two \upmt{s} of transaction $T_1$ are $ins_{11}(k_3)$ and $ins_{12}(k_5)$ in \figref{mvostm11} (c). At stage $s_1$ in \figref{mvostm11} (c) both the \upmt{s} identify the same $pred$ and $curr$ from underlying DS as $B_k.\lsl$ shown in \figref{mvostm11} (a). After the successful insertion done by first \upmt{} at stage $s_2$ in \figref{mvostm11} (c), key $k_3$ is part of $B_k.\lsl$ (\figref{mvostm11} (b)). At stage $s_3$ in \figref{mvostm11} (c), $ins_{12}(k_5)$ identified $(\bp.\bn \neq \bc)$ in \emph{intraTransValidation()} at \Lineref{intra1}. So it updates the $\bp$ in \Lineref{intra3} for correct updation in $B_k.\lsl$. 




\begin{figure}[H]
	\captionsetup{justification=centering}
	\centerline{\scalebox{0.43}{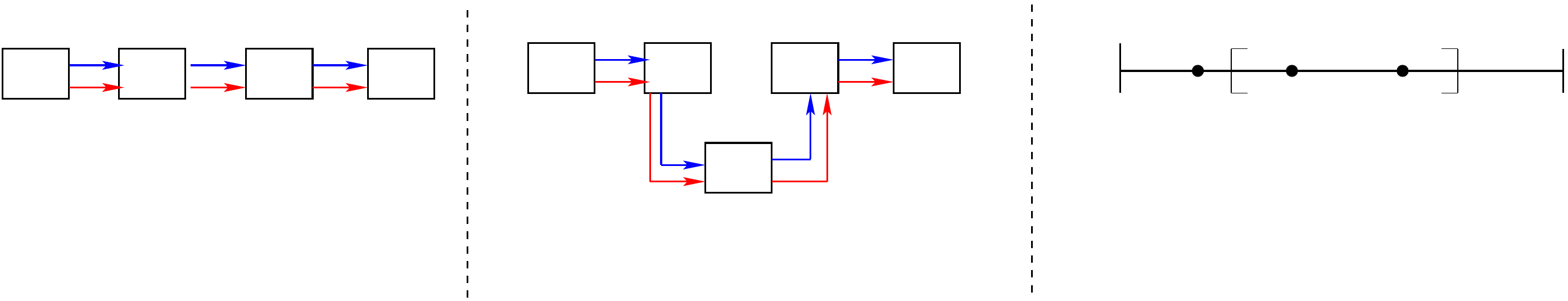}}
	\caption{Intra transaction validation}
	\label{fig:mvostm11}
\end{figure}

\noindent

\section{Correctness of \hmvotm}
\label{sec:cmvostm}
In this section, we will prove that our implementation satisfies opacity. Consider the history $H$ generated by \emph{MVOSTM} algorithm. Recall that only the \emph{STM\_begin}, \rvmt{}, \upmt{} (or $tryC$) access shared memory. 

Note that $H$ is not necessarily sequential: the transactional \mth{s} can execute in overlapping manner. To reason about correctness we have to prove $H$ is opaque. Since we defined opacity for histories which are sequential, we order all the overlapping \mth{s} in $H$ to get an equivalent sequential history. We then show that this resulting sequential history satisfies \mth{}.

We order overlapping \mth{s} of $H$ as follows: (1) two overlapping \emph{STM\_begin} \mth{s} based on the order in which they obtain lock over $counter$; (2) two \rvmt{s} accessing the same key $k$ by their order of obtaining lock over $k$; (3) a \rvmt{} $rvm_i(k)$ and a $\tryc_j$, of a transaction $T_j$ which has written to $k$, are similarly ordered by their order of obtaining lock over $k$; (4) similarly, two \tryc{} \mth{s} based on the order in which they obtain lock over same key $k$.

Combining the real-time order of events with above mentioned order, we obtain a partial order which we denote as \emph{$\locko_H$}. (It is a partial order since it does not order overlapping \rvmt{s} on different $keys$ or an overlapping \rvmt{} and a \tryc{} which do not access any common $key$).

In order for $H$ to be sequential, all its \mth{s} must be ordered. Let $\alpha$ be a total order or \emph{linearization} of \mth{s} of $H$ such that when this order is applied to $H$, it is sequential. We denote the resulting history as $H^{\alpha} = \seq{\alpha}{H}$. 
We now argue about the \validity{} of histories generated by the algorithm. 
\begin{lemma}
	\label{lem:histvalid}
	Consider a history $H$ generated by the algorithm. Let $\alpha$ be a linearization of $H$ which respects $\locko_H$, i.e. $\locko_H \subseteq \alpha$. Then $H^{\alpha} = \seq{\alpha}{H}$ is \valid. 
\end{lemma}

\begin{proof}
	Consider a successful \rvmt{} $rvm_i(k)$ that returns value $v$. The \rvmt{} first obtains lock on key $k$. Thus the value $v$ returned by the \rvmt{} must have already been stored in $k$'s version list by a transaction, say $T_j$ when it successfully returned OK from its \tryc{} \mth{} (if $T_j \neq T_0$). For this to have occurred, $T_j$ must have successfully locked and released $k$ prior to $T_i$'s locking \mth. Thus from the definition of $\locko_H$, we get that $\tryc_j(ok)$ occurs before $rvm_i(k,v)$ which also holds in $\alpha$. 
	
	If $T_j$ is $T_0$, then by our assumption we have that $T_j$ committed before the start of any \mth{} in $H$. Hence, this automatically implies that in both cases $H^{\alpha}$ is \valid.  
\end{proof}

It can be seen that for proving correctness, any linearization of a history $H$ is sufficient as long as the linearization respects $\locko_H$. The following lemma formalizes this intuition, 

\begin{lemma}
	\label{lem:histseq}
	Consider a history $H$. Let $\alpha$ and $\beta$ be two linearizations of $H$ such that both of them respect $\locko_H$, i.e. $\locko_H \subseteq \alpha$ and $\locko_H \subseteq \beta$. Then, $H^{\alpha} = \seq{\alpha}{H}$ is opaque if $H^{\beta} = \seq{\beta}{H}$ is opaque. 
\end{lemma}

\begin{proof} From \lemref{histvalid}, we get that both $H^{\alpha}$ and $H^{\beta}$ are \valid{} histories. Now let us consider each case \\
	\textbf{If:}  Assume that $H^{\alpha}$ is opaque. Then, we get that there exists a legal \tseq{} history $S$ that is equivalent to $\overline{H^{\alpha}}$. From the definition of $H^{\beta}$, we get that $\overline{H^{\alpha}}$ is equivalent to $\overline{H^{\beta}}$. Hence, $S$ is equivalent to $\overline{H^{\beta}}$ as well. We also have that, $\prec_{H^{\alpha}}^{RT} \subseteq \prec_{S}^{RT}$. From the definition of $\locko_H$, we get that $\prec_{H^{\alpha}}^{RT} = \prec_{\locko_H}^{RT} = \prec_{H^{\beta}}^{RT}$. This automatically implies that $\prec_{H^{\beta}}^{RT} \subseteq \prec_{S}^{RT}$. Thus $H^{\beta}$ is opaque as well. 
	
	~ \\
	\textbf{Only if:} This proof comes from symmetry since $H^{\alpha}$ and $H^{\beta}$ are not distinguishable. 
\end{proof}

This lemma shows that, given a history $H$, it is enough to consider one sequential history $H^{\alpha}$ that respects $\locko_H$ for proving correctness. If this history is opaque, then any other sequential history that respects $\locko_H$ is also opaque.

Consider a history $H$ generated by \hmvotm{} algorithm. We then generate a sequential history that respects $\locko_H$. For simplicity, we denote the resulting sequential history of \hmvotm{} as $H_{to}$. Let $T_i$ be a committed transaction in $H_{to}$ that writes to $k$ (i.e. it creates a new version of $k$). 

To prove the correctness, we now introduce some more notations. We define $\stl{T_i}{k}{H_{to}}$ as a committed transaction $T_j$ such that $T_j$ has the smallest timestamp greater than $T_i$ in $H_{to}$ that writes to $k$ in $H_{to}$. Similarly, we define $\lts{T_i}{k}{H_{to}}$ as a committed transaction $T_k$ such that $T_k$ has the largest timestamp smaller than $T_i$ that writes to $k$ in $H_{to}$. \cmnt{ We denote $\vli{ts}{x}{H_{to}}$, as the $v\_tuple$ in $x.vl$ after executing all the events in $H_{to}$ created by transaction $T_{ts}$. If no such $v\_tuple$ exists then it is nil. }Using these notations, we describe the following properties and lemmas on $H_{to}$,

\begin{property}
	\label{prop:uniq-ts}
	Every transaction $T_i$ is assigned an unique numeric timestamp $i$.
\end{property}

\begin{property}
	\label{prop:ts-inc}
	If a transaction $T_i$ begins after another transaction $T_j$ then $j < i$.
\end{property}

\begin{property}
	\label{prop:readsfrom}
	If a transaction $T_k$ lookup key $k_x$ from (a committed transaction) $T_j$ then $T_j$ is a committed transaction updating to $k_x$ with $j$ being the largest timestamp smaller than $k$. Formally, $T_j = \lts{k_x}{T_k}{H_{to}}$. 
\end{property}

\cmnt{
	\begin{definition}
		\label{def:fLP}
		Linearization Point (LP) is the first unlocking point of each successful method.
	\end{definition}
}

\begin{lemma}
	\label{lem:readswrite}
	Suppose a transaction $T_k$ lookup $k_x$ from (a committed transaction) $T_j$ in $H_{to}$, i.e. $\{\up_j(k_{x,j}, v), \rvm_k(k_{x,i}, v)\} \in \evts{H_{to}}$. Let $T_i$ be a committed transaction that updates to $k_x$, i.e. $\up_i(k_{x,i}, u) \in \evts{T_i}$. Then, the timestamp of $T_i$ is either less than $T_j$'s timestamp or greater than $T_k$'s timestamp, i.e. $i<j \oplus k<i$ (where $\oplus$ is XOR operator). 
\end{lemma}

\begin{proof}
	We will prove this by contradiction. Assume that $i<j \oplus k<i$ is not true. This implies that, $j<i<k$. But from the implementation of \rvmt{} and \tryc{} \mth{s}, we get that either transaction $T_i$ is aborted or $T_k$ lookup $k$ from $T_i$ in $H$. Since neither of them are true, we get that $j<i<k$ is not possible. Hence, $i<j \oplus k<i$. 
\end{proof}

To show that $H_{to}$ satisfies opacity, we use the graph characterization developed above in \secref{gcofo}. For the graph characterization, we use the version order defined using timestamps. Consider two committed transactions $T_i, T_j$ such that $i < j$. Suppose both the transactions write to key $k$. Then the versions created are ordered as: $k_i \ll k_j$. We denote this version order on all the $keys$ created as $\ll_{to}$. Now consider the opacity graph of $H_{to}$ with version order as defined by $\ll_{to}$, $G_{to} = \opg{H_{to}}{\ll_{to}}$. In the following lemmas, we will prove that $G_{to}$ is acyclic.

\begin{lemma}
	\label{lem:edgeorder}
	All the edges in $G_{to} = \opg{H_{to}}{\ll_{to}}$ are in timestamp order, i.e. if there is an edge from $T_j$ to $T_i$ then the $j < i$. 
\end{lemma}

\begin{proof}
	To prove this, let us analyze the edges one by one, 
	\begin{itemize}
		\item \rt{} edges: If there is a \rt{} edge from $T_j$ to $T_i$, then $T_j$ terminated before $T_i$ started. Hence, from \propref{ts-inc} we get that $j < i$.
		
		\item \rvf{} edges: This follows directly from \propref{readsfrom}. 
		
		\item \mv{} edges: The \mv{} edges relate a committed transaction $T_k$ updates to a key $k$, $up_k(k,v)$; a successful \rvmt{} $rvm_j(k,u)$ belonging to a transaction $T_j$ lookup $k$ updated by a committed transaction $T_i$, $up_i(k, u)$. Transactions $T_i, T_k$ create new versions $k_i, k_k$ respectively. According to $\ll_{to}$, if $k_k \ll_{to} k_i$, then there is an edge from $T_k$ to $T_i$. From the definition of $\ll_{to}$ this automatically implies that $k < i$.
		
		On the other hand, if $k_i \ll_{to} k_k$ then there is an edge from $T_j$ to $T_k$. Thus in this case, we get that $i < k$. Combining this with \lemref{readswrite}, we get that $j < k$.

		
	\end{itemize}
	Thus in all the cases we have shown that if there is an edge from $T_j$ to $T_i$ then the $j < i$.
\end{proof}

\begin{theorem}
	\label{thm:to-opaque}
	Any history $H_{to}$ generated by \hmvotm is \opq. 
\end{theorem}

\begin{proof}
	From the definition of $H_{to}$ and \lemref{histvalid}, we get that $H_{to}$ is \valid. We show that $G_{to} = \opg{H_{to}}{\ll_{to}}$ is acyclic. We prove this by contradiction. Assume that $G_{to}$ contains a cycle of the form, $T_{c1} \rightarrow T_{c2} \rightarrow .. T_{cm} \rightarrow T_{c1}$. From \lemref{edgeorder} we get that, $c1 < c2 < ... < cm < c1$ which implies that $c1 < c1$. Hence, a contradiction. This implies that $G_{to}$ is acyclic. Thus from \thmref{opg} we get that $H_{to}$ is opaque.
\end{proof}

Now, it is left to show that our algorithm is \emph{live}, i.e., under certain conditions, every operation eventually completes. We have to show that the transactions do not deadlock. This is because all the transactions lock all the $keys$ in a predefined order. As discussed earlier, the STM system orders all $keys$. We denote this order as \aco and denote it as $\prec_{ao}$. Thus $k_1 \prec_{ao} k_2 \prec_{ao} ... \prec_{ao} k_n$. 

From \aco, we get the following property

\begin{property}
	\label{prop:aco}
	Suppose transaction $T_i$ accesses shared objects $p$ and $q$ in $H$. If $p$ is ordered before $q$ in \aco, then $lock(p)$ by transaction $T_i$ occurs before $lock(q)$. Formally, $(p \prec_{ao} q) \Leftrightarrow (lock(p) <_H lock(q))$. 
\end{property}

\begin{theorem}
	\label{thm:prog}
	\hmvotm with unbounded versions ensures that \rvmt{s} do not abort.
\end{theorem}
\begin{proof}
	This is self explanatory with the help of \hmvotm{} algorithm because each $key$ is maintaining multiple versions in the case of unbounded versions. So \rvmt{} always finds a correct version to read it from. Thus, \rvmt{s} do not $abort$. 
\end{proof}

\noindent \thmref{prog} gives us a nice property a transaction with \tlook only \mth{s} will not abort. 
\section{Experimental Evaluation}
\label{sec:exp}

In this section, we present our experimental results. We have two main goals in this section: (1) evaluating the benefit of multi-version object STMs over the single-version object STMs, and (2) evaluating the benefit of multi-version object STMs over  multi-version read-write STMs. We use the \hmvotm described in \secref{pcode} as well as the corresponding \lmvotm which implements the list object. We also consider extensions of these multi-version object STMs to reduce the memory usage. Specifically, we consider a variant that implements garbage collection with unbounded versions and another variant where the number of versions never exceeds a given threshold $K$. 


\vspace{1mm}
\noindent
\textbf{Experimental system:} The Experimental system is a large-scale 2-socket Intel(R) Xeon(R) CPU E5-2690 v4 @ 2.60GHz with 14 cores per socket and two hyper-threads (HTs) per core, for a total of 56 threads. Each core has a private 32KB L1 cache and 256 KB L2 cache (which is shared among HTs on that core). All cores on a socket share a 35MB L3 cache. The machine has 32GB of RAM and runs Ubuntu 16.04.2 LTS. All code was compiled with the GNU C++ compiler (G++) 5.4.0 with the build target x86\_64-Linux-gnu and compilation option -std=c++1x -O3.

\vspace{1mm}
\noindent
\textbf{STM implementations:} We have taken the implementation of NOrec-list \cite{Dalessandro+:NoRec:PPoPP:2010}, Boosting-list \cite{HerlihyKosk:Boosting:PPoPP:2008}, Trans-list \cite{ZhangDech:LockFreeTW:SPAA:2016},  ESTM \cite{Felber+:ElasticTrans:2017:jpdc}, and RWSTM directly from the TLDS framework\footnote{https://ucf-cs.github.io/tlds/}. And the implementation of OSTM and MVTO published by the author. We implemented our algorithms in C++. Each STM algorithm first creates N-threads, each thread, in turn, spawns a transaction. Each transaction exports the following methods as follows: \tbeg{}, \tins{}, \tlook{}, \tdel{} and \tryc{}.

\vspace{1mm}
\noindent
\textbf{Methodology:\footnote{Code is available here: https://github.com/PDCRL/MVOSTM}} We have considered two types of workloads: ($W1$) Li - Lookup intensive (90\% lookup, 8\% insert and 2\% delete) and ($W2$) Ui - Update intensive(10\% lookup,  45\% insert and 45\% delete). The experiments are conducted by varying number of threads from 2 to 64 in power of 2, with 1000 keys randomly chosen. We assume that the \tab{} of \hmvotm has five buckets and each of the bucket (or list in case of \lmvotm) can have a maximum size of 1000 keys. Each transaction, in turn, executes 10 operations which include \tlook, \tdel{} and \tins{} operations. 
We take an average over 10 results as the final result for each experiment.

\vspace{1mm}
\noindent
\textbf{Results:}
\figref{htw1} shows \hmvotm outperforms all the other algorithms(HT-MVTO, RWSTM, ESTM, HT-OSTM) by a factor of 2.6, 3.1, 3.8, 3.5 for workload type $W1$ and by a factor of 10, 19, 6, 2 for workload type $W2$ respectively. As shown in \figref{htw1}, List based MVOSTM (\lmvotm)  performs even better compared with the existing state-of-the-art STMs (list-MVTO, NOrec-list, Boosting-list, Trans-list, list-OSTM) by a factor of 12, 24, 22, 20, 2.2 for workload type $W1$ and by a factor of  169, 35, 24, 28, 2 for workload type $W2$ respectively.
As shown in \figref{htw2} for both types of workloads, HT-MVOSTM and list-MVOSTM have the least number of aborts.

\begin{figure}[H]
	\centering
	\captionsetup{justification=centering}
	\includegraphics[width=13cm]{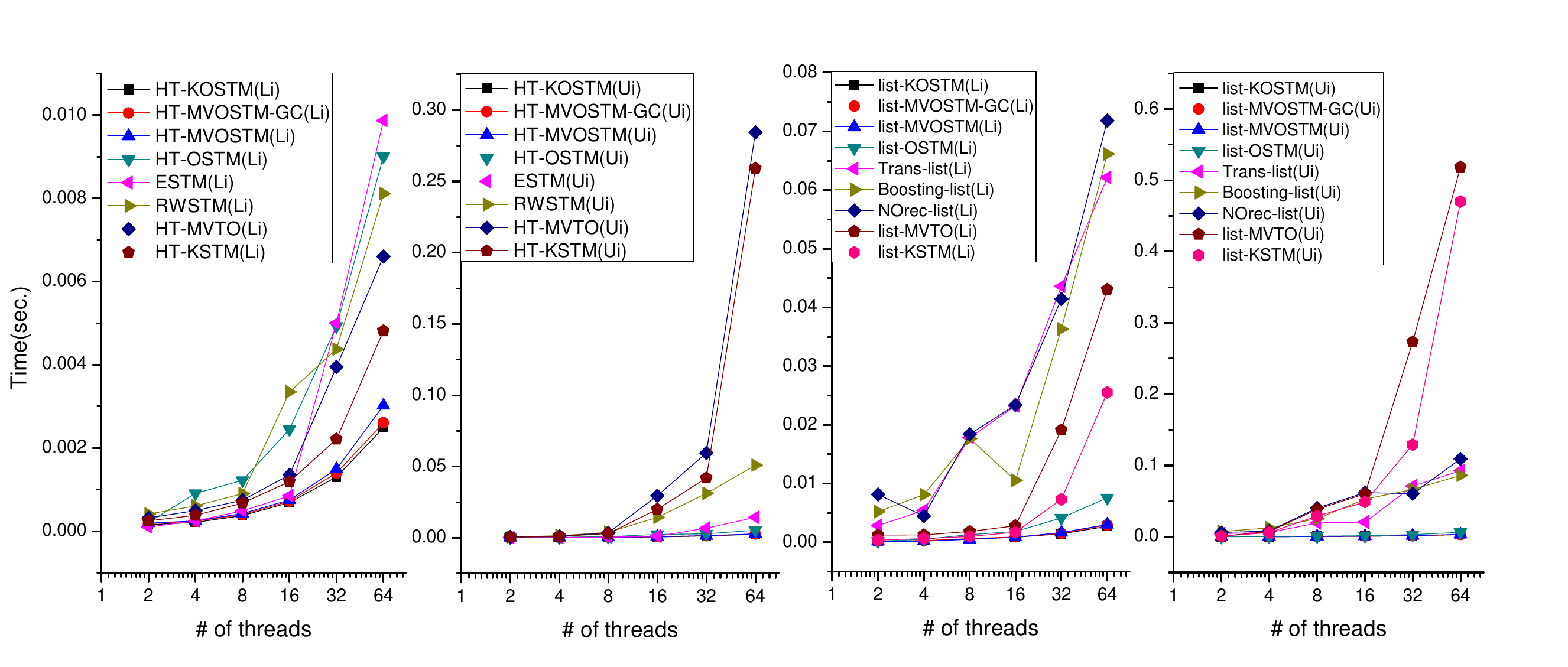}
	\centering
	\caption{Performance of \hmvotm and \lmvotm}\label{fig:htw1}
\end{figure}
\begin{figure}[H]
	\captionsetup{justification=centering}
	\includegraphics[width=13.5cm]{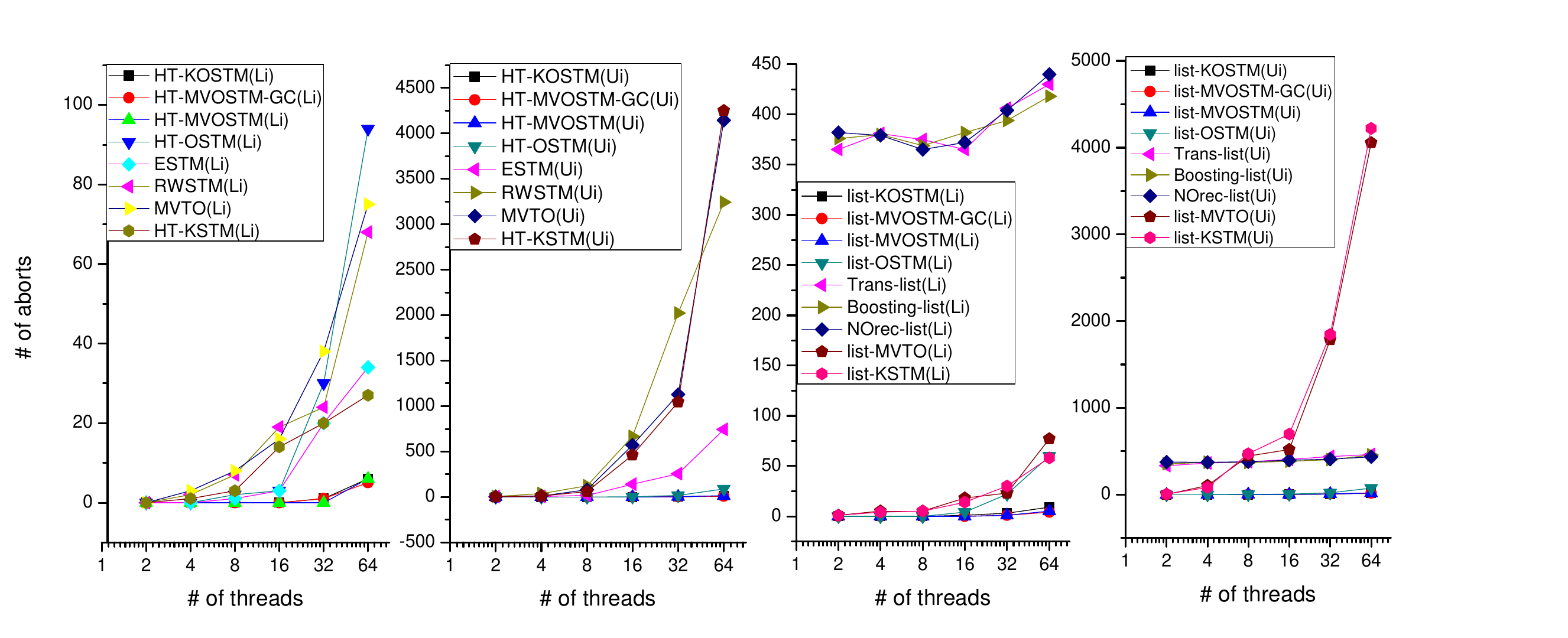}
	\centering
	\caption{Aborts of \hmvotm and \lmvotm}\label{fig:htw2}
\end{figure}

\noindent
\\
\textbf{MVOSTM-GC and KOSTM:} For efficient memory utilization, we develop two variations of \mvotm.\cmnt{\todo{SK: Fix KOSTM macro}}
The first, \mvotmgc, uses unbounded versions but performs garbage collection. \textbf{This is achieved by deleting non-latest versions whose timestamp is less than the timestamp of  the least live transaction.} \mvotmgc gave a performance gain of 15\%  over \mvotm without garbage collection in the best case.
%
The second, \kotm, keeps at most $K$  versions by deleting the oldest version when $(K+1)^{th}$ version is created by a current transaction. As \kotm has limited number of versions while \mvotmgc can have infinite versions, the memory consumed by \kotm is 21\% less than \mvotm. (Implementation details for both are in the below.)

\cmnt{\todo{SK:Remove this paragraph}
one with garbage collection on unbounded MVOSTMs (MVOSTM-GC.\footnote{Implementation details of MVOSTM-GC and KOSTM are give in appendix. ???\label{chirag}}), where each transaction that wants to create a new version of a key checks for the least live transaction (LLT) in the system, if current transaction is LLT then it deletes all the previous versions of that key and create one of its own. Other variation is by using finite limit on the number of versions or bounding the versions in MVOSTMs (KOSTM \footref{chirag}), where oldest versions is overridden by a validated transaction that wishes to create a new version once the limit of version count is reached.} 

We have integrated these variations in both \tab{} based (\hmvotmgc and \hkotm) and linked-list based MVOSTMs (\lmvotmgc and \lkotm), we observed that these two variations increase the performance, concurrency and reduces the number of aborts as compared to MVOSTM.

Experiments show that these variations outperform the corresponding MVOSTMs. Between these two variations, \kotm perform better than \mvotmgc as shown in \figref{htw1} and \figref{htw2}. \hkotm helps to achieve a performance speedup of 1.22 and 1.15 for workload type $W1$ and speedup of 1.15 and 1.08 for workload type $W2$ as compared to \hmvotm and \hmvotmgc respectively. Whereas \lkotm (with four versions) gives a speedup of 1.1, 1.07 for workload type $W1$ and speedup of 1.25, 1.13 for workload type $W2$ over the \lmvotm and \lmvotmgc respectively. 

\vspace{1mm}
\noindent
\textbf{Mid-Intensive workload:} 
Similar to Lookup intensive and Update intensive experiments we have conducted experiments for mid intensive workload ($W3$) as well where we have considered 50\% update operations(25\% insert, 25\% delete) and 50\% read operations. Under this setting again \mvotm outperforms all the other algorithms for both \hmvotm and \lmvotm. 
\figref{pmidI} shows \hmvotm outperforms all the other algorithms(HT-MVTO, RWSTM, ESTM, HT-OSTM) by a factor of 10.1, 4.85, 3, 1.4 for workload type $W3$ respectively. As shown in \figref{pmidI}, list-based MVOSTM (\lmvotm)  performs even better compared with the existing state-of-the-art algorithms (list-MVTO, NOrec-list, Boosting-list, Trans-list, list-OSTM) by a factor of 26.8, 29.4, 25.9, 20.9, 1.58 for workload type $W3$ respectively. Even the abort count for \mvotm is least as compared to all other algorithms for both \hmvotm as well as \lmvotm. \figref{amidI} shows our experimental results for the abort count.
\begin{figure}
	\centering
	\captionsetup{justification=centering}
	\includegraphics[width=13cm]{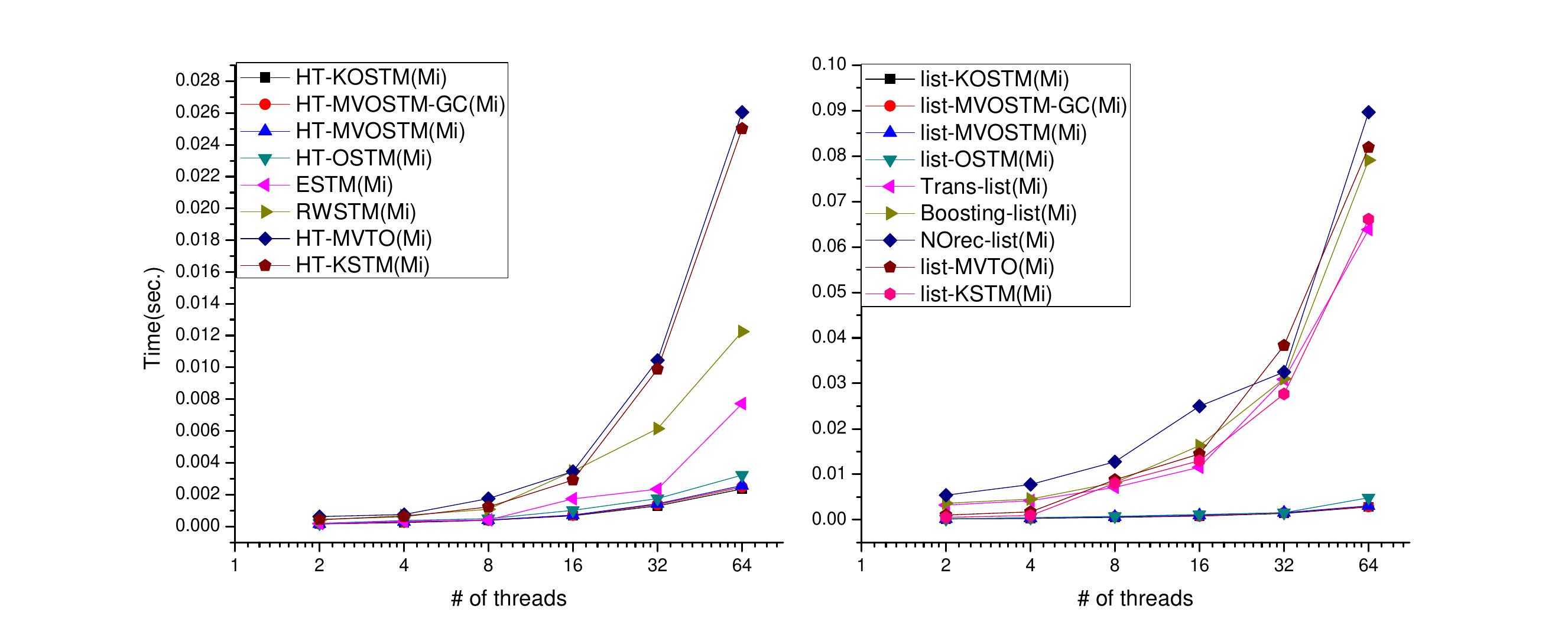}
	\centering
	\caption{Performance of \hmvotm and \lmvotm}\label{fig:pmidI}
\end{figure}
\begin{figure}
	\centering
	\captionsetup{justification=centering}
	\includegraphics[width=13cm]{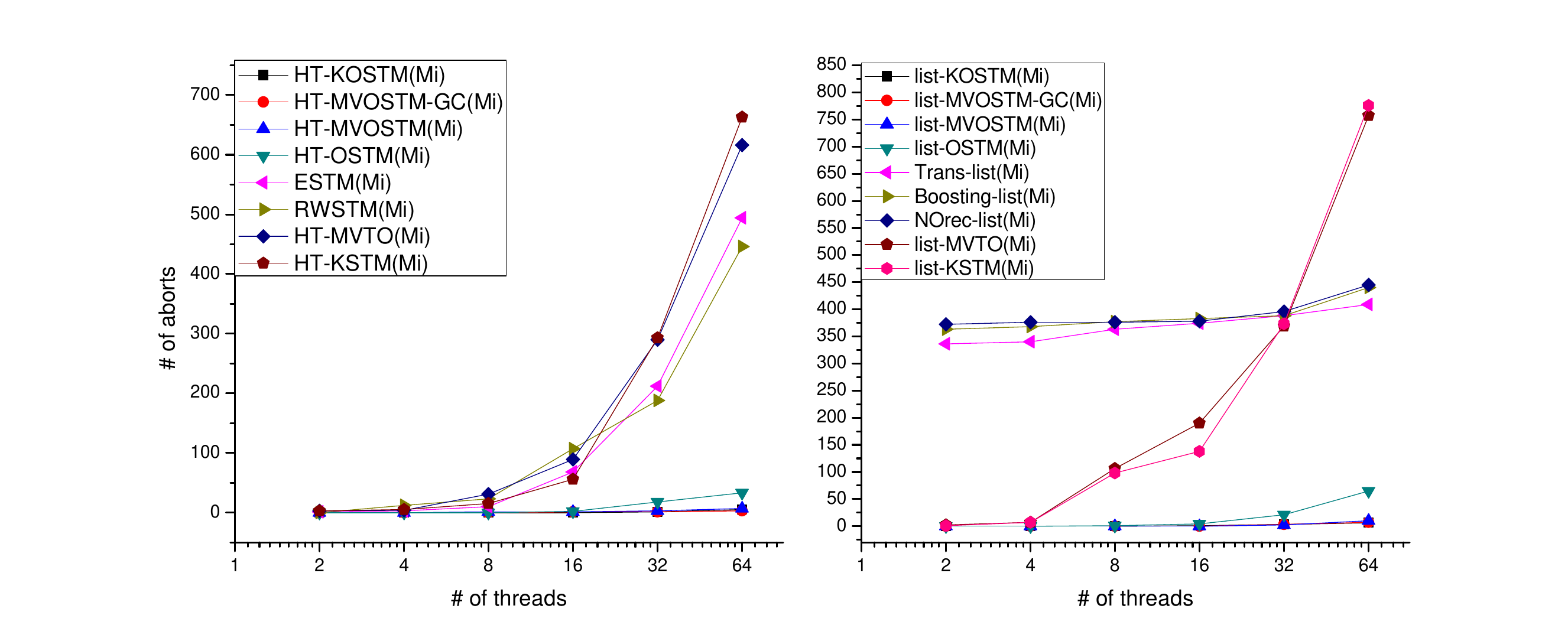}
	\centering
	\caption{Aborts of \hmvotm and \lmvotm}\label{fig:amidI}
\end{figure}

\vspace{1mm}
\noindent
\textbf{Garbage Collection in MVOSTMs (\mvotmgc):} Providing multiple versions to increase the performance of OSTMs in MVOSTMs lead to more space requirements. As many unnecessary versions pertain in the memory a technique to remove these versions or to collect these garbage versions is required. Hence we came up with the idea of garbage collection in MVOSTMs. We have implemented garbage collection for \mvotm for both \tab{} and linked-list based approaches. Each transaction, in the beginning, logs its time stamp in a global list named as ALTL (All live transactions list), which keeps track of all the live transactions in the system. Under the optimistic approach of STM, each transaction performs its updates in the shared memory in the update execution phase. Each transaction in this phase performs some validations and if all validations are completed successfully a version of that key is created by that transaction. When a transaction goes to create a version of a key in the shared memory, it checks for the least time stamp live transaction present in the ALTL. If the current transaction is the one with least timestamp present in ALTL, then this transaction deletes all the older versions of the current key and create a version of its own. If current transaction is not the least timestamp live transaction then it doesn't do any garbage collection. In this way, we ensure each transaction performs garbage collection on the keys it is going to create a version on. Once the transaction, changes its state to commit, it removes its entry from the ALTL. As shown in \figref{KHTGC} and \figref{KlistGC} MVOSTM with garbage collection (\hmvotmgc and \lmvotmgc) performs better than MVOSTM without garbage collection.
\begin{figure}
	\centering
	\captionsetup{justification=centering}
	\includegraphics[width=13cm]{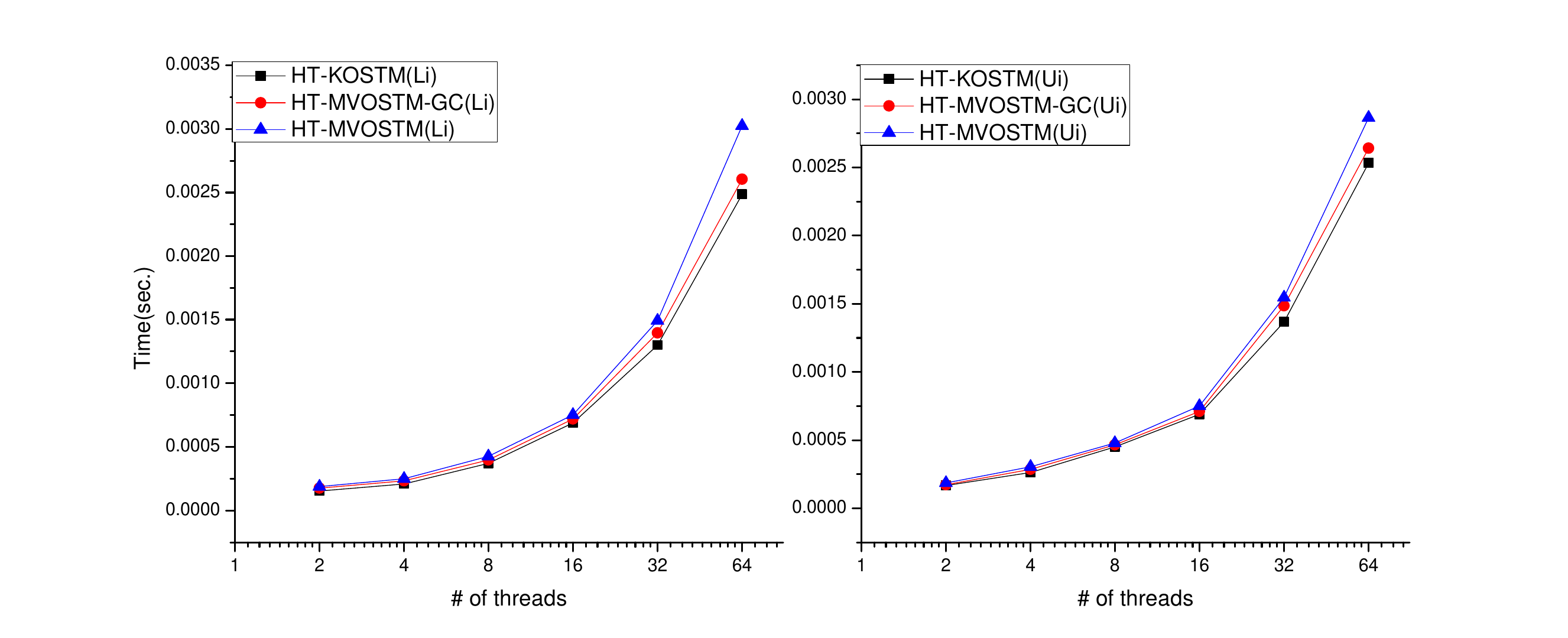}
	\centering
	\caption{Performance comparisons of variations (\hmvotm and \hkotm) of \hmvotm}\label{fig:KHTGC}
\end{figure}
\begin{figure}
	\centering
	\captionsetup{justification=centering}
	\includegraphics[width=13cm]{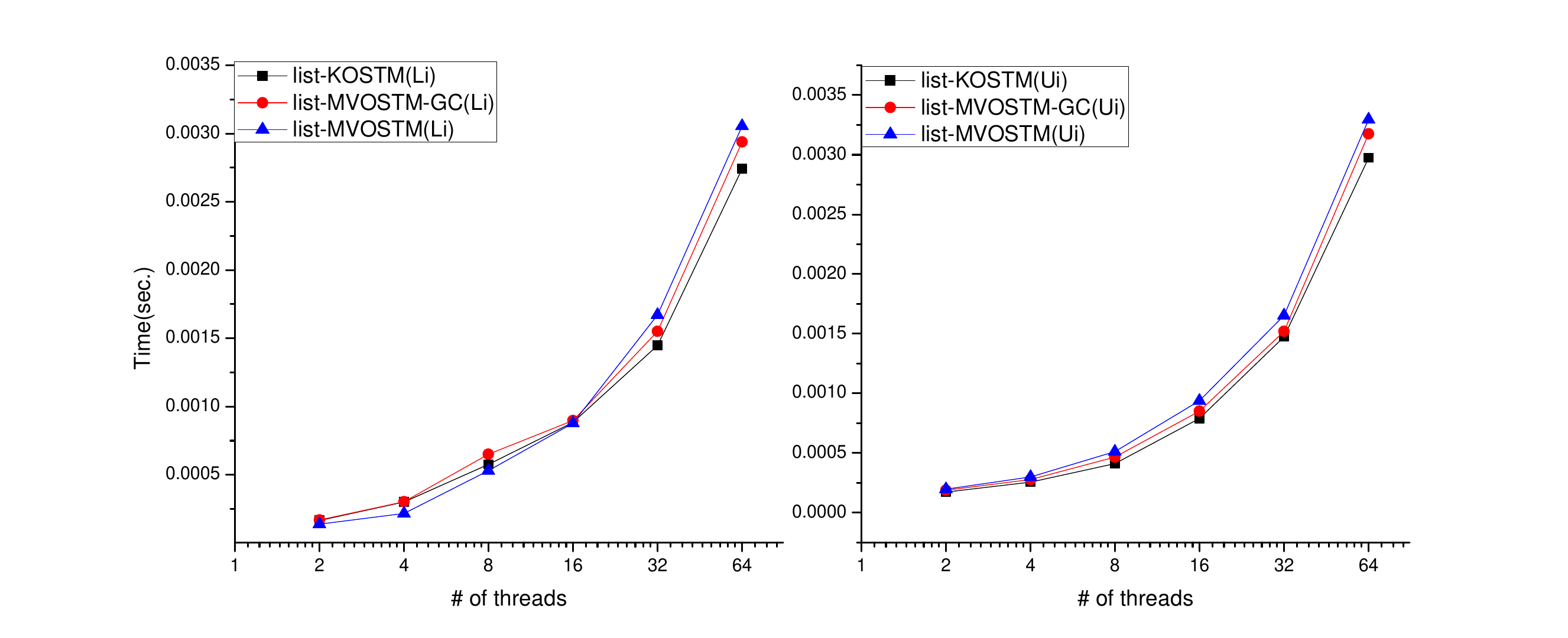}
	\centering
	\caption{Performance comparisons of variations (\lmvotm and \lkotm) of \lmvotm}\label{fig:KlistGC}
\end{figure}

\noindent
\textbf{Finite version \mvotm (\kotm):} Another technique to efficiently use memory is to restrict the number of versions rather than using unbounded number of versions, without compromising on the benefits of multi-version. \kotm, keeps at most $K$  versions by deleting the oldest version when $(K+1)^{th}$ version is created by a validated transaction. That is, once a key reaches its maximum number of versions count $K$, no new version is created in a new memory location rather new version overrides the version with the oldest time stamp. 
To find the ideal value of $K$ such that performance as compared to \mvotmgc does not degrade or can be increased, we perform experiments on two settings one on high contention high workload ($C1$) and other on low contention low workload ($C2$).
\begin{figure}
	\centering
	\captionsetup{justification=centering}
	\includegraphics[width=15cm]{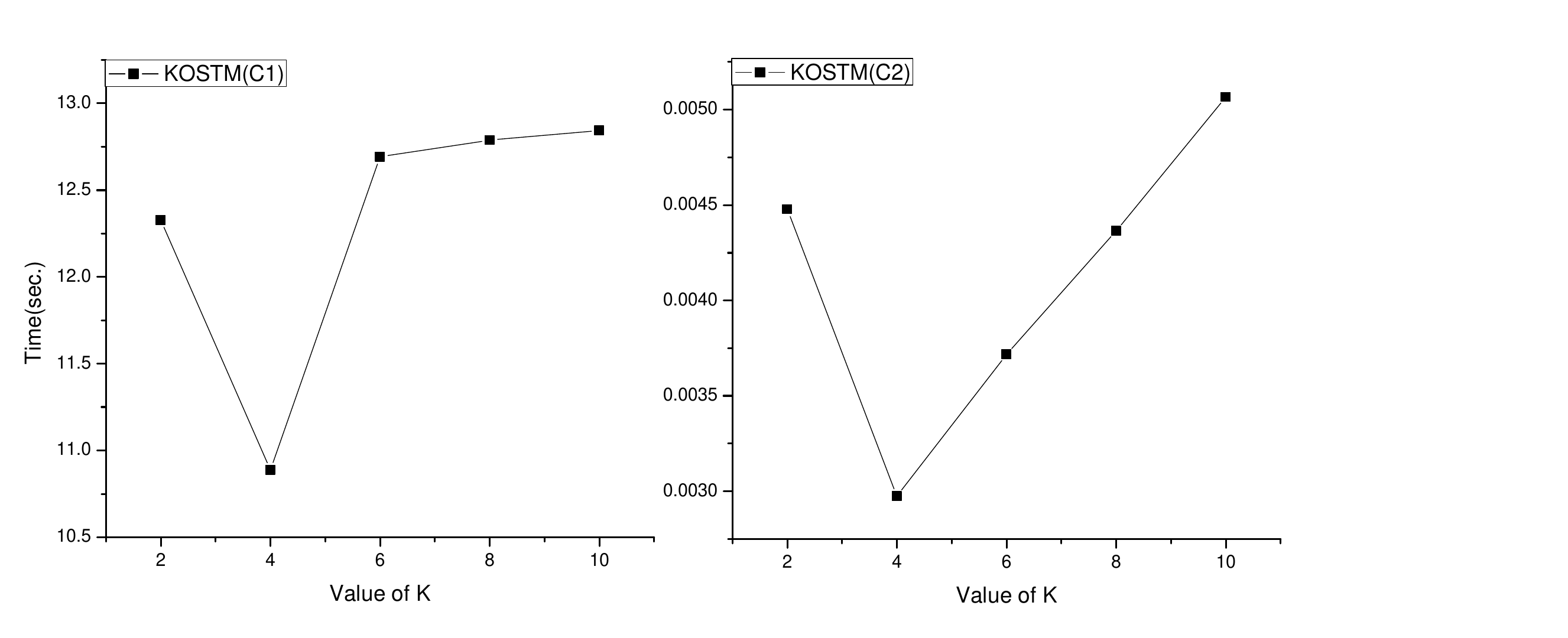}
	\centering
	\caption{Optimal value of K as 4}\label{fig:opmitalK}
\end{figure}

Under high contention $C1$, each thread spawns over 100  different transactions and each transaction performs insert (45\%)/delete (45\%)/lookup (10\%) operations over 50 random keys. And under low contention $C2$, each thread spawns over one transaction and each transaction performs insert (45\%)/delete (45\%)/lookup (10\%) operations over 1000 random keys. Our experiments as shown in \figref{opmitalK} give the best value of $K$ as 4 under both contention settings. These experiments are performed for \lmvotm and similar experiments can be performed for \hmvotm.
Between these two variations, \kotm performs better than \mvotmgc as shown in \figref{KHTGC} and \figref{KlistGC}. \hkotm helps to achieve a performance speedup of 1.22 and 1.15 for workload type $W1$ and speedup of 1.15 and 1.08 for workload type $W2$ as compared to \hmvotm and \hmvotmgc respectively. Whereas \lkotm (with four versions) gives a speedup of 1.1, 1.07 for workload type $W1$ (8\% insert, 2\% delete and 90\% look Up) and speedup of 1.25, 1.13 for workload type $W2$ (45\% insert, 45\% delete and 10\% lookUp) over the \lmvotm and \lmvotmgc respectively.

\vspace{1mm}
\noindent
\textbf{Memory Consumption by \mvotmgc and \kotm:} As depicted above \kotm performs better than \mvotmgc. Continuing the comparison between the two variations of \mvotm we chose another parameter as memory consumption. Here we test for the memory consumed by each variation algorithms in creating a version of a key. We count the total versions created, where creating a version increases the counter value by 1 and deleting a version decreases the counter value by 1. Our experiments, as shown in \figref{memoryC}, under the same contentions $C1$ and $C2$ show that \kotm needs less memory space than \mvotmgc. These experiments are performed for \lmvotm and similar experiments can be performed for \hmvotm. 
\vspace{-.7cm}
\begin{figure}
	\centering
	\captionsetup{justification=centering}
	\includegraphics[width=13cm]{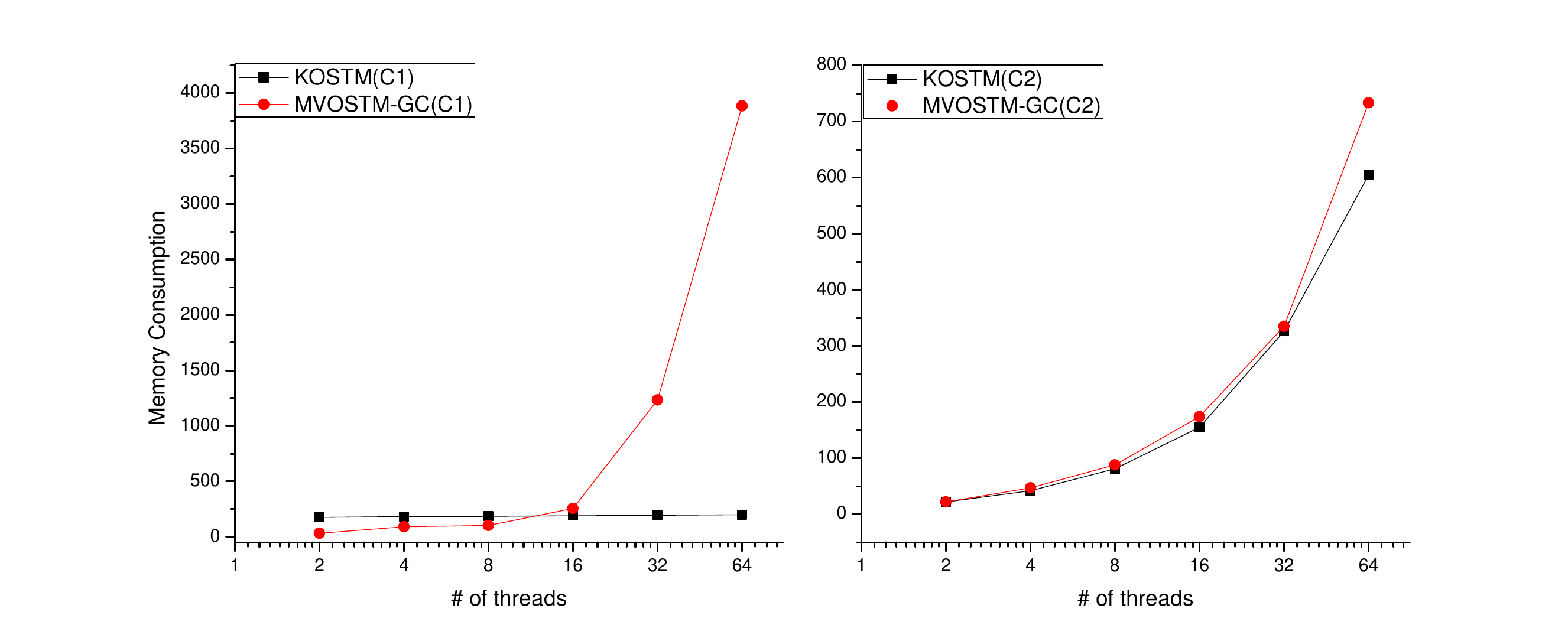}
	\centering
	\caption{Memory Consumption}\label{fig:memoryC}
\end{figure}


\cmnt{

We have implemented garbage collection for MVOSTM for both hash-based and linked-list based approaches. Each transaction, in the beginning, does its entry in a global list named as ALTL (All live transactions list), which keeps track of all the live transactions. Under the
optimistic approach of STM, each transaction performs its updates in the shared memory in the update execution phase. Each transaction in this phase performs some validations and if all validations are completed successfully a version of that key is created by that transaction. When a transaction goes to create a version of a key in the shared memory, it checks for the least time stamp live transaction present in the ALTL. If the current transaction is the one with least time-stamp present in ALTL, then this transaction deletes all the older versions of the current key and create a version of its own. If current transaction is not the least time-stamp live transaction then it doesn't do any garbage collection. In this way, we ensure each transaction performs garbage collection on the keys it is going to create a version on. Once the transaction, changes its state to commit, it removes its entry from the ALTL. As shown in all the graphs above MVOSTM with garbage collection (HT-MVOSTM-GC and list MVOSTM-GC) performs better than MVOSTM without garbage collection.
}
\section{Conclusion and Future Work}
\label{sec:confu}
Multi-core systems have become very common nowadays. Concurrent programming using multiple threads has become necessary to utilize all the cores present in the system effectively. But concurrent programming is usually challenging due to synchronization issues between the threads. 

In the past few years, several STMs have been proposed which address these synchronization issues and provide greater concurrency. STMs hide the synchronization and communication difficulties among the multiple threads from the programmer while ensuring correctness and hence making programming easy. Another advantage of STMs is that they facilitate compositionality of concurrent programs with great ease. Different concurrent operations that need to be composed to form a single atomic unit is achieved by encapsulating them in a single transaction.
 
In literature, most of the STMs are \rwtm{s} which export read and write operations. To improve the performance, a few researchers have proposed \otm{s} \cite{HerlihyKosk:Boosting:PPoPP:2008,Hassan+:OptBoost:PPoPP:2014, Peri+:OSTM:Netys:2018} which export higher level objects operation such as \tab{} insert, delete etc. By leveraging the semantics of these higher level \op{s}, these STMs provide greater concurrency. On the other hand, it has been observed in STMs and databases that by storing multiple versions for each \tobj in case of \rwtm{s} provides greater concurrency \cite{perel+:2010:MultVer:PODC,Kumar+:MVTO:ICDCN:2014}.

\ignore{
	
\todo{SK: Remove this para}
So, we get inspired from literature and proposed a new STM as \mvotm which is the combination of both of these ideas (multi versions in OSTMs). It provides greater concurrency and reduces number of abort while maintaining multiple versions corresponding to each key. \mvotm ensures compositionality by making the transactions atomic. In addition to that, we develop garbage collection for MVOSTM (MVOSTM-GC) to delete unwanted versions corresponding to the each keys to reduce traversal overhead. \mvotm satisfies \ccs{} as \emph{opacity}.

}

This paper presents the notion of multi-version object STMs and compares their effectiveness with single version object STMs and multi-version read-write STMs. We find that multi-version object STM provides a significant benefit over both of these for different types of workloads. Specifically, we  have  evaluated the effectiveness of MVOSTM  for the list and \tab{} data structure as \lmvotm and \hmvotm. Experimental results of \lmvotm provide almost two to twenty fold speedup over existing state-of-the-art list based STMs (Trans-list,  Boosting-list,  NOrec-list,  list-MVTO,  and  list-OSTM). Similarly, \hmvotm shows a  significant performance gain of almost two to nineteen times better than existing state-of-the-art \tab{} based STMs (ESTM, RWSTMs, HT-MVTO, and HT-OSTM). 

\hmvotm and \lmvotm and use unbounded number of versions for each key. To limit the number of versions, we develop two variants for both \tab and list data-structures: (1) A garbage collection method in \mvotm to delete the unwanted versions of a key, denoted as \mvotmgc. (2) Placing a limit of $k$ on the number versions in \mvotm, resulting in \kotm. Both these variants gave a performance gain of over 15\% over \mvotm.



\bibliographystyle{splncs}
\bibliography{citations}

\begin{thebibliography}{10}

\bibitem{Dalessandro+:NoRec:PPoPP:2010}
Dalessandro, L., Spear, M.F., Scott, M.L.:
\newblock {NOrec: Streamlining STM by Abolishing Ownership Records}.
\newblock In Govindarajan, R., Padua, D.A., Hall, M.W., eds.: PPOPP, ACM (2010)
   67--78

\bibitem{Felber+:ElasticTrans:2017:jpdc}
Felber, P., Gramoli, V., Guerraoui, R.:
\newblock {Elastic Transactions}.
\newblock J. Parallel Distrib. Comput. \textbf{100}(C) (February 2017)
  103--127

\bibitem{HerlihyKosk:Boosting:PPoPP:2008}
Herlihy, M., Koskinen, E.:
\newblock Transactional boosting: a methodology for highly-concurrent
  transactional objects.
\newblock In: Proceedings of the 13th {ACM} {SIGPLAN} Symposium on Principles
  and Practice of Parallel Programming, {PPOPP} 2008, Salt Lake City, UT, USA,
  February 20-23, 2008. (2008)  207--216

\bibitem{Hassan+:OptBoost:PPoPP:2014}
Hassan, A., Palmieri, R., Ravindran, B.:
\newblock {Optimistic transactional boosting}.
\newblock In: PPoPP. (2014)  387--388

\bibitem{Peri+:OSTM:Netys:2018}
Peri, S., Singh, A., Somani, A.:
\newblock {Efficient means of Achieving Composability using Transactional
  Memory}.
\newblock NETYS '18 (2018)

\bibitem{Heller+:LazyList:PPL:2007}
Heller, S., Herlihy, M., Luchangco, V., Moir, M., III, W.N.S., Shavit, N.:
\newblock {A Lazy Concurrent List-Based Set Algorithm}.
\newblock Parallel Processing Letters \textbf{17}(4) (2007)  411--424

\bibitem{GuerKap:Opacity:PPoPP:2008}
Guerraoui, R., Kapalka, M.:
\newblock {On the {C}orrectness of {T}ransactional {M}emory}.
\newblock In: PPoPP, ACM (2008)  175--184

\bibitem{WeiVoss:TIS:2002:Morg}
Weikum, G., Vossen, G.:
\newblock {Transactional Information Systems: Theory, Algorithms, and the
  Practice of Concurrency Control and Recovery}.
\newblock Morgan Kaufmann (2002)

\bibitem{Kumar+:MVTO:ICDCN:2014}
Kumar, P., Peri, S., Vidyasankar, K.:
\newblock {A TimeStamp Based Multi-version STM Algorithm}.
\newblock In: ICDCN. (2014)  212--226

\bibitem{ZhangDech:LockFreeTW:SPAA:2016}
Zhang, D., Dechev, D.:
\newblock {Lock-free Transactions Without Rollbacks for Linked Data
  Structures}.
\newblock SPAA '16, New York, NY, USA, ACM (2016)  325--336

\bibitem{KuznetsovPeri:Non-interference:TCS:2017}
Kuznetsov, P., Peri, S.:
\newblock {Non-interference and local correctness in transactional memory}.
\newblock Theor. Comput. Sci. \textbf{688} (2017)  103--116

\bibitem{KuznetsovRavi:ConcurrencyTM:OPODIS:2011}
Kuznetsov, P., Ravi, S.:
\newblock On the cost of concurrency in transactional memory.
\newblock In: OPODIS. (2011)  112--127

\bibitem{Papad:1979:JACM}
Papadimitriou, C.H.:
\newblock {The serializability of concurrent database updates}.
\newblock J. ACM \textbf{26}(4) (1979)  631--653

\bibitem{tm-book}
Guerraoui, R., Kapalka, M.:
\newblock {Principles of Transactional Memory, Synthesis Lectures on
  Distributed Computing Theory}.
\newblock Morgan and Claypool (2010)

\bibitem{Harris:NBList:DISC:2001}
Harris, T.L.:
\newblock A pragmatic implementation of non-blocking linked-lists.
\newblock In: Distributed Computing, 15th International Conference, {DISC}
  2001, Lisbon, Portugal, October 3-5, 2001, Proceedings. (2001)  300--314

\bibitem{perel+:2010:MultVer:PODC}
Perelman, D., Fan, R., Keidar, I.:
\newblock {{O}n {M}aintaining {M}ultiple {V}ersions in {STM}}.
\newblock In: PODC. (2010)  16--25

\end{thebibliography}

\clearpage
\section*{Appendix}
\label{apn:appendix}

\appendix

\section{Detailed Pcode of MVOSTM}
\label{apn:rpcode}

\subsection{Global DS}
\begin{lstlisting}{language = C++}
struct G_node{
int G_key;
struct G_vl;
G_lock;			
node G_knext;    
};
\end{lstlisting}

\begin{lstlisting}{language = C++}
struct G_vl{
int G_ts;
int G_val;
bool G_mark; 	
/*rvl stands for return value list*/
int G_rvl[];			
vl G_vnext;    
};
\end{lstlisting}

\subsection{Local DS}

\begin{lstlisting}{language = C++}
class L_txlog{
int L_t_id;			
STATUS L_tx_status;
vector <L_rec> L_list;
find();
getList();
};
\end{lstlisting}

\begin{lstlisting}{language = C++}
class L_rec{
int L_obj_id;
int L_key; 
int L_val;		
node* L_knext, G_pred, G_curr, node;
STATUS L_op_status;		
OP_NAME L_opn;
getOpn();	getPreds&Currs();  getOpStatus();
getKey&Objid();	getVal();	   getAptCurr();
setVal();	setPreds&Currs();  setOpStatus();   setOpn();
};
enum OP_NAME = {INSERT, DELETE, LOOKUP}
enum STATUS = {ABORT = 0, OK, FAIL, COMMIT}
\end{lstlisting}

\begin{table}[H]
	\centering
	\begin{tabular}{ || m{8em} | m{9cm}|| } 
		\hline
		\textbf{Functions} & \textbf{Description} \\ 
		\hline 		\hline
		setOpn() & set method name into transaction local log\\ 
		\hline
		setVal() & set value of the key into transaction local log\\ 
		\hline
		setOpStatus() & set status of method into transaction local log\\ 
		\hline
		setPred\&Curr() & set location of $G\_pred$ and $G\_curr$ according to the node corresponding to the key into transaction local log\\ 
		\hline
		getOpn() & get method name from transaction local log\\ 
		\hline
		getVal() & get value of the key from transaction local log\\ 
		\hline
		getOpStatus() & get status of the method from transaction local log\\ 
		\hline
		getKey\&Objid() & get key and obj\_id corresponding to the method from transaction local log\\ 
		\hline
		getPred\&Curr() & get location of $G\_preds$ and $G\_currs$ according to the node corresponding to the key from transaction local log \\ 
		\hline		
	\end{tabular}
	\caption{Description of accessing transaction local log methods}
	\label{tabel:1}
\end{table}

\begin{table}[H]
	\centering
	\begin{tabular}{ || m{2.3cm} | m{2.3cm} | m{2.3cm} | m{2.3cm} | m{2.3cm}|| } 
		\hline
		\textbf{p/q} & \textbf{\npins{}} & \textbf{\npdel{}} & \textbf{\npluk{}} & \textbf{\nptc} \\ 
		\hline 		
		\textbf{\npins{}} & + & + & + & +\\
		\hline
		\textbf{\npdel{}} & + & + & + & -\\
		\hline
		\textbf{\npluk{}} & + & + & + & -\\
		\hline
		\textbf{\nptc{}} & + & - & - & -\\
		\hline
	\end{tabular}
	\caption{Commutative table}
	\label{tabel:2}
\end{table}

\begin{algorithm} 
\label{alg:init} 

\caption{STM $\init()$: This method invokes at the start of the STM system. Initialize the global counter ($\cnt$) as 1 at \Lineref{init1} and return it.}
\setlength{\multicolsep}{0pt}
\begin{algorithmic}[1]
\makeatletter\setcounter{ALG@line}{52}\makeatother	
\Procedure{STM init}{$\cnt \uparrow$}
\cmnts{Initializing the global counter}
\State $\cnt$ $\gets$ 1; \label{lin:init1}
\cmnt{
\ForAll {key $G\_k$ used by the STM System}
\State /*$T_0$ is initializing key $k$*/
\State add $\langle 0, 0, T, NULL, NULL \rangle$ to $G\_k.vl$;  \label{lin:init} 
\EndFor;
}
\State return $\langle \cnt \rangle$; 
\EndProcedure
\end{algorithmic}
	
\end{algorithm}

\begin{algorithm} 
\label{alg:begin} 
\caption{STM $begin()$: It invoked by a thread to being a new transaction $T_i$. It creates transaction local log and allocate unique id at \Lineref{begin3} and \Lineref{begin5} respectively.}
\setlength{\multicolsep}{0pt}
\begin{algorithmic}[1]
\makeatletter\setcounter{ALG@line}{88}\makeatother
\Procedure{STM begin}{$\cnt \downarrow$, $L\_t\_id \uparrow$} \label{lin:begin1}
\cmnts{Creating a local log for each transaction}\label{lin:begin2}
\State \txll $\gets$ create new \txllf;  \label{lin:begin3}
\cmnts{Getting transaction id ($L\_t\_id$) from $\cnt$}\label{lin:begin4}
\State \txll.$L\_t\_id$ $\gets$ $\cnt$;  \label{lin:begin5}
\cmnts{Incremented global counter atomically $\cnt$}\label{lin:begin6}
\State $\cnt$ $\gets$ \gi; //$\Phi_{lp}(Linearization Point)$ \label{lin:begin7}
\State return $\langle L\_t\_id \rangle$; \label{lin:begin8}
\EndProcedure\label{lin:begin9}
\end{algorithmic}
\end{algorithm}


\begin{algorithm}[H]

	\caption{STM $insert():$ Optimistically, the actual insertion will happen in the \nptc{} method. First, it will identify the node corresponding to the key in local log. If the node exists then it just update the local log with useful information like value, operation name and status for the node corresponding to the key at \Lineref{insert8}, \Lineref{insert9} and \Lineref{insert10} respectively for later use in \nptc{}. Otherwise, it will create a local log and update it.}
\setlength{\multicolsep}{0pt}
	\label{algo:insert}
	\begin{algorithmic}[1]
	\makeatletter\setcounter{ALG@line}{97}\makeatother
	\Procedure{STM insert}{$L\_t\_id \downarrow, L\_obj\_id \downarrow, L\_key \downarrow, L\_val \downarrow$}\label{lin:insert1}
	\cmnts{First identify the node corresponding to the key into local log using $find()$ funciton}\label{lin:insert2}
		\If{$(!$\txlfind$)$}\label{lin:insert3}
		\cmnts{Create local log record and append it into increasing order of keys} \label{lin:insert4}
		\State $L\_rec$ $\gets$ create new $L\_rec \langle L\_obj\_id, L\_key \rangle$; \label{lin:insert5}
		\EndIf \label{lin:insert6}
		\cmnts{Updating the local log} \label{lin:insert7}		\State \llsval{$L\_val \downarrow$};//$\Phi_{lp} (Linearization Point)$ \label{lin:insert8}
		\State \llsopn{$INSERT \downarrow$}; \label{lin:insert9}
		\State \llsopst{$OK \downarrow$}; \label{lin:insert10}
		\State return $\langle void \rangle$; \label{lin:insert11} 
    \EndProcedure	\label{lin:insert12}
	\end{algorithmic}
\end{algorithm}

\begin{algorithm}[H]
	
	\caption{\tabspace[0.2cm] STM $lookup_{i}()$: If \npluk{} is not the first method on a particular key means if its a subsequent method of the same transaction on that key then first it will search into the local log from \Lineref{lookup3} to \Lineref{lookup14}. If the previous method on the same key of same transaction was insert or lookup (from \Lineref{lookup7} to \Lineref{lookup9}) then \npluk{} will return the value and operation status based on previous operation value and status. If the previous method on the same key of same transaction was delete (from \Lineref{lookup11} to \Lineref{lookup13}) then \npluk{} will return the value and operation status as NULL and FAIL respectively. If \npluk{} is the first method on that key (from \Lineref{lookup16} to \Lineref{lookup22}) then it will identify the location of node corresponding to the key in underlying DS with the help of \nplsls{} inside the \npcld{} method at \Lineref{lookup17}.} 
	\setlength{\multicolsep}{0pt}
		
		\label{algo:lookup}
		
		\begin{algorithmic}[1]
			\makeatletter\setcounter{ALG@line}{109}\makeatother
			\Procedure{STM lookup}{$L\_t\_id \downarrow, L\_obj\_id \downarrow, L\_key \downarrow, L\_val \uparrow, L\_op\_status \uparrow$} \label{lin:lookup1}
			\cmnts{First identify the node corresponding to the key into local log}\label{lin:lookup2}
			\If{$($\txlfind$)$} \label{lin:lookup3}
			\cmnts{Getting the previous operation's name}\label{lin:lookup4}
			\State $L\_opn$ $\gets$ \llgopn{} \label{lin:lookup5}; 
			\cmnts{If previous operation is insert/lookup then get the value/op\_status based on the previous operations value/op\_status}\label{lin:lookup6}
			\If{$(($\textup{INSERT} $=$ \textup{$L\_opn$} $)||($ \textup{LOOKUP} $=$ \textup{$L\_opn$}$))$} \label{lin:lookup7}
			
			\State $L\_val$ $\gets$ \llgval{} \label{lin:lookup8};
			\State $L\_op\_status$ $\gets$  $L\_rec.L\_getOpStatus$($L\_obj\_id \downarrow, L\_key \downarrow$) \label{lin:lookup9};
			\cmnts{If previous operation is delete then set the value as NULL and op\_status as FAIL}\label{lin:lookup10}
			\ElsIf{$($\textup{DELETE} $=$ \textup{$L\_opn$}$)$} \label{lin:lookup11}
			\State $L\_val$ $\gets$ NULL \label{lin:lookup12}; 
			\State $L\_op\_status$ $\gets$ FAIL \label{lin:lookup13}; 
			\EndIf \label{lin:lookup14}
			\Else \label{lin:lookup15}
			\cmnts{ Common function for \rvmt{}, if node corresponding to the key is not part of local log}\label{lin:lookup16}
			\State \cld{};\label{lin:lookup17}

			\EndIf \label{lin:lookup18}
			\cmnts{Update the local log} \label{lin:lookup19}
			\State \llsopn{$LOOKUP \downarrow$} \label{lin:lookup20};
			\State \llsopst{$L\_op\_status \downarrow$} \label{lin:lookup21};
			
			\State return $\langle L\_val, L\_op\_status\rangle$\label{lin:lookup22}; 
			
			\EndProcedure \label{lin:lookup23}
		\end{algorithmic}
		
	
\end{algorithm}

\begin{algorithm}[H]
	\caption{\tabspace[0.2cm] STM $delete_{i}()$ : It will work same as a \npluk{}. If it is not the first method on a particular key means if its a subsequent method of the same transaction on that key then first it will search into the local log from \Lineref{delete3} to \Lineref{delete23}. If the previous method on the same key of same transaction was insert (from \Lineref{delete77} to \Lineref{delete10}) then \npdel{} will return the value based on previous operation value and status as OK and set the value and operation as NULL and DELETE respectively. 
    If previous method on the same key of same transaction was delete (from \Lineref{delete12} to \Lineref{delete15}) then \npdel{} will return the value and operation status as NULL and FAIL respectively.
    If previous method on the same key of same transaction was lookup (from \Lineref{delete16} to \Lineref{delete21}) then \npdel{} will return the value and operation status based on the previous operation value and status. If \npdel{} is the first method on that key (from \Lineref{delete24} to \Lineref{delete30}) then it will identify the location of node corresponding to the key in underlying DS with the help of \nplsls{} inside the \npcld{} method at \Lineref{delete25}.}
	\label{algo:delete}
	\setlength{\multicolsep}{0pt}
	
	\begin{algorithmic}[1]
	\makeatletter\setcounter{ALG@line}{132}\makeatother
		\Procedure{STM delete}{$L\_t\_id \downarrow, L\_obj\_id \downarrow, L\_key \downarrow, L\_val \uparrow, L\_op\_status \uparrow$} \label{lin:delete1}
			\cmnts{First identify the node corresponding to the key into local log}\label{lin:delete2}
		\If{$($\txlfind$)$} \label{lin:delete3}
	    \cmnts{Getting the previous operation's name}\label{lin:delete4}
			\State $L\_opn$ $\gets$ \llgopn{} ;\label{lin:delete5} 
		\cmnts{If previous operation is insert then get the value based on the previous operations value and set the value and operation name as NULL and DELETE respectively}\label{lin:delete6}
		\If{$($\textup{INSERT} $=$ \textup{$L\_opn$}$)$} \label{lin:delete77}
		\State $L\_val$ $\gets$ \llgval{} \label{lin:delete7};
		
		\State \llsval{NULL $\downarrow$} \label{lin:delete8};

		\State \llsopn{DELETE $\downarrow$} \label{lin:delete9};
		
		\State $L\_op\_status$ $\gets$ OK \label{lin:delete10};
		\cmnts{If previous operation is delete then set the value as NULL}\label{lin:delete11}
		\ElsIf{$($\textup{DELETE} $=$ \textup{$L\_opn$}$)$} \label{lin:delete12}
		\State \llsval{NULL $\downarrow$} \label{lin:delete13};
		\State $L\_val$ $\gets$ NULL \label{lin:delete14}; 
		\State $L\_op\_status$ $\gets$ FAIL \label{lin:delete15}; 
		\Else \label{lin:delete16}
		
		\cmnts{If previous operation is lookup then get the value based on the previous operations value and set the value and operation name as NULL and DELETE respectively}\label{lin:delete17}
		\State $L\_val$ $\gets$ \llgval{} \label{lin:delete18}; 
		
		\State \llsval{NULL$ \downarrow$} \label{lin:delete19};

		\State \llsopn{DELETE $\downarrow$} \label{lin:delete20};
		\State $L\_op\_status$ $\gets$  $L\_rec.getOpStatus$($L\_obj\_id \downarrow$, $L\_key \downarrow$) \label{lin:delete21};
		
		\EndIf \label{lin:delete22}
		
		\Else \label{lin:delete23}

		    \cmnts{Common function for \rvmt{}, if node corresponding to the key is not part of local log}\label{lin:delete24}
			\State \cld{};		\label{lin:delete25}
		\EndIf \label{lin:delete26}
		\cmnts{Update the local log}\label{lin:delete27}
		\State \llsopn{$DELETE \downarrow$} ;\label{lin:delete28}
			\State \llsopst{$L\_op\_status \downarrow$} ;\label{lin:delete29}
			
			\State return $\langle L\_val, L\_op\_status\rangle$; \label{lin:delete30}
				
	\EndProcedure\label{lin:delete31}
	\end{algorithmic}
	
	
\end{algorithm}
\cmnt{
\begin{figure}[H]
	\captionsetup{justification=centering}
	\centerline{\scalebox{0.38}{\input{figs/pnetys5.pdf_t}}}
	\caption{Need of inserting $0^{th}$ version by \rvmt{} to satisfy opacity}
	\label{fig:mvostm8}
\end{figure}	
}
\begin{figure}[H]
	\centerline{\scalebox{0.35}{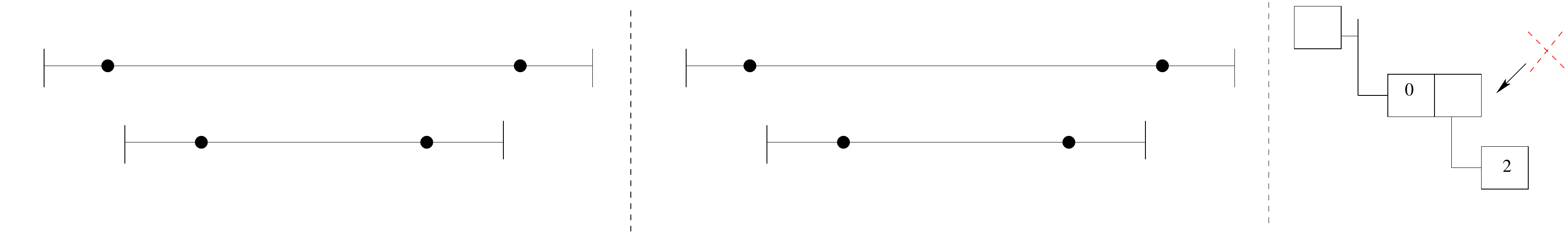}}
	\caption{Need of inserting $0^{th}$ version by \rvmt{} to satisfy opacity}
	\label{fig:mvostm8}
\end{figure}

\begin{algorithm}[H]
	\caption{\tabspace[0.2cm] $commonLu\&Del()$ : This method is invoked by a \rvmt{} (\npluk{} and \npdel{}$)$, if node corresponding to the key is not part of local log. At \Lineref{com3} it identify the $\preds$ and $\currs$ for the node corresponding to the key in underlying DS with the help of \nplsls{}. If node corresponding to the key is in \bn{} of underlying DS then (from \Lineref{com5} to \Lineref{com17}) it finds the version tuple corresponding to the key which is having the largest timestamp less than itself as $closest\_tuple$ at \Lineref{com7}. After that, it will add itself into $closest\_tuple.rvl$ at \Lineref{com9}. If identified version mark field is TRUE then it will set the $L\_op\_status$ and $L\_val$ as FAIL and NULL otherwise, OK and value of identified tuple from \Lineref{com11} to \Lineref{com17} respectively. If node corresponding to the key is in \rn{} of underlying DS then (from \Lineref{scom5} to \Lineref{scom17}) it finds the version tuple corresponding to the key which is having the largest timestamp less than itself as $closest\_tuple$ at \Lineref{scom7}. After that, it will add itself into $closest\_tuple.rvl$ at \Lineref{scom9}. If identified version mark field is TRUE then it will set the $L\_op\_status$ and $L\_val$ as FAIL and NULL otherwise, OK and value of identified tuple from \Lineref{scom11} to \Lineref{scom17} respectively. 
	If node corresponding to the key is not part of underlying DS then (from \Lineref{com19} to \Lineref{com27}) it will create the new node corresponding to the key and add it into \rn{} of underlying DS with the help of \nplslins{} at \Lineref{com20}. After that it creates the $0^{th}$ version (at \Lineref{com22}) and add itself into $0^{th}.rvl$ at \Lineref{com23}. Then, it will set the $L\_op\_status$ and $L\_val$ as FAIL and NULL respectively at \Lineref{com25} and \Lineref{com26}. Finally, it will release the lock which is acquired in \nplsls{} at \Lineref{com3} and update the local log to help the upcoming method of the same transaction on the same key. \textbf{Why do we need to create a $0^{th}$ version by \rvmt{} in \rn{}?} This will be clear by the \figref{mvostm8}, where we have two concurrent transactions $T_1$ and $T_2$. History in the \figref{mvostm8}.a) is not opaque because we can't come up with any serial order. To make it serial (or opaque) first method $lu_2(ht, k_3, NULL)$ of transaction $T_2$ have to create the $0^{th}$ version in \rn{} if its not present in the underlying DS and add itself into $0^{th}.rvl$. So in future if any lower timestamp transaction less than $T_2$ will come then that lower transaction will ABORT (in this case transaction $T_1$ is aborting in (\figref{mvostm8}.b))) because higher timestamp already present in the $rvl$ (\figref{mvostm8}.c)) of the same version. After aborting $T_1$ we will get the serial history.}
\setlength{\multicolsep}{0pt}
	
	\label{algo:commonlu&del}
	
	\begin{algorithmic}[1]
	\makeatletter\setcounter{ALG@line}{164}\makeatother
		\Procedure{commonLu\&Del}{$L\_t\_id \downarrow, L\_obj\_id \downarrow, L\_key \downarrow, L\_val \uparrow, L\_op\_status \uparrow$} \label{lin:com1} 
    \cmnts{If node corresponding to the key is not present in local log then search into underlying DS with the help of list\_lookup} \label{lin:com2}
			\State \lsls{}\label{lin:com3};
			\cmnts{ If node corresponding to the key is part of \bn{}}\label{lin:com4}
			\If{$(\bc.key = L\_key)$}\label{lin:com5}
						\cmnts{From $\bc.vls$, identify the right $version\_tuple$} \label{lin:com6}
                    \State \find($L\_t\_id \downarrow,\bc \downarrow, closest\_tuple \uparrow)$;\label{lin:com7}
                    
                    \cmnts{ Closest\_tuple is $\langle j,val,mark,rvl,vnext \rangle$} \label{lin:com8}

                    \State Adding $L\_t\_id$ into $j$'s $rvl$; \label{lin:com9}
\cmnts{If the $closest\_tuple$ mark field is TRUE then $L\_op\_status$ and $L\_val$ set as FAIL and NULL otherwise set OK and value of $closest\_tuple$ respectively}\label{lin:com10}
                \If{$(closest\_tuple.mark = TRUE)$}\label{lin:com11}
                    
                    \State $L\_op\_status$ $\gets$ FAIL \label{lin:com12};
                    \State $L\_val$ $\gets$ NULL \label{lin:com13};
                \Else\label{lin:com14}
                    \State $L\_op\_status$ $\gets$ OK;\label{lin:com15}
                    \State $L\_val$ $\gets$ $closest\_tuple.val$; \label{lin:com16}
                \EndIf    \label{lin:com17}
        \cmnts{If node corresponding to the the key is part of \rn }
			\ElsIf{$(\rc.key = L\_key)$}\label{lin:scom5}
						\cmnts{From $\rc.vls$, identify the right $version\_tuple$} \label{lin:scom6}
                    \State \find($L\_t\_id \downarrow,\rc \downarrow, closest\_tuple \uparrow)$;\label{lin:scom7}
                    
                    \cmnts{Closest\_tuple is $\langle j,val,mark,rvl,vnext \rangle$} \label{lin:scom8}
\algstore{myalg}
\end{algorithmic}
\end{algorithm}

\begin{algorithm}                     
	\begin{algorithmic} [1]                   
		\algrestore{myalg}
                    \State Adding $L\_t\_id$ into $j$'s $rvl$; \label{lin:scom9}
\cmnts{If the $closest\_tuple$ mark field is TRUE then $L\_op\_status$ and $L\_val$ set as FAIL and NULL otherwise set OK and value of $closest\_tuple$ respectively}\label{lin:scom10}
                \If{$(closest\_tuple.mark = TRUE)$}\label{lin:scom11}
                    
                    \State $L\_op\_status$ $\gets$ FAIL \label{lin:scom12};
                    \State $L\_val$ $\gets$ NULL \label{lin:scom13};
                \Else\label{lin:scom14}
                    \State $L\_op\_status$ $\gets$ OK;\label{lin:scom15}
                    \State $L\_val$ $\gets$ $closest\_tuple.val$; \label{lin:scom16}
                \EndIf    \label{lin:scom17}

        \Else\label{lin:com18}
        		\cmnts{If node corresponding to the key is not part of \rn as well as \bn then create the node into \rn with the help of list\_Ins()}\label{lin:com19}
		        \State \lslins{$RL \downarrow$};\label{lin:com20}
		        \cmnts{Insert the $0^{th}$ version tuple} \label{lin:com21}
		        \State insert $v\_tuple \langle 0,NULL,T,NULL,NULL \rangle$ into $node.vl$ in the increasing order;	 \label{lin:com22}
                  \State Adding $L\_t\_id$ into $0^{th}.rvl$;\label{lin:com23}
                \cmnts{Setting $L\_op\_status$ and $L\_val$ as FAIL and NULL because its reading from the marked version, which is TRUE}\label{lin:com24}    
                \State    $L\_op\_status$ $\gets$ FAIL ;\label{lin:com25}
                    \State $L\_val$ $\gets$ NULL ;\label{lin:com26}
			\EndIf\label{lin:com27}
			\cmnts{Releasing the locks in increasing order}\label{lin:com28}
			  	    \State releasePred\&CurrLocks($ \preds \downarrow$, $ \currs \downarrow$);
\cmnts{Create local log record and append it into increasing order of keys}\label{lin:com31}						\State $L\_rec$ $\gets$ Create new $L\_rec\langle L\_obj\_id, L\_key \rangle$\label{lin:com32}; 
						\State \llsval{$L\_val \downarrow$} \label{lin:com33}
						\State \llspc{} ;

			\State return $\langle L\_val, L\_op\_status\rangle$;\label{lin:com34}
				
	\EndProcedure 
	\end{algorithmic}
	
	
\end{algorithm}


\begin{algorithm}[H]
	\caption{\tabspace[0.2cm] STM $tryC()$ : The actual effect of \upmt{s} (\npins{} and \npdel{}) will take place in \nptc{} method. From \Lineref{tryc5} to \Lineref{tryc15} will identify and validate the $\preds$ and $\currs$ of each \upmt{} of same transaction. At \Lineref{tryc9} it will validate if there exist any higher timestamp transaction in the $rvl$ of the $closest\_tuple$ of $\bc$ then returns ABORT at \Lineref{tryc11}. Same as at \Lineref{stryc9} it will validate if there exist any higher timestamp transaction in the $rvl$ of the $closest\_tuple$ of $\rc$ then returns ABORT at \Lineref{stryc11}. Otherwise it will perform the above steps for remaining \upmt{s}. On successful validation of all the \upmt{s}, the actually effect will be taken place from \Lineref{tryc17} to \Lineref{tryc43}. If the \upmt{} is insert and node corresponding to the key is part of \bn{} then it creates the new version tuple and add it in increasing order of version list from \Lineref{tryc22} to \Lineref{tryc24}. If node corresponding to the key is part of \rn{} then it adds the same node in the \bn{} as well with the help of \nplslins{} at \Lineref{stryc24} and creates the new version tuple and add it in increasing order of version list from \Lineref{stryc23} to \Lineref{stryc25}. Otherwise it will create the node and insert it into \bn{} with the help of \nplslins{} and insert the version tuple from \Lineref{tryc25} to \Lineref{tryc29}. If the \upmt{} is delete and node corresponding to the key is part of \bn{} then it creates the new version tuple and set its mark field as TRUE and add it in increasing order of version list from \Lineref{tryc31} to \Lineref{stryc35}. After successful completion of each \upmt{}, it will validate the $\preds$ and $\currs$ of upcoming \upmt{} of the same transaction with the help of \npintv{} at \Lineref{tryc42}. Eventually, it will release all the locks at \Lineref{tryc45} in the same order of lock acquisition.} 
	\label{algo:trycommit}
	\setlength{\multicolsep}{0pt}
	\begin{algorithmic}[1]
\makeatletter\setcounter{ALG@line}{213}\makeatother
		\Procedure{STM tryC}{$L\_t\_id \downarrow, L\_tx\_status \uparrow$} \label{lin:tryc1}

		\cmnts{Get the local log list corresponding to each transaction which is in increasing order of keys}\label{lin:tryc2}
		\State $L\_list$ $\gets$ $L\_txlog.getList$($L\_t\_id \downarrow$) \label{lin:tryc3};
\cmnts{Identify the new $\preds$ and $\currs$ for all update methods of a transaction and validate it}\label{lin:tryc4}
		\While{$(\textbf{$L\_rec_{i} \gets \textup{next}(L\_list$}))$} \label{lin:tryc5}
		\State ($L\_key, L\_obj\_id$) $\gets$ \llgkeyobj{} \label{lin:tryc6};
		\cmnts{Identify the new $G\_pred$ and $G\_curr$ location with the help of list\_lookup()}\label{lin:tryc7}
		\State \lsls{$COMMIT \downarrow$} \label{lin:tryc8};
			
		\If {$((\bc.key = L\_key) \& (\checkv(L\_t\_id \downarrow,\bc \downarrow) = FALSE))$}\label{lin:tryc9}
\State Unlock all the variables;\label{lin:tryc10}
\State return $ABORT$;\label{lin:tryc11}
\ElsIf {$((\rc.key = L\_key) \& (\checkv(L\_t\_id \downarrow,\rc \downarrow) = FALSE))$}\label{lin:stryc9}
\State Unlock all the variables;\label{lin:stryc10}
\State return $ABORT$;\label{lin:stryc11}

\EndIf;\label{lin:tryc12}

		
		
			\cmnts{Update the log entry} \label{lin:tryc13}

		\State \llspc{} \label{lin:tryc14};
		\EndWhile \label{lin:tryc15}
	\cmnts{Get each update method one by one and take effect in underlying DS}\label{lin:tryc16}
		\While{$(\textbf{$L\_rec_{i} \gets \textup{next}(L\_list$}))$} \label{lin:tryc17}
		\State ($L\_key, L\_obj\_id$) $\gets$ \llgkeyobj{} \label{lin:tryc18};
		\cmnts{Get the operation name from local log record}\label{lin:tryc19}
		\State $L\_opn$ $\gets$ $(L\_rec)_{i}$.$L\_opn$ \label{lin:tryc20};
		
					\cmnts{Modify the $\preds$ and $\currs$ for the consecutive update methods which are working on overlapping zone in lazy-list}\label{lin:tryc41}
		\State intraTransValdation($L\_rec_{i} \downarrow$, $\preds \uparrow$, $\currs \uparrow$) \label{lin:tryc42};
\cmnts{If operation is insert then after successful completion of it node corresponding to the key should be part of \bn}\label{lin:tryc21}
\algstore{myalg}
\end{algorithmic}
\end{algorithm}

\begin{algorithm}                     
	\begin{algorithmic} [1]                   
		\algrestore{myalg}
		\If{$($\textup{INSERT} $=$ \textup{$L\_opn$}$)$} \label{lin:tryc22}
		\If{$(\bc.key) = L\_key)$} \label{lin:tryc23}
	\State insert $v\_tuple \langle L\_t\_id,val,F,NULL,NULL \rangle$ into $G\_curr.vl$ in the increasing order;	\label{lin:tryc24}
	
	\ElsIf{$(\rc.key) = L\_key)$} \label{lin:stryc23}
	\State \lslins{$RL\_BL \downarrow$} \label{lin:stryc24}
	\State insert $v\_tuple \langle L\_t\_id,val,F,NULL,NULL \rangle$ into $G\_curr.vl$ in the increasing order;	\label{lin:stryc25}
	
		\Else \label{lin:tryc25}
		\cmnts{If node corresponding to the key is not part underlying DS then create the node with the help of list\_Ins() and insert it into \bn}\label{lin:tryc26}
		\State \lslins{$BL \downarrow$} \label{lin:tryc27};
		\State insert $v\_tuple \langle L\_t\_id,val,F,NULL,NULL \rangle$ into $node.vl$ in the increasing order;	\label{lin:tryc28}
		\EndIf \label{lin:tryc29}
		\cmnts{If operation is delete then after successful completion of it node corresponding to the key should part of \rn only}\label{lin:tryc30}
		\ElsIf{$($\textup{DELETE} $=$ $L\_opn)$} \label{lin:tryc31}
	\cmnts{If node corresponding to the key is part of \bn}\label{lin:tryc32}

		\If{$(\bc.key) = L\_key)$} \label{lin:tryc33}
			\State insert $v\_tuple \langle L\_t\_id,NULL,T,NULL,NULL \rangle$ into $G\_curr.vl$ in the increasing order;	\label{lin:tryc34}
			\State \lsldel{} \label{lin:stryc35};
\cmnt{
		\Else \label{lin:tryc35}
		\State    /*If node corresponding to the key is not part underlying DS then create the node with the help of list\_Ins() */\label{lin:tryc36}
		\State \lslins{$RL \downarrow$};\label{lin:tryc37}
		
		\State insert $v\_tuple \langle L\_t\_id,NULL,T,NULL,NULL \rangle$ into $node.vl$ in the increasing order;	\label{lin:tryc38}
		
}	
		\EndIf \label{lin:tryc39}
	
		\EndIf \label{lin:tryc40}

		\EndWhile \label{lin:tryc43}
		\cmnts{Release all the locks in increasing order}\label{lin:tryc44}
		\State \rlsol{} \label{lin:tryc45};  
		\cmnts{Set the transaction status as OK}\label{lin:tryc46}
		\State $L\_tx\_status$ $\gets$ OK \label{lin:tryc47};
%
		\State return $\langle L\_tx\_status\rangle$\label{lin:tryc48};
		\EndProcedure \label{lin:tryc49}
	\end{algorithmic}
\end{algorithm}

	
	\begin{algorithm}[H]
		\caption{\tabspace[0.2cm] \dell() : Delete a node from blue link in underlying hash table at location corresponding to $\preds$ \& $\currs$.}
		\label{algo:lsldelete}
		\setlength{\multicolsep}{0pt}
			\begin{algorithmic}[1]
				\makeatletter\setcounter{ALG@line}{265}\makeatother
				\Function{list\_del}{$\preds \downarrow, \currs \downarrow$} \label{lin:lsldel1}
				
				\cmnts{Update the blue links}
				\State $\bp$.\bn $\gets$ $\bc$.\bn \label{lin:lsldel3};
				\State return $\langle void \rangle$;
				\EndFunction \label{lin:lsldel4}
			\end{algorithmic}
			
	\end{algorithm}

\begin{algorithm}[H]
	\caption{list\_lookup() : This method is called by \rvmt{} and \upmt{}. It finds the location of the node corresponding to the key in underlying DS from \Lineref{lslsearch5} to \Lineref{slslsearch15}. First it identifies the node in \bn{} (from \Lineref{lslsearch5} to \Lineref{lslsearch15}) then in \rn{} (from \Lineref{slslsearch9} to \Lineref{slslsearch15}). After finding the appropriate location of the node corresponding to the key in the form of $\preds$ and $\currs$, it will acquire the locks on it at \Lineref{lslsearch17} and validate it at \Lineref{lslsearch20}.}
	\label{algo:lslsearch}
	\setlength{\multicolsep}{0pt}
	\begin{algorithmic}[1]
	\makeatletter\setcounter{ALG@line}{270}\makeatother	
		\Procedure{list\_lookup}{$L\_obj\_id \downarrow, L\_key \downarrow, G\_preds[] \uparrow, G\_currs[] \uparrow$} \label{lin:lslsearch1}

	\cmnts{By default setting the $L\_op\_status$ as RETRY}\label{lin:lslsearch2}
	    \State STATUS $L\_op\_status$ $\gets$ RETRY; \label{lin:lslsearch3}
	 \cmnts{Identify the \preds and \currs for node corresponding to the key if $L\_op\_status$ is RETRY}\label{lin:lslsearch4}   
		\While{($L\_op\_status$ = \textup{RETRY})} \label{lin:lslsearch5}
		\cmnts{Get the head of the bucket in chaining hash-table with the help of $L\_obj\_id$ and $L\_key$}\label{lin:lslsearch6}
		\State $G\_head$ $\gets$ \glslhead \label{lin:lslsearch7};
		\cmnts{Initialize $\bp$ to head}\label{lin:lslsearch8}
		\State $\bp$ $\gets$ $G\_head$ \label{lin:lslsearch9}; 
		\cmnts{Initialize $\bc$ to $\bp.\bn$}\label{lin:lslsearch10}
		\State $\bc$ $\gets$ $\bp.\bn$ \label{lin:lslsearch11};
		\cmnts{Searching node corresponding to the key into \bn}
		\While{$((\bc.key) < L\_key)$} \label{lin:lslsearch12}
		\State $\bp$ $\gets$ $\bc$ \label{lin:lslsearch13};
				
		\State $\bc$ $\gets$ $\bc.\bn$ \label{lin:lslsearch14};
			
		\EndWhile \label{lin:lslsearch15}
		
		\cmnts{Initialize $\rp$ to head}\label{lin:slslsearch8}
		\State $\rp$ $\gets$ $\bp$ \label{lin:slslsearch9}; 
		\cmnts{Initialize $\rc$ to $\rp.\rn$}\label{lin:slslsearch10}
		\State $\rc$ $\gets$ $\rp.\rn$ \label{lin:slslsearch11};
		\cmnts{Searching node corresponding to the key into \rn}
		\While{$((\rc.key) < L\_key)$} \label{lin:slslsearch12}
		\State $\rp$ $\gets$ $\rc$ \label{lin:slslsearch13};
				
		\State $\rc$ $\gets$ $\rc.\rn$ \label{lin:slslsearch14};
			
		\EndWhile \label{lin:slslsearch15}

		\cmnts{Acquire the locks on increasing order of keys}\label{lin:lslsearch16}
	    \State acquirePred\&CurrLocks($ \preds \downarrow$, $ \currs \downarrow$); \label{lin:lslsearch17}
\cmnts{Method validation to identify the changes done by concurrent conflicting method}\label{lin:lslsearch19}
		\State methodValidation($\preds$ $\downarrow$, $\currs$ $\downarrow$, $L\_op\_status \uparrow$)\label{lin:lslsearch20};	
\cmnts{If $L\_op\_status$ is RETRY then release all the locks}		\label{lin:lslsearch21}
		\If{(($L\_op\_status$ = \textup{RETRY}))} \label{lin:lslsearch22}
	    \State releasePred\&CurrLocks($ \preds \downarrow$, $ \currs \downarrow$);
	    \EndIf \label{lin:lslsearch25}
			
		\EndWhile \label{lin:lslsearch26}
		
		\State return $\langle G\_preds[], G\_currs[]\rangle$ \label{lin:lslsearch27};
	
	\EndProcedure \label{lin:lslsearch28}
	\end{algorithmic}
\end{algorithm}

	
	\begin{algorithm}[H]
		\caption{\tabspace[0.2cm] acquirePred\&CurrLocks() : acquire all locks taken during \nplsls{}.}
		\label{algo:acquirepreds&currs}
		\setlength{\multicolsep}{0pt}
			\begin{algorithmic}[1]
				\makeatletter\setcounter{ALG@line}{306}\makeatother
				\Function{acquirePred\&CurrLocks}{$ \preds \downarrow$, $ \currs \downarrow$}
				\State $\bp$.$\texttt{lock()}$;
				\State $\rp$.$\texttt{lock()}$;
				\State $\rc$.$\texttt{lock()}$;
				\State $\bc$.$\texttt{lock()}$;
				\State return $\langle void \rangle$;        
				\EndFunction
			\end{algorithmic}
			
	\end{algorithm}
	
	
	\begin{algorithm}[H]
		\caption{\tabspace[0.2cm] releasePred\&CurrLocks() : Release all locks taken during \nplsls{}.}
		\label{algo:releasepreds&currs}
		\setlength{\multicolsep}{0pt}
			\begin{algorithmic}[1]
				\makeatletter\setcounter{ALG@line}{313}\makeatother
				\Function{releasePred\&CurrLocks}{$ \preds \downarrow$, $ \currs \downarrow$}
				\State $\bp$.$\texttt{unlock()}$\label{lin:rpandc};//$\Phi_{lp}$ 
				\State $\rp$.$\texttt{unlock()}$;
				\State $\rc$.$\texttt{unlock()}$;
				\State $\bc$.$\texttt{unlock()}$;
				\State return $\langle void \rangle$;        
				\EndFunction
			\end{algorithmic}
			
	\end{algorithm}


\begin{algorithm}[H]
	\caption{\tabspace[0.2cm] list\_Ins(): This method is called by the \rvmt{} and \upmt{}. Color of preds \& currs depicts the red or blue node.}
	\label{algo:lslins}
		\setlength{\multicolsep}{0pt}
	\begin{algorithmic}[1]
\makeatletter\setcounter{ALG@line}{320}\makeatother
		\Procedure{list\_Ins}{$\preds \downarrow, \currs \downarrow, list\_type \downarrow, node \uparrow$} \label{lin:lslins1}

\cmnts{Inserting the node from redlist to bluelist}
				\If{$((list\_type)$ $=$ $($\textcolor{black}{$RL$}$\_$\textcolor{black}{$BL$}$))$} \label{lin:lslins2}
				
				\State $\rc$.\bn $\gets$ $\bc$ \label{lin:lslins4};
				\State $\bp$.\bn $\gets$ $\rc$ \label{lin:lslins5};
				\cmnts{Inserting the node into redlist only}
				\ElsIf{$((list\_type$) $=$ \textcolor{black}{$RL$}$)$} \label{lin:lslins6}
				\State node = Create new node()
				\label{lin:lslins7};
						\cmnts{After created the node acquiring the lock on it}\label{lin:lslins4}
		\State node.lock();\label{lin:lslins5}
		
				\State node.\rn $\gets$ $\rc$ \label{lin:lslins9};
				\State $\rp$.\rn $\gets$ node \label{lin:lslins10};
				
				\Else \label{lin:lslins11}
				\cmnts{Inserting the node into red as well as blue list}
				\State node = new node() \label{lin:lslins12}; 
				\cmnts{After creating the node acquiring the lock on it}
				\State node.lock();
				\State node.\rn $\gets$ $\rc$ \label{lin:lslins13};
				\State node.\bn $\gets$ $\bc$ \label{lin:lslins14};
				
				\State $\rp$.\rn $\gets$ node \label{lin:lslins15};
				
				\State $\bp$.\bn $\gets$ node \label{lin:lslins16};
				\EndIf \label{lin:lslins17}
				\State return $\langle node \rangle$;
			
\cmnt{

        \State    /*Inserting the new node corresponding to the key into Underlying DS/\label{lin:lslins2}
		\State node = Create new node() \label{lin:lslins3}; 
		\State    /*After created the node acquiring the lock on it*/\label{lin:lslins4}
		\State node.lock();\label{lin:lslins5}
		\State /*Adding the new node at appropriate location (in increasing order of the keys) with the help of $G\_pred$ and $G\_curr$*/\label{lin:lslins6}
		\State node.knext $\gets$ $G\_curr$ \label{lin:lslins7};
		\State $G\_pred$.knext $\gets$ node \label{lin:lslins8};
		\State return $\langle node \rangle$;\label{lin:lslins9}}
		\EndProcedure \label{lin:lslins10}
	\end{algorithmic}
\end{algorithm}


\begin{algorithm}  
\label{alg:select}
\caption{$\find()$: This method is called by \rvmt{} and \upmt{} to identify a  $closest\_tuple$ $\langle j,val,mark,rvl,vnext \rangle$ created by the transaction $T_j$ with the largest timestamp smaller than $L\_t\_id$ from \Lineref{findlts5} to \Lineref{findlts10}.}
		\setlength{\multicolsep}{0pt}
\begin{algorithmic}[1]
\makeatletter\setcounter{ALG@line}{344}\makeatother
\Procedure{find\_lts}{$L\_t\_id \downarrow, \currs \downarrow, closest\_tuple \uparrow$} \label{lin:findlts1}
\cmnts{Initialize $closest\_tuple$} \label{lin:findlts2}
\State $closest\_tuple = \langle 0,NULL,F,NULL,NULL \rangle$;\label{lin:findlts3}
\cmnts{For all the version of $\currs$ identify the largest timestamp less than $L\_t\_id$} \label{lin:findlts4}
\ForAll {$\langle p,val,mark,rvl,vnext \rangle \in \currs.vl$} \label{lin:findlts5}
\If {$(p < L\_t\_id)$ and $(closest\_tuple.ts < p)$} \label{lin:findlts6}
\cmnts{Assign closest tuple as $\langle p,val,mark,rvl,vnext \rangle$, if any version tuple is having largest timestamp less then L\_t\_id exist} \label{lin:findlts7}
\State $closest\_tuple = \langle p,val,mark,rvl,vnext \rangle$; \label{lin:findlts8}
\EndIf\label{lin:findlts9}
\EndFor\label{lin:findlts10}
\State return $\langle closest\_tuple\rangle$; \label{lin:findlts11}
\EndProcedure\label{lin:findlts12}
\end{algorithmic}
\end{algorithm}

\begin{figure} [tbph]
	\captionsetup{justification=centering}
	\centerline{\scalebox{0.5}{\input{figs/disc2.pdf_t}}}
	\caption{Validation by $check\_version$}
	\label{fig:mvostmm12}
\end{figure}

\cmnt{
\begin{figure} [tbph]
	\centerline{\scalebox{0.5}{\input{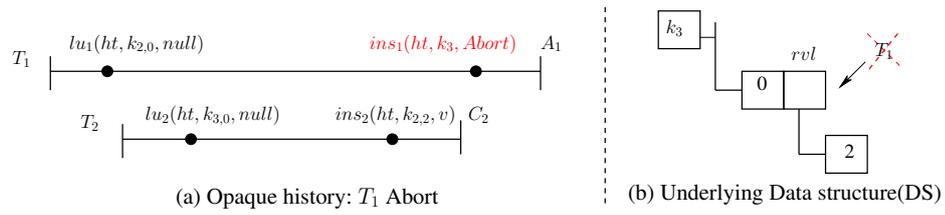}}}
	\caption{Validation by $check\_version$}
	\label{fig:mvostm12}
\end{figure}
}


\begin{algorithm}  
\label{alg:checkVersion} 
\caption{$\checkv()$: This method is called by the \nptc{}. First it will find the $closest\_tuple$ $\langle j,val,mark,rvl,vnext \rangle$ created by the transaction $T_j$ with the largest timestamp smaller than $L\_t\_id$ at \Lineref{checkVersion3}. Then, it checks the version list to identify is there any higher timestamp already present in the $rvl$ of $closest\_tuple$ from \Lineref{checkVersion5} to \Lineref{checkVersion11}. If it presents then it will returns FALSE at \Lineref{checkVersion9} otherwise, TRUE at \Lineref{checkVersion13}. It will be more clear by the \figref{mvostmm12}.a) where second method $ins_1(ht, k_3, ABORT)$ of transaction $T_1$ will ABORT because higher transaction $T_2$ timestamp is already present in the $rvl$ of $closest\_tuple$ as $0^{th}$ version in \figref{mvostmm12}.b).}
		\setlength{\multicolsep}{0pt}
\begin{algorithmic}[1]
\makeatletter\setcounter{ALG@line}{356}\makeatother
\Procedure{check\_versions}{$L\_t\_id \downarrow,\currs \downarrow$} \label{lin:checkVersion1}
\cmnts{From $\currs.vls$, identify the correct $version\_tuple$ means identfy the tuple which is having higher time-stamp but less then it} \label{lin:checkVersion2}
\State $\find(L\_t\_id \downarrow, \currs \downarrow, closest\_tuple \uparrow)$;\label{lin:checkVersion3}	
\cmnts{Got the closest\_tuple as $\langle j,val,mark,rvl,vnext \rangle$}\label{lin:checkVersion4}
\ForAll {$T_k$ in $rvl$ of $closest\_tuple.j$}\label{lin:checkVersion5}
\cmnts{$T_k$ has already read the version created by $T_j$}\label{lin:checkVersion6}
\If {$(L\_t\_id < k)$}   \label{lin:checkVersion7}
\cmnts{If in $rvl$ of $closest\_tuple.j$, any higher time-stamp exists then $L\_t\_id$ then return $FALSE$}\label{lin:checkVersion8}
\State return $\langle FALSE\rangle$;  \label{lin:checkVersion9}
\EndIf\label{lin:checkVersion10}
\EndFor\label{lin:checkVersion11}
\cmnts{If in $rvl$ of $closest\_tuple.j$, there is no higher time-stamp exists then $L\_t\_id$ then return $TRUE$}\label{lin:checkVersion12}
\State return $\langle TRUE\rangle$;\label{lin:checkVersion13}
\EndProcedure\label{lin:checkVersion14}
\end{algorithmic}
\end{algorithm}

\begin{figure}[tbph]
	\captionsetup{justification=centering}
	\centerline{\scalebox{0.47}{\input{figs/disc1.pdf_t}}}
	\caption{Method validation}
	\label{fig:mvostm9}
\end{figure}
\cmnt{
\begin{figure}[tbph]
	\centerline{\scalebox{0.5}{\input{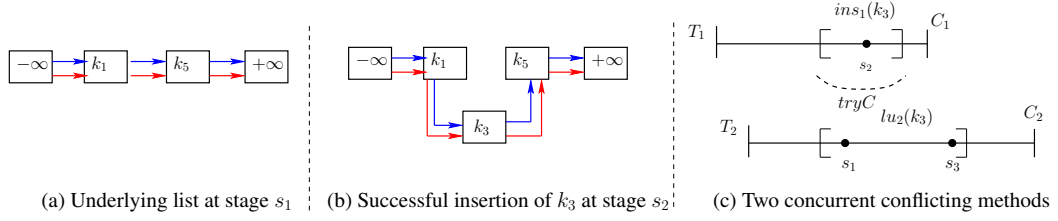}}}
	\caption{Method validation}
	\label{fig:mvostm9}
\end{figure}
}


\begin{algorithm}[H]
	\caption{\tabspace[0.2cm] methodValidation() : This method is called by the \rvmt{} and \upmt{}. It will identify the conflicts among the concurrent methods of different transactions at \Lineref{iv2}. It will be more clear by the \figref{mvostm9}, where two concurrent conflicting methods of different transactions are working on the same key $k_3$. Initially, at stage $s_1$ in \figref{mvostm9}.c) both the conflicting method optimistically (without acquiring locks) identify the same $\preds$ and $\currs$ for key $k_3$ from underlying DS in \figref{mvostm9}.a). At stage $s_2$ in \figref{mvostm9}.c), method $ins_1(k_3)$ of transaction $T_1$ acquired the lock on $\preds$ and $\currs$ and inserted the node into it (\figref{mvostm9}.b)). After successful insertion by $T_1$, $\preds$ and $\currs$ will change for $lu_2(k_3)$ at stage $s_3$ in \figref{mvostm9}.c). It will caught via method validation function at \Lineref{iv2} when $(\preds.\bn \neq G\_curr)$ for $lu_2(k_3)$. After that again it will find the new $\preds$ and $\currs$ for $lu_2(k_3)$ with the help of \nplsls{} method and eventually it will commit.}
	\label{algo:interferenceValidation}
			\setlength{\multicolsep}{0pt}
	\begin{algorithmic}[1]
\makeatletter\setcounter{ALG@line}{370}\makeatother	
		\Procedure{methodValidation}{$\preds \downarrow, \currs \downarrow, L\_op\_status \uparrow$} \label{lin:iv1}
	    \cmnts{Validating $\preds$ and $\currs$}
		\If{$((\bp.marked) || (\bc.marked) ||(\bp.\bn) \neq \bc || (\rp.\rn) \neq {\rc})$}  \label{lin:iv2}
		\cmnts{If validation fail then $L\_op\_status$ set as $RETRY$}
        \State $L\_op\_status$ $\gets$ RETRY

		\Else \label{lin:iv4}
		
	    \State $L\_op\_status$ $\gets$ OK

		\EndIf \label{lin:iv6}
				\State return $\langle L\_op\_status\rangle$ \label{lin:iv5};
	
		\EndProcedure \label{lin:iv7}
	\end{algorithmic}
\end{algorithm}


\begin{algorithm}[H]
	\caption{\tabspace[0.2cm] $L\_find()$ : This method is called by \npins{}, \rvmt{} and \upmt{}. It will check whether any method corresponding to $\left\langle L\_obj\_id, L\_key  \right\rangle$ is present in local log from \Lineref{findll4} to \Lineref{findll8}.}
	\label{algo:findInLL}
				\setlength{\multicolsep}{0pt}
	\begin{algorithmic}[1]
\makeatletter\setcounter{ALG@line}{380}\makeatother		
		\Procedure{L\_find}{$L\_t\_id \downarrow, L\_obj\_id \downarrow, L\_key \downarrow, L\_rec \uparrow$} \label{lin:findll1}
		
		
		\State $L\_list$ $\gets$ \txgllist{} \label{lin:findll3}; 
	\cmnts{Every method first identify the node corresponding to the key into local log}
		\While{$(L\_rec_{i} \gets next(L\_list))$} \label{lin:findll4}
		\cmnts{Taking one by one $L\_obj\_id$ and $L\_key$ form $L\_rec$}
		\If{$((L\_rec_{i}.first = L\_obj\_id) \& (L\_rec_{i}.sec = L\_key))$} \label{lin:findll5}

		\State return $\langle TRUE, L\_rec \rangle$ \label{lin:findll6};
		\EndIf \label{lin:findll7}
		\EndWhile \label{lin:findll8}
		\State return $\langle FALSE, NULL \rangle$ \label{lin:findll9};
		\EndProcedure \label{lin:findll10}
	\end{algorithmic}
\end{algorithm}


\begin{algorithm}[H]
	\caption{\tabspace[0.2cm] releaseOrderedLocks(): Release all locks in increasing order of their keys from \Lineref{rlock2} to \Lineref{rlock7}.}
	\label{algo:releaseorderedlocks}
			\setlength{\multicolsep}{0pt}
	\begin{algorithmic}[1]
	\makeatletter\setcounter{ALG@line}{391}\makeatother
		\Procedure{releaseOrderedLocks}{$L\_list \downarrow$} \label{lin:rlock1}
	\State    /*Releasing all the locks in increasing order of the keys */
		\While{($\textbf{$L\_rec_{i} \gets next(L\_list$})$)} \label{lin:rlock2}
		
		\State $L\_rec_{i}$.$\preds$.$\texttt{unlock()}$ \label{lin:rlock3};//$\Phi_{lp}$ 
		\State $L\_rec_{i}$.$\currs$.$\texttt{unlock()}$ \label{lin:rlock4};
		
		\EndWhile \label{lin:rlock7}
		\State return $\langle void \rangle$;
		\EndProcedure \label{lin:rlock8}
	\end{algorithmic}
			
\end{algorithm}

\begin{figure} [H]
	\captionsetup{justification=centering}
	\centerline{\scalebox{0.45}{\input{figs/disc3.pdf_t}}}
	\caption{Intra transaction validation}
	\label{fig:mvostmm11}
\end{figure}

\cmnt{
\begin{figure} [H]
	\centerline{\scalebox{0.5}{\input{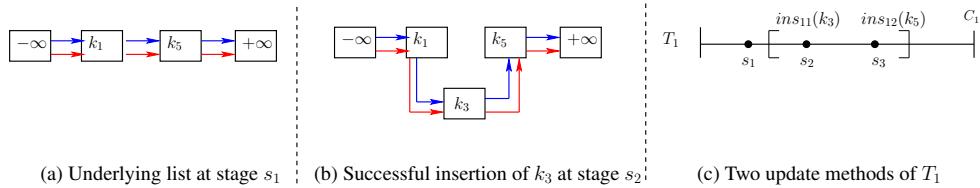}}}
	\caption{Intra transaction validation}
	\label{fig:mvostm11}
\end{figure}
}

\begin{algorithm}[H]
	\caption{\tabspace[0.2cm] intraTransValidation() : This method is called by \nptc{} only. If two $\upmt{s}$ within same transaction have at least one shared node among its recorded $\preds$ and $\currs$, in this case the previous $\upmt{}$ effect might be overwritten if the next $\upmt{}$ $\preds$ and $\currs$ are not updated according to the updates done by the previous $\upmt{}$. Thus to solve this we have intraTransValidation() after each $\upmt{}$ in $\nptc$. This will be more clear by the \figref{mvostmm11}, where two \upmt{s} of same transaction $T_1$ are $ins_{11}(k_3)$ and $ins_{12}(k_5)$ (\figref{mvostmm11}.c)). At stage $s_1$ in \figref{mvostmm11}.c) both the \upmt{s} identify the same $\preds$ and $\currs$ from underlying DS (\figref{mvostmm11}.a)). After the successful insertion done by first \upmt{} at stage $s_2$ in \figref{mvostmm11}.c), key $k_3$ is part of underlying DS (\figref{mvostmm11}.b)). At stage $s_3$ in \figref{mvostmm11}.c) if we will not update the $\preds$ and $\currs$ for $ins_{12}(k_5)$ then it will overwrite the previous method updates. To resolve this issue we are doing the intraTransValidation() after each \upmt{} to assign the appropriate $\preds$ and $\currs$ for the upcoming \upmt{} of same transaction.}
	\label{algo:povalidation}
	\setlength{\multicolsep}{0pt}
	\begin{algorithmic}[1]
\makeatletter\setcounter{ALG@line}{399}\makeatother		
		\Procedure{intraTransValidation}{$L\_rec_{i} \downarrow, \preds \uparrow, \currs \uparrow$} \label{lin:threadv1}

	\State $L\_rec.getAllPreds\&Currs(L\_rec$ $\downarrow$, $\preds$ $\uparrow$, $\currs$ $\uparrow$) \label{lin:threadv2};	
				\cmnts{if $\bp$ is marked or $\bc$ is not reachable from $\bp.\bn$ then modify the next consecutive \upmt{} $\bp$ based on previous \upmt{}}
				\If{$(($$\bp$.\textup{marked}$) ||$ $($ $\bp$.\textup{\bn}$)$ != $\bc$$))$} \label{lin:threadv3}
				\cmnts{find $k$ < $i$; such that $le_k$ contains previous update method on same bucket }
				\If{$(($$L\_rec_{k}$.\textup{opn}$)$ $=$ INSERT$)$} \label{lin:threadv4}
				
				\State $L\_rec_{i}$.$\bp$.$\texttt{unlock()}$ \label{lin:thredv4-5};
				\State $\bp$ $\gets$ ($L\_rec_{k}.\bp.\bn)$ \label{lin:threadv5};
				\State $L\_rec_{i}$.$\bp$.$\texttt{lock()}$ \label{lin:thredv5-5};		
				\Else \label{lin:threadv6}
				\cmnts{\upmt{} method $\bp$ will be previous method $\bp$}
				\State $L\_rec_{i}$.$\bp$.$\texttt{unlock()}$ \label{lin:thredv6-5};
				\State $\bp$ $\gets$ ($L\_rec_{k}$.$\bp$) \label{lin:threadv7};
				
				\State $L\_rec_{i}$.$\bp$.$\texttt{lock()}$ \label{lin:thredv7-5};
				\EndIf \label{lin:threadv8}
				
				\EndIf \label{lin:threadv9}
				\cmnts{if $\rc$ \& $\rp$ is modified by prev operation then update them also}
				\If{$($$\rp$.\textup{\rn} != $\rc$$)$} \label{lin:threadv10}
				\State $L\_rec_{i}$.$\rp$.$\texttt{unlock()}$
				\State $\rp$ $\gets$ ($L\_rec_{k}$.$\rp.\rn$)  \label{lin:threadv11}; 
				\State $L\_rec_{i}$.$\rp$.$\texttt{lock()}$
				\EndIf \label{lin:threadv12}
				\State return $\langle \preds, \currs\rangle$;
	\cmnt{

	        \State $L\_rec.getAllPreds\&Currs(L\_rec_{i}$ $\downarrow$, $G\_pred$ $\uparrow$, $G\_curr$ $\uparrow$) \label{lin:threadv2};	
		\If{$(G\_pred.knext \neq G\_curr)$} \label{lin:threadv3}
		\State    /*Find $k$ $>$ $i$; such that $L\_rec_{k}$ contains next update method on same bucket */
		\If{$((($$L\_rec_{k}$.\textup{L\_opn}$)$ $=$ INSERT$)$ $||$ $(($$L\_rec_{k}$.\textup{L\_opn}$)$ $=$ DELETE$))$} \label{lin:threadv4}
		
	\State $L\_rec_{k}$.$G\_pred$.$\texttt{unlock()}$ \label{lin:thredv4-5};
		\State $G\_pred$ $\gets$ ($L\_rec_{i}.G\_pred.knext)$ \label{lin:threadv5};
	\State $L\_rec_{k}$.$G\_pred$.$\texttt{lock()}$ \label{lin:thredv5-5};		
		
		\EndIf \label{lin:threadv8}
		
		\EndIf \label{lin:threadv9}
	\State return $\langle G\_pred, G\_curr\rangle$;
	}
	\EndProcedure	
	\end{algorithmic}
\end{algorithm}




\section{Garbage Collection}
\label{sec:gc}
We have performed garbage collection method to delete the unwanted version of keys i.e. if the particular version corresponding to any key is not going to use in future then we can delete that version. For the better understanding of it please consider \figref{mvostm4}. Here, we are having 3 versions of key $k_1$ with timestamp 0, 15 and 25 respectively. Each version is having 5 fields described in \secref{mvdesign}. Now, consider the version 15, there exist the next version 25 and all the transactions between 15 to 25 has been terminated (either commit or abort) then we are deleting version 15. Similarly, we can delete other versions corresponding to each key as well and optimize the memory. 

\begin{figure} [H]
	\centerline{\scalebox{0.35}{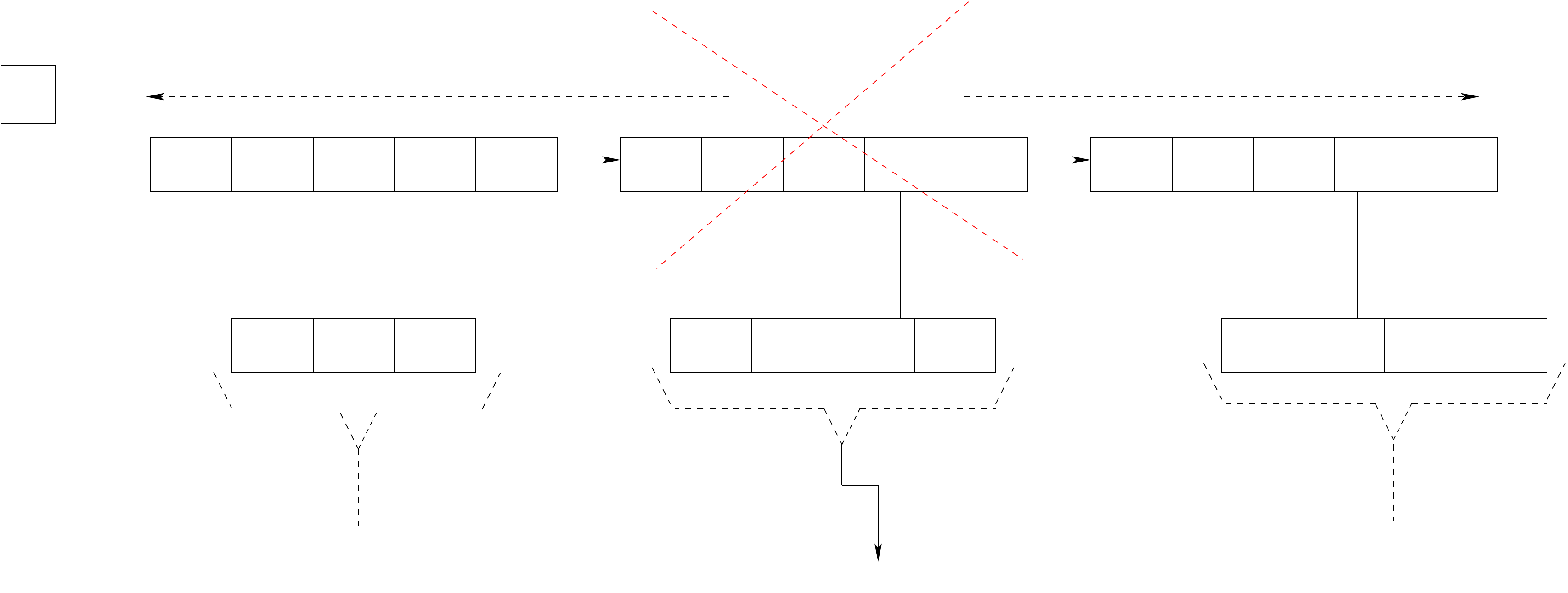}}
	\caption{Data Structures for Garbage Collection}
	\label{fig:mvostm4}
\end{figure}

\begin{algorithm}[H] 
\label{alg:begin1} 
\caption{STM $begin()$: Invoked by a thread to being a new transaction $T_i$}
	\setlength{\multicolsep}{0pt}
\begin{algorithmic}[1]
\makeatletter\setcounter{ALG@line}{423}\makeatother
\Procedure{STM begin}{$L\_t\_id \uparrow$}
\State /*Creating a local log for each transaction*/
\State \txll $\gets$ create new \txllf;  
\State /*Get $t\_id$ from $\cnt$*/
\State \txll.$L\_t\_id$ $\gets$ $\cnt$;  
\State /*Incremented $\cnt$*/
\State $\cnt$ $\gets$ \gi;  
\State $\livel.lock()$; \label{lin:beginit}
\State add $L\_t\_id$ to $\livel$; 
\State $\livel.unlock()$;
\State return $\langle L\_t\_id \rangle$; 
\EndProcedure
\end{algorithmic}
\end{algorithm}

\begin{algorithm}[H]  
\label{alg:instuple}
\caption{$\instup()$: Inserts the version tuple for $(L\_t\_id,v)$ created by the transaction $T_i$ into the version list of $L\_key$}
	\setlength{\multicolsep}{0pt}
\begin{algorithmic}[1]
\makeatletter\setcounter{ALG@line}{435}\makeatother
\Procedure{ins\_tuple}{$L\_key \downarrow, L\_t\_id\downarrow, v\downarrow, NULL\downarrow, NULL\downarrow$}
\State /*Initialize $cur\_tuple$*/
\State $cur\_tuple = \langle L\_t\_id, val, F, NULL, NULL \rangle$; 
\State /* Finds the tuple with the largest timestamp smaller than i */
\State $\find(L\_t\_id \downarrow, L\_key \downarrow, prev\_tuple \uparrow)$; \label{lin:prevtup}
\State /*$prev\_tuple$ is $\langle ts, val, mark, rvl, nts\rangle$*/
\State $cur\_tuple.nts = prev\_tuple.nts$;
\State $prev\_tuple.nts = L\_t\_id$;
\State insert $cur\_tuple$ into $L\_key.vl$ in the increasing order of timestamps;  
\State /* $\lvert L\_key.vl \lvert$ denotes number of versions of $L\_key$ created and threshold is a predefined value. */
\If {($\lvert L\_key.vl \lvert > threshold$)} 
\State /*If number of created versions for $L\_key$ crossed the threshold value then calling the Garbage Collection*/ 
\State $gc(L\_key)$;
\EndIf 
\State return $\langle void \rangle$
\EndProcedure
\end{algorithmic}
\end{algorithm}

\begin{algorithm}[H]
\caption{STM $\gc()$: Unused version of a \tobj{} $L\_key$ will deleted from $L\_key.vl$}
	\setlength{\multicolsep}{0pt}
\begin{algorithmic}[1]
\makeatletter\setcounter{ALG@line}{451}\makeatother
\Procedure{gc}{$L\_key \downarrow$}
\State $\livel.lock()$;
\State /*\tobj{} $L\_key$ is already locked*/
\ForAll {$(cur\_tuple \in L\_key.vl)$} 
\If {$(cur\_tuple.nts == NULL)$}
\State /* If $nts$ is NULL, check the next tuple in the version list */
\State continue; 
\EndIf
\State $j = cur\_tuple.ts + 1$;
\State /*Check for all ids $j$ in the range $j < nts$*/
\While {$(j < cur\_tuple.nts)$} 
\If {$(j \in \livel)$}
\State /* If any tuples with timestamp $j$, such that  $i< j < nts$ have not terminated (means exist in $liveList$) then $cur\_tuple$ can't be deleted*/

\State break;
\EndIf
\EndWhile
\State /* If all the tuples with timestamp $j$, such that  $i< j < nts$ have terminated then $cur\_tuple$ can be deleted*/
\State delete $cur\_tuple$;
\EndFor 
\State /* $\livel$ is not unlocked when this function returns */
\State return $\langle void \rangle$
\EndProcedure
\end{algorithmic}
\end{algorithm}

\end{document}